\documentclass[sigconf, nonacm]{acmart}

\newcommand{\nop}[1]{}

\usepackage[linesnumbered,ruled,vlined]{algorithm2e}
\usepackage{algpseudocode}
\usepackage{color}
\usepackage{amsmath}
\usepackage{multirow}
\usepackage{subfigure}
\usepackage{soul}
\usepackage{balance} 
\usepackage[normalem]{ulem}
\usepackage{setspace}
\usepackage{flushend}
\usepackage{bm}
\usepackage{mathbbol}
\usepackage{pifont}

\usepackage{tcolorbox}
\usepackage{xcolor}
\usepackage{subcaption}

\usepackage{geometry}
\usepackage{placeins}
\usepackage[dvipsnames]{xcolor}

\allowdisplaybreaks

\newtheorem{definition}{Definition}
\newtheorem{example}{Example}
\newtheorem{theorem}{Theorem}
\newtheorem{lemma}[theorem]{Lemma}

\newcommand{\eg}{\emph{e.g.,}\xspace}
\newcommand{\ie}{\emph{i.e.,}\xspace}
\newcommand{\wrt}{\emph{w.r.t.}\xspace}
\newcommand{\kw}[1]{{\ensuremath {\mathsf{#1}}}\xspace}

\newcommand\vldbdoi{XX.XX/XXX.XX}
\newcommand\vldbpages{XXX-XXX}
\newcommand\vldbvolume{14}
\newcommand\vldbissue{1}
\newcommand\vldbyear{2020}
\newcommand\vldbauthors{\authors}
\newcommand\vldbtitle{\shorttitle} 
\newcommand\vldbavailabilityurl{URL_TO_YOUR_ARTIFACTS}
\newcommand\vldbpagestyle{plain}

\begin{document}

\title{\textsc{LIVE}: Learnable Monotonic Vertex Embedding for Efficient Exact Subgraph Matching (Technical Report)}

\author{Yutong Ye}
\affiliation{%
  \institution{Beihang University}
  \city{Beijing, China}
}
\email{yutongye@buaa.edu.cn}

\author{Weilong Ren}
\affiliation{%
  \institution{ShenZhen Institute of Computing Sciences}
  \city{Shenzhen, China}
}
\email{renweilong@sics.ac.cn}

\author{Yang Liu}
\affiliation{%
  \institution{Beihang University}
  \city{Beijing, China}
}
\email{ly_act@buaa.edu.cn}

\author{Mengyi Yan}
\affiliation{%
  \institution{Shandong University}
  \city{Jinan, China}
}
\email{yanmy@sdu.edu.cn}

\author{Ruijie Wang}
\affiliation{%
  \institution{Beihang University}
  \city{Beijing, China}
}
\email{ruijiew@buaa.edu.cn}

\author{Li Sun}
\affiliation{%
  \institution{Beijing University of Posts and Telecommunications}
  \city{Beijing, China}
}
\email{lsun@bupt.edu.cn}

\author{Jianxin Li}
\affiliation{%
  \institution{Beihang University}
  \city{Beijing, China}
}
\email{lijx@buaa.edu.cn}

\author{Philip S. Yu}
\affiliation{%
  \institution{University of Illinois at Chicago}
  \city{Chicago, USA}
}
\email{psyu@uic.edu}

\begin{abstract}
Exact subgraph matching is a fundamental graph operator that supports many graph analytics tasks, yet it remains computationally challenging due to its NP-completeness. 
Recent learning-based approaches accelerate query processing via dominance-preserving vertex embeddings, but they suffer from expensive offline training, limited pruning effectiveness, and heavy reliance on complex index structures, all of which hinder the scalability to large graphs. 
In this paper, we propose \textit{\underline{L}earnable Monoton\underline{I}c \underline{V}ertex \underline{E}mbedding} (\textsc{LIVE}), a learning-based framework for efficient exact subgraph matching that scales to large graphs. \textsc{LIVE} enforces monotonicity among vertex embeddings by design, making dominance correctness an inherent structural property and enabling embedding learning to directly optimize vertex-level pruning power. 
To this end, we introduce a query cost model with a differentiable surrogate objective to guide efficient offline training. 
Moreover, we design a lightweight one-dimensional \textit{iLabel} index that preserves dominance relationships and supports efficient online query processing. 
Extensive experiments on both synthetic and real-world datasets demonstrate that \textsc{LIVE} significantly outperforms state-of-the-art methods in efficiency and pruning effectiveness.\looseness=-1
\end{abstract}

\maketitle

\pagestyle{\vldbpagestyle}
\begingroup\small\noindent\raggedright\textbf{PVLDB Reference Format:}\\
\vldbauthors. \vldbtitle. PVLDB, \vldbvolume(\vldbissue): \vldbpages, \vldbyear.\\
\href{https://doi.org/\vldbdoi}{doi:\vldbdoi}
\endgroup
\begingroup
\renewcommand\thefootnote{}\footnote{\noindent
This work is licensed under the Creative Commons BY-NC-ND 4.0 International License. Visit \url{https://creativecommons.org/licenses/by-nc-nd/4.0/} to view a copy of this license. For any use beyond those covered by this license, obtain permission by emailing \href{mailto:info@vldb.org}{info@vldb.org}. Copyright is held by the owner/author(s). Publication rights licensed to the VLDB Endowment. \\
\raggedright Proceedings of the VLDB Endowment, Vol. \vldbvolume, No. \vldbissue\ %
ISSN 2150-8097. \\
\href{https://doi.org/\vldbdoi}{doi:\vldbdoi} \\
}\addtocounter{footnote}{-1}\endgroup

\ifdefempty{\vldbavailabilityurl}{}{
\begingroup\small\noindent\raggedright\textbf{PVLDB Artifact Availability:}\\
The source code, data, and/or other artifacts have been made available at \url{https://github.com/JamesWhiteSnow/LIVE}.
\endgroup
}

\section{Introduction}
\label{sec:introduction}
Graph data management has become increasingly important in a wide range of real-world applications, 
\eg~social network analysis \cite{al2020topic}, knowledge graph discovery \cite{lian2011efficient}, and biological network mining \cite{szklarczyk2015string}.
Among various graph operators, exact subgraph matching is one of the most fundamental yet computationally challenging tasks.
Given a large data graph and a user-specified query graph, an exact subgraph matching query aims to retrieve all subgraphs that are isomorphic to the query graph, preserving both structural and label consistency. Such a query 
plays a central role in many graph data-driven analytical tasks, 
\eg~pattern recognition \cite{yan2008mining}, community search \cite{zhang2024top}, and graph-based query answering \cite{deutsch2022graph}.

Below, we give an example of a subgraph matching query for the discovery of research collaboration patterns in academic networks.

\begin{figure}[t]
    \centering
    \subfigure[academic network $G$]{
        \includegraphics[height=2.6cm]{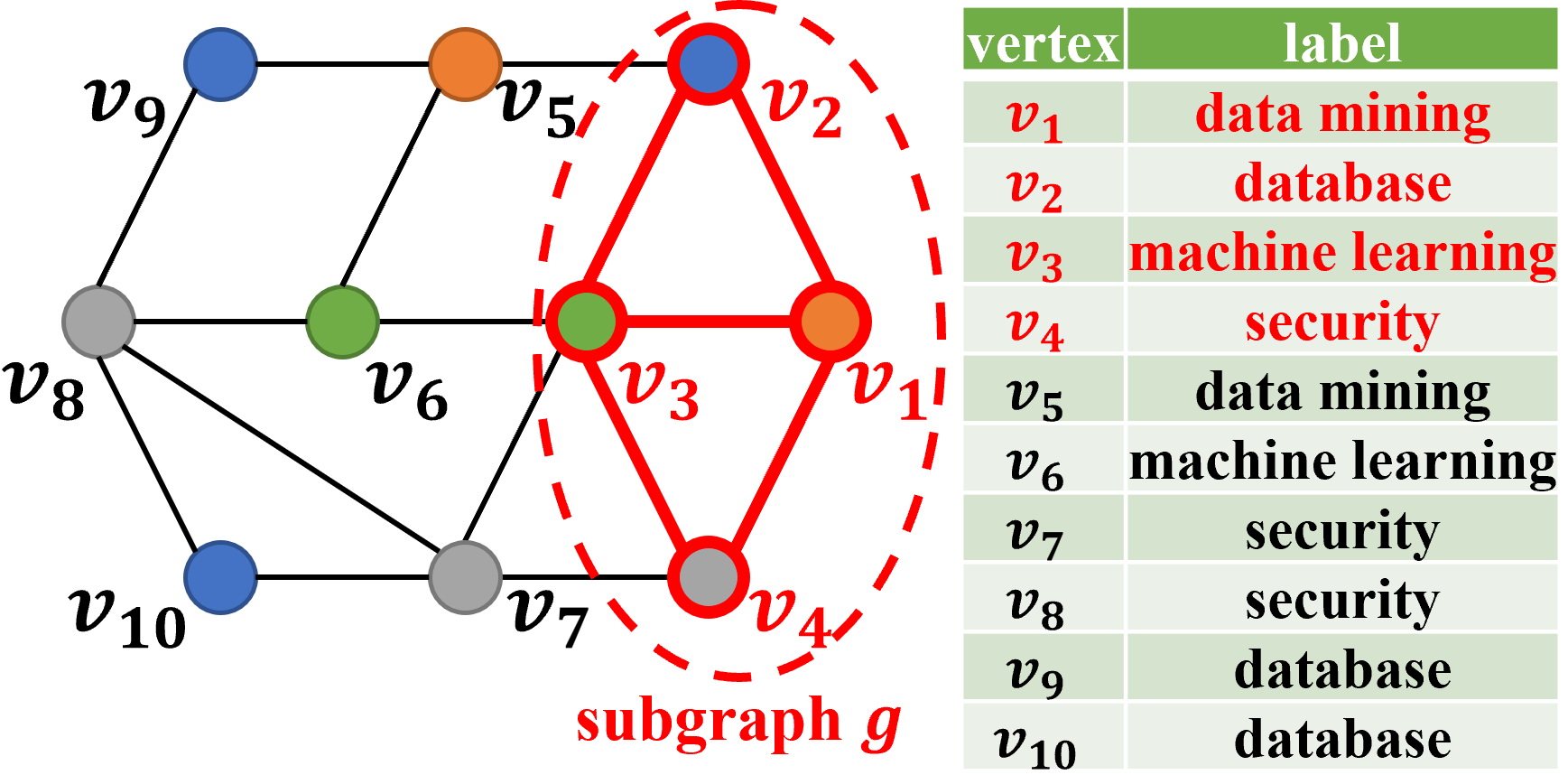}
        \label{subfig:data_graph}
    }
    \subfigure[query graph $q$]{
        \includegraphics[height=2.2cm]{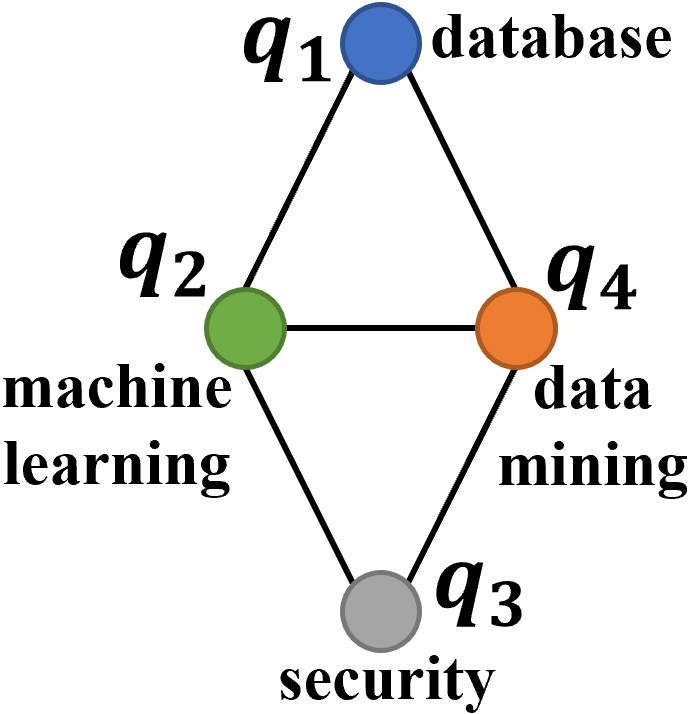}
        \label{subfig:query_graph}
    }
    \caption{Research collaboration pattern discovery.}
    \label{fig:subgraph_matching}
\end{figure}

\begin{example}
\textbf{(Research Collaboration Pattern Discovery in Academic Networks)}
As shown in Figure~\ref{subfig:data_graph}, an academic collaboration network can be modeled as an undirected, vertex-labeled graph, where each vertex represents a researcher labeled by their primary research field (\eg~\textit{database}, \textit{data mining}, \textit{machine learning}, and \textit{security}), and each edge represents a collaborative relationship among two researchers. Such networks naturally form large-scale graphs capturing scholarly collaborations.

A common analytical task is to identify recurring collaboration structures that match a particular team formation pattern. 
For example, a conference organizer may wish to find research teams exhibiting a specific interdisciplinary pattern, \eg~``a core researcher in data mining collaborating with colleagues in database systems, machine learning, security''. 
This pattern can be expressed as a query graph $q$ (Figure~\ref{subfig:query_graph}), where vertices denote researchers with specified research fields and edges denote collaborations.

Given a large academic network $G$ (Figure~\ref{subfig:data_graph}), exact subgraph matching 
retrieves all subgraphs $g \subseteq G$ 
isomorphic to the query graph $q$. 
As illustrated, one valid match is 
$g=\{v_1, v_2, v_3, v_4\}$, corresponding to researchers in data mining, database, machine learning, and security, with all collaborations 
preserved.
\end{example}

Despite its importance, exact subgraph matching remains highly challenging.
The problem has been proven to be NP-complete \cite{cordella2004sub,lewis1983michael}, and the search space grows exponentially with graph size and structural complexity,
\eg~an exact subgraph matching query on a graph with only $3K$ vertices and $10K$ edges can yield over 4 billion answers \cite{sun2020memory}, rendering such queries very costly on large graphs.

\noindent{\bf Prior Arts}.
Existing subgraph matching methods can be broadly classified into the following two categories.

\noindent{\textbf{(1) Structural-based approaches~\cite{he2008graphs,shang2008taming,bonnici2013subgraph,bi2016efficient,juttner2018vf2++,han2019efficient,bhattarai2019ceci,sun2020memory,li2025subgraph}}}.
These methods rely on explicit structural comparisons between the query graph and candidate subgraphs, which often incur high computational overhead and limit query efficiency.

\noindent{\textbf{(2) Learning-based approaches~\cite{bai2019simgnn,li2019graph,lou2020neural,wang2022reinforcement,ye2024efficient,yang2025neuso}}}.
Recent works in this line leverage Graph Neural Networks (GNNs) to embed graphs into embedding spaces, transforming subgraph matching into efficient vector-based operations.
Early studies primarily targeted approximate matching or graph similarity search~\cite{bai2019simgnn,li2019graph,lou2020neural}, where embeddings are used to estimate structural similarity rather than to guarantee exact matches.
More recent efforts have extended learning-based techniques to exact subgraph matching, either by learning matching orders~\cite{wang2022reinforcement,yang2025neuso} or by preserving dominance relationships in the embedding space~\cite{ye2024efficient}.

Among them, dominance-based methods (\eg~GNN-PE~\cite{ye2024efficient}) learn vertex embeddings from labels and 1-hop neighborhoods, such that subgraph containment can be verified via coordinate-wise dominance.
This property enables candidate filtering without false dismissals and yields substantial speedups over traditional structure-based methods (\eg~up to 1–2 orders of magnitude~\cite{ye2024efficient}).

Despite these advances, dominance-learning–based approaches still suffer from two fundamental limitations:
(i) expensive offline training that requires enumerating 1-hop subgraphs and all possible 1-hop substructures for each vertex, and
(ii) limited pruning effectiveness during online queries, as embedding learning is tightly constrained by dominance correctness.
To compensate for the latter, existing methods often rely on complex index structures~\cite{ye2024efficient}, which incur substantial storage and traversal overhead.

Consequently, while learning-based approaches make exact subgraph matching feasible on small to medium-sized graphs, they remain inadequate for large-scale graphs.
For instance, these meth-\\ods often fail to complete offline preprocessing on graphs with millions of vertices and edges within practical time budgets~\cite{ye2024efficient}.
This leads to a natural question:
\textit{Is it possible to design a learning-based approach that further accelerates exact subgraph matching with a lightweight offline training phase, while enabling efficient online query processing without relying on heavy index structures?}

\noindent{\bf Our Contributions}.
In this paper, we answer the question above affirmatively by proposing \textit{\underline{L}earnable Monoton\underline{I}c \underline{V}ertex \underline{E}mbedding} (\textsc{LIVE}), a novel framework for efficient exact subgraph matching on large graphs. 
Unlike prior learning-based 
approaches, 
\textsc{LIVE} rethinks the coupling between no-false-dismissal guarantees and pruning power for vertex embedding learning in exact subgraph matching. 

Specifically, \textsc{LIVE} introduces a design that enforces monotonicity among vertex embeddings by construction. 
As a result, the no-false-dismissal guarantee of candidate filtering becomes an inherent structural property of the embeddings, rather than an outcome learned through enumerating subgraph--substructure pairs.
This decoupling allows embedding learning to move beyond mere correctness preservation and directly optimize embedding quality, 
\ie~vertex-level pruning power. 
Accordingly, we develop a cost model that estimates the query cost induced by vertex embeddings, and
derive a continuously differentiable surrogate of this model,
enabling effective optimization of vertex embeddings during offline training and 
leading to more efficient online query processing.

Building on the learned monotonic vertex embeddings, 
we further design a lightweight \textit{iLabel} index that maps multi-dimensional embeddings into a one-dimensional key space while preserving dominance relationships. 
This design transforms dominance-region search for query vertices to efficient range queries supported by a B$^+$-tree, thereby eliminating the need for complex indexing structures. 
In addition, we incorporate multiple pruning strategies during index traversal to further reduce the candidate search space.

\begin{table}[t]\footnotesize
\begin{center}
\caption{Symbols and Descriptions}
\label{tab:notations}
\begin{tabular}{|l||p{4.5cm}|}
\hline
\textbf{Symbol}&\textbf{Description} \\
\hline\hline
    $G$ & a data graph\\\hline
    $q$ & a query graph\\\hline
    $g$ & a subgraph of the data graph $G$\\\hline
    $v_i$ (or $q_i$) & a vertex in graph $G$ (or $q$)\\\hline
    $e_{ij}$ (or $e_{q_iq_j}$) & an edge in graph $G$ (or $q$)\\\hline
    $V(G)$ (or $V(q)$) & a set of vertices $v_i$ (or $q_i$) in graph $G$ (or $q$) \\\hline
    $L(G)$ (or $L(q)$) & a labeling function of graph $G$ (or $q$) \\\hline
    $o(v_i)$ (or $o(q_i)$) & a vertex embedding vector of $v_i$ (or $q_i$) \\\hline
    $g_1(v_i)$ (or $s_1(v_i)$) & a 1-hop subgraph (or substructure) of $v_i$ \\\hline
    $\mathcal{N}_1(v_i)$ (or $\mathcal{N}_1(q_i)$) & a set of $v_i$'s (or $q_i$'s) 1-hop neighbors \\\hline
\end{tabular}
\end{center}
\end{table}

In summary, we make the following contributions:
\begin{itemize}
    \item We propose \textsc{LIVE} (Section~\ref{sec:problem_definition}), a learning-based framework for efficient exact subgraph matching scaling to large graphs.

    \item We develop a learnable monotonic vertex embedding method (Section~\ref{sec:embedding}) that enforces dominance relationships by design, allowing embedding learning to focus solely on optimizing vertex-level pruning power. 
    We further introduce a query cost model and derive a continuous, differentiable surrogate objective to guide offline training.

    \item We design a lightweight one-dimensional \textit{iLabel} index (Section~\ref{sec:index}) that preserves dominance relationships among multi-dimensional vertex embeddings, enabling efficient online query processing with low storage and traversal overhead.

    \item We propose an efficient online query processing algorithm (Section~\ref{sec:algorithm}) that integrates multiple pruning strategies supported by the \textit{iLabel} index.

    \item We conduct extensive experiments on both synthetic and real-world datasets (Section~\ref{sec:experiment}), confirming the efficiency and effectiveness of our \textsc{LIVE} for exact subgraph matching.
\end{itemize}

We review related work in Section~\ref{sec:related_work} and conclude in Section~\ref{sec:conclusion}.

\section{Problem Definition}
\label{sec:problem_definition}
This section formally defines the graph data model and the subgraph matching problem over a graph database. 
Table~\ref{tab:notations} summarizes the frequently used symbols and their descriptions.

\subsection{Preliminaries}
\label{subsec:preliminaries}
We start with basic notations.

\noindent{\bf Graph}.
A graph, $G$ is denoted by $G=(V(G),$ $E(G), \phi(G), L(G))$, where 
$V(G)$ is a set of vertices $v_i$, 
$E(G)$ is a set of edges $e_{ij}=(v_i, v_j)$ connecting vertex pairs, $\phi(G)$ is a mapping function from vertex pairs to edges 
(\ie~$\Phi(G)$:$V(G)\times V(G)\rightarrow E(G)$), and 
$L(G)$ is a vertex labeling function that assigns each vertex $v_i\in V(G)$ a label $L(v_i)$ describing its attribute or type.

In this paper, we consider an undirected labeled graph, which is widely used in applications such as social networks~\cite{wasserman1994social}, knowledge graphs~\cite{lian2011efficient}, and biological interaction networks~\cite{karlebach2008modelling}, where vertices represent entities and edges capture their relationships.

\noindent{\bf Graph Isomorphism~\cite{babai2018group,grohe2020graph}}.
Given two graphs $G_1=(V_1,E_1,\phi_1,L_1)$ and $G_2=(V_2,E_2,\phi_2,L_2)$, $G_1$ is said to be \textit{isomorphic} to $G_2$, denoted by $G_1 \equiv G_2$, 
if there exists a bijection mapping $f: V_1 \rightarrow V_2$ satisfying the following two conditions:
(i) \emph{label preservation}: for each vertex $v_i \in V_1$, $L_1(v_i)=L_2(f(v_i))$; and
(ii) \emph{edge preservation}: for any vertex pair $(v_i,v_j)\in E_1$, $(f(v_i),f(v_j))\in E_2$.\looseness=-1

Intuitively, $G_1$ and $G_2$ are isomorphic, if they have identical structural topology and vertex labels under the mapping $f$.

\noindent{\bf Subgraph Isomorphism}.
Let $G$ be a data graph and $q$ be a query graph. 
We say that $q$ is subgraph-isomorphic to $G$, 
denoted by $q \subseteq G$, if there exists a subgraph $g \subseteq G$ such that $g \equiv q$.

The subgraph isomorphism problem asks whether a query graph $q$ exactly matches a subgraph of $G$ while preserving structure and labels. 
The problem has been proven to be NP-complete~\cite{cordella2004sub,lewis1983michael}.

\subsection{Exact Subgraph Matching Query}
Based on the preliminaries in Section~\ref{subsec:preliminaries}, we formally define the subgraph matching query as follows:

\begin{definition} \textbf{(Exact Subgraph Matching Query)}
Given a large data graph $G$ and a query graph $q$, an exact subgraph matching query retrieves all subgraphs $g \subseteq G$ such that $g \equiv q$.
    \label{def:subquery}
\end{definition}

In other words, the goal is to enumerate all subgraphs of 
$G$ that are structurally identical to the query pattern $q$.
Exact subgraph matching is fundamental to many applications,
including pattern discovery in social networks \cite{qiao2017subgraph}, molecule search in bioinformatics \cite{alon2007network}, and entity-relation extraction in knowledge graphs \cite{sahu2017ubiquity}.

\subsection{Challenges}
Recent learning-based techniques~\cite{wang2022reinforcement,yang2025neuso} have explored exact subgraph matching via learned matching orders, while dominance-learning-based methods~\cite{ye2024efficient} further enable embedding-based candidate filtering without false dismissals. 
However, integrating learning into exact subgraph matching poses several challenges.

First, to guarantee correctness during candidate search, existing methods often require enumerating a massive number of subgraph--substructure pairs to learn dominance-preserving embeddings, resulting in substantial offline training costs and limited scalability. Second, prior embedding learning approaches are primarily driven by correctness constraints and do not explicitly optimize the pruning power of vertex embeddings, leading to suboptimal filtering effectiveness during query processing, especially on large graphs. Finally, even with dominance-preserving embeddings, efficiently retrieving all dominating vertices typically relies on complex indexing structures, incurring high storage and traversal overhead.

These challenges highlight the need for a new framework that preserves dominance correctness while scaling to large graphs in both offline vertex embedding training and online query processing.

\begin{algorithm}[t]
\caption{{\bf The \textsc{LIVE} Framework for Exact Subgraph Matching}}
\label{alg:framework}
{\footnotesize
\KwIn{
    a data graph $G$ and a query graph $q$\\
}
\KwOut{
    subgraphs $g$ ($\subseteq G$) that are isomorphic to $q$
}

\tcp{\bf Offline Phase}

\tcp{Generate vertex embeddings}

train a vertex embedding model $\text{Emb}(\cdot)$ that minimizes query processing cost over the data graph $G$\\

generate an embedding $o(v_i)$ for $\forall v_i \in G$ using the trained $\text{Emb}(\cdot)$\\

\tcp{Index construction}

build an \textit{iLabel} index $\mathcal{I}$ over the learned data 
vertex embeddings $o(v_i)$ 

\tcp{\bf Online Phase}
\For{each query vertex $q_i\in q$}{
    \tcp{Retrieve candidate matching vertices}

    generate the query vertex embedding $o(q_i)$ using the trained embedding model $\text{Emb}(\cdot)$

    obtain the candidate matching vertices $q_i.\textit{cand\_set}$ for $q_i$ by traversing the index $\mathcal{I}$
}

\tcp{Obtain and refine candidate subgraphs}
obtain a matching order $Q$ of query vertices $q_i \in q$ based on the sizes of the their 
candidate sets $|q_i.cand\_set|$\\

assemble and refine candidate subgraphs $g$ from candidate vertices in $q_i.cand\_set$\\
    
\Return subgraphs $g$ ($\equiv q$)
}
\end{algorithm}

\begin{figure*}[t]
    \begin{minipage}[t]{0.35\textwidth}
        \centering
        \subfigure[{\footnotesize $g_1(v_1)$}]{\label{subfig:one_hop_subgraph}
        {\includegraphics[height=2.2cm]{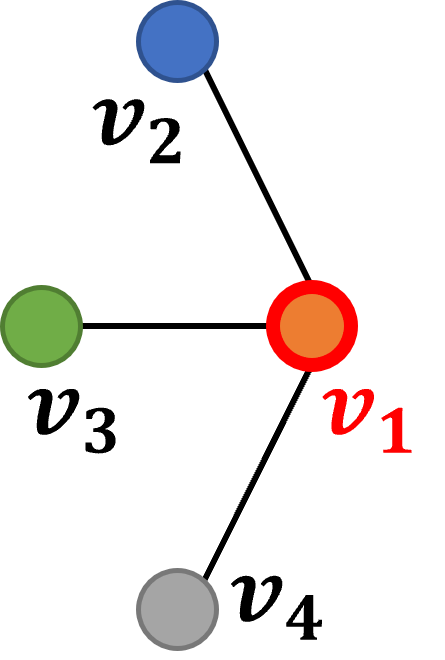}}} 
        \subfigure[{\footnotesize $s_1(v_1)\subseteq g_1(v_1)$}]{\label{subfig:substructure}
        {\includegraphics[height=2.5cm]{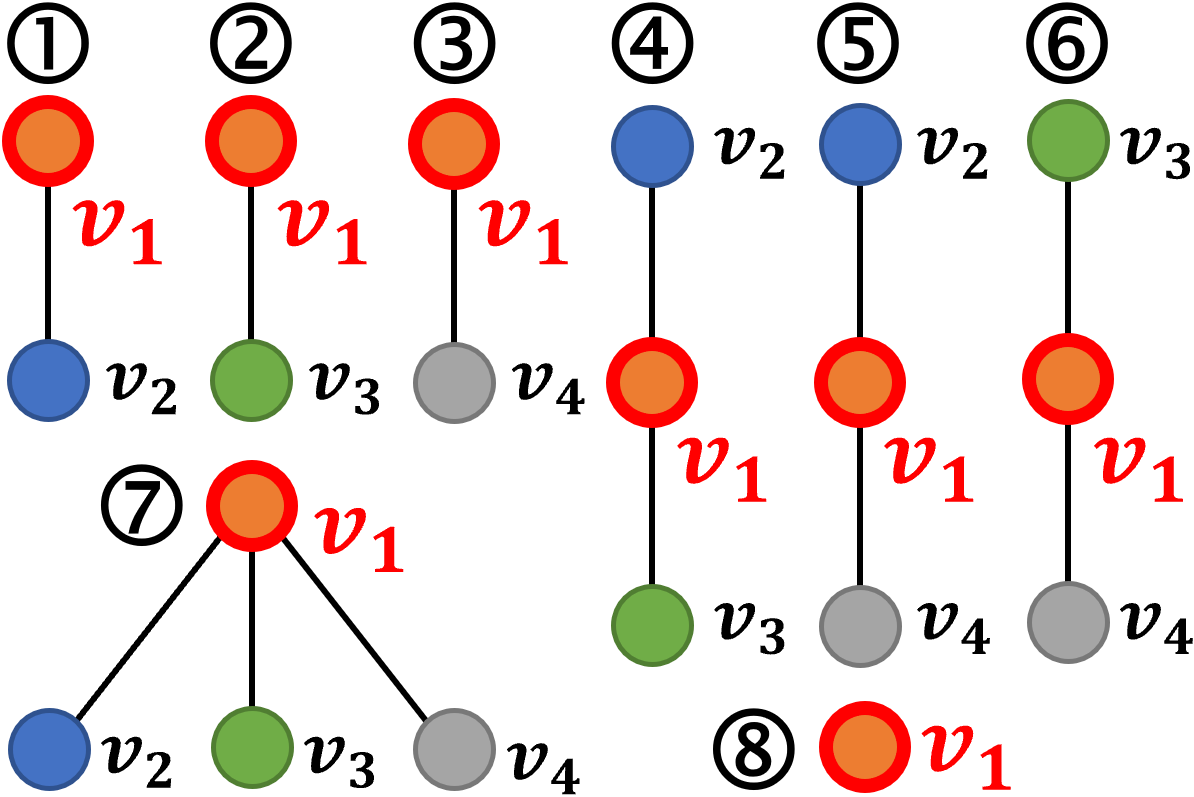}}} 
        \caption{$g_1(v_1)$ and all its possible $s_1(v_1)$.}
        \label{fig:subgraph_substructure}
    \end{minipage}\hfill
    \begin{minipage}[t]{0.65\textwidth}
        \centering
        \subfigure{
        {\includegraphics[height=2.5cm]{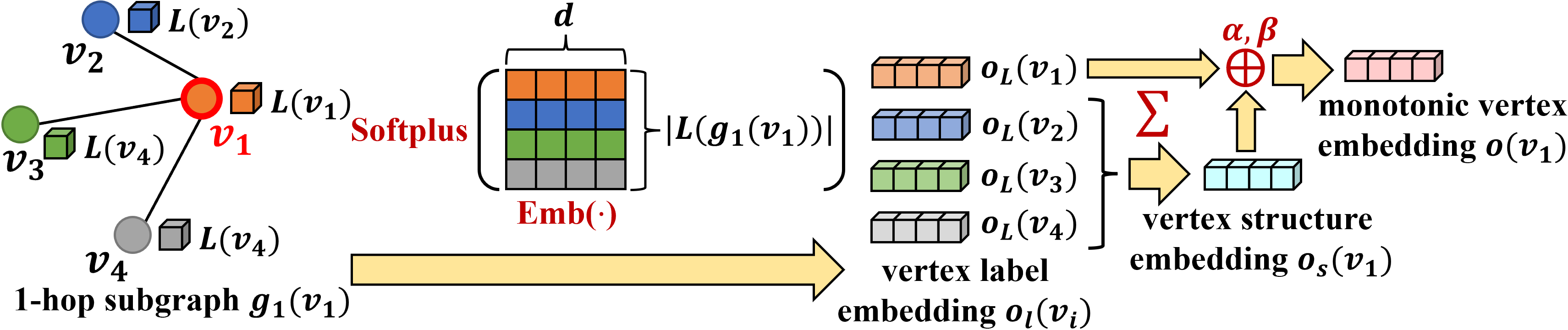}}} 
        \caption{Illustration of our monotonic vertex embedding $o(v_i)$.}
        \label{fig:vertex_embedding}
    \end{minipage}
\end{figure*}

\subsection{The \textsc{LIVE} Framework}
\label{subsec:framework}
Algorithm~\ref{alg:framework} illustrates a novel \underline{L}earn\underline{i}ng-based \underline{V}ertex \underline{E}mbedding (\textsc{LIVE}) framework for efficiently answering exact subgraph matching queries on large graphs using learned vertex embeddings. The framework consists of two phases: 
an offline precomputation phase (lines~1--3) and 
an online subgraph matching phase (lines~4--9).

\noindent{\bf Offline Phase}.
The offline phase learns vertex embeddings and builds an index to support efficient query processing. 
Specifically, \textsc{LIVE} trains an embedding model $\mathrm{Emb}(\cdot)$ on the data graph $G$ using a proposed cost model (see Section~\ref{subsec:embedding_optimizaition}) to 
minimize the expected query cost induced by vertex embeddings (line~1). 
After training, \textsc{LIVE} generates an embedding $o(v_i)$ for each data vertex $v_i \in V(G)$ using the learned model (line~2). 
Based on these embeddings, an \textit{iLabel} index $\mathcal{I}$ is constructed by mapping multi-dimensional vertex embeddings into a one-dimensional key space, enabling efficient range-based retrieval during query processing (line~3).

\noindent{\bf Online Phase}.
Given a query graph $q$, the online phase retrieves and refines candidate vertices to enumerate exact matches. 
For each query vertex $q_i \in V(q)$, \textsc{LIVE} first computes its embedding $o(q_i)$ using the trained model (lines~4--5). 
The \textit{iLabel} index $\mathcal{I}$ is then traversed to obtain a candidate set $q_i.\textit{cand\_set}$ of vertices in $G$ that 
may match $q_i$ (line~6). 
After candidate sets are generated for all query vertices, 
\textsc{LIVE} determines a matching order $Q$ of query vertices $q_i \in q$ based on their candidate set sizes, 
prioritizing query vertices with fewer candidates (line~7). 
Finally, an exact refinement procedure assembles and verifies candidate subgraphs by extending partial matches according to the matching order and enforcing subgraph isomorphism constraints (line~8). 
Finally, \textsc{LIVE} returns all subgraphs $g \subseteq G$ that are isomorphic to $q$ (line~9).

The design of \textsc{LIVE} is detailed in the following sections.
Sections~\ref{sec:embedding} and~\ref{sec:index} describe the offline phase,
including monotonic vertex embedding learning and \textit{iLabel} index construction, respectively, 
while Section~\ref{sec:algorithm} presents the online subgraph matching algorithm.

\noindent{\bf Note}.
The key advantage of \textsc{LIVE} lies in its principled separation of correctness guarantees from efficiency-oriented optimization. 
By enforcing dominance correctness as an intrinsic property of vertex embeddings through monotonicity, 
\textsc{LIVE} allows embedding learning to directly optimize vertex-level pruning power for improved query efficiency. 
Meanwhile, the dominance-preserving index reduces the cost of candidate retrieval while ensuring exactness, 
\ie~no false dismissals occur during index traversal.


\section{Learnable Monotonic Vertex Embedding}
\label{sec:embedding}
This section presents the design of monotonic vertex embeddings and the training of the embedding model $\text{Emb}(\cdot)$ in \textsc{LIVE} (lines~1--2 of Algorithm~\ref{alg:framework}). 
Specifically, Section~\ref{subsec:embedding_preliminary} introduces structural notions for characterizing vertex neighborhoods, Section~\ref{subsec:embedding_design} presents the monotonic embedding design and its dominance-preserving properties, and Section~\ref{subsec:embedding_optimizaition} describes cost-model-based embedding learning for achieving strong pruning power.

\subsection{Preliminaries}
\label{subsec:embedding_preliminary}

We introduce two structural notions used throughout the paper, which 
serve as conceptual units for characterizing local matching feasibility 
without requiring explicit enumeration during learning.

\noindent{\bf 1-hop Subgraph.}
Given a data graph $G$, the 1-hop subgraph centered at a vertex $v_i \in G$ is the induced subgraph containing $v_i$ and all its immediate neighbors: $g_1(v_i)=G[\{v_i\} \cup \mathcal{N}_1(v_i)]$,
where $\mathcal{N}_1(v_i)$ denotes the set of 1-hop neighbors of $v_i$ in $G$. The 1-hop subgraph captures both the label of the center vertex $v_i$ and the structural context of its local neighborhood.

In subgraph matching, $g_1(v_i)$ provides the necessary local context for determining whether $v_i \in G$ can serve as a candidate match for a query vertex $q_i \in q$. 
In particular, any valid match between $q_i$ and $v_i$ must satisfy local consistency constraints within their respective 1-hop neighborhoods.

\noindent{\bf 1-hop Substructure.}
Given the 1-hop subgraph $g_1(v_i)$ of a data vertex $v_i \in G$, 
a 1-hop substructure $s_1(v_i)$ is a subgraph induced by $v_i$ and a subset of its 1-hop neighbors $\mathcal{N}_1(v_i)$. 
For a query vertex $q_i \in q$ with 1-hop subgraph $g_1(q_i)$, 
if $q_i$ can be matched to $v_i$, then 
there exists a substructure $s_1(v_i) \subseteq g_1(v_i)$ such that 
$g_1(q_i) \equiv s_1(v_i)$.

This observation yields a necessary condition for vertex-level matching:  the local neighborhood $\mathcal{N}_1(q_i)$ of a query vertex $q_i \in q$ must be realizable as a substructure of the 1-hop subgraph of data vertex $v_i \in G$. 
Figure~\ref{subfig:one_hop_subgraph} illustrates an example 1-hop subgraph $g_1(v_1)$, and 
Figure~\ref{subfig:substructure} shows all $8\,(=2^3)$ possible 1-hop substructures contained in $g_1(v_1)$.

\begin{figure*}[t]
    \begin{minipage}[t]{0.56\textwidth}
        \centering
        \subfigure{
        {\includegraphics[height=3cm]{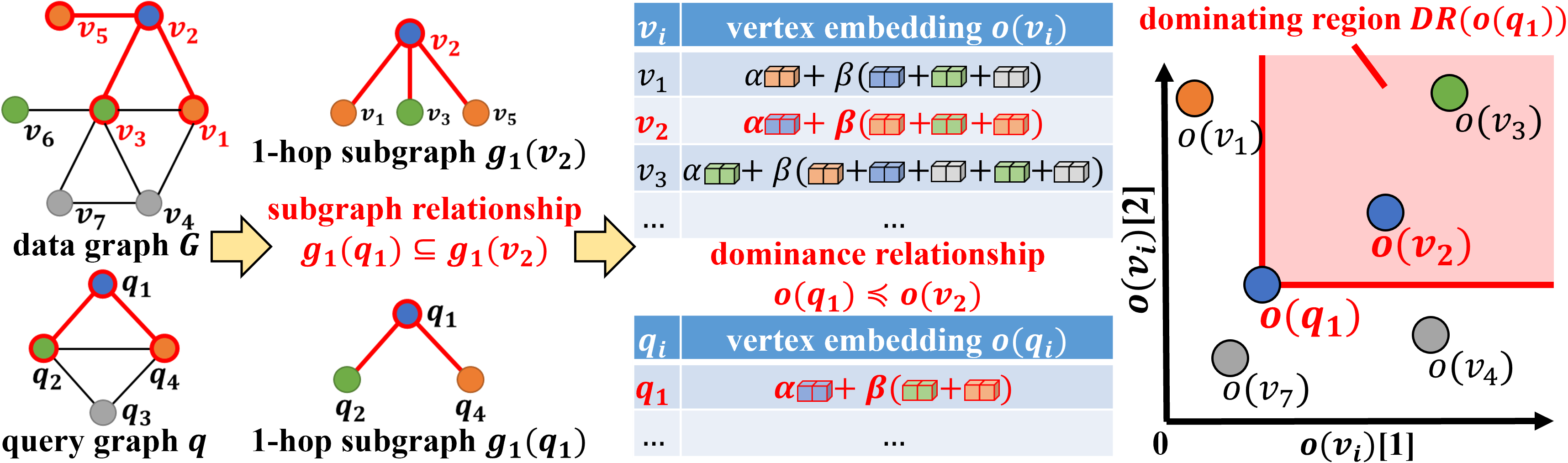}}}
        \caption{An illustration of the monotonicity of our vertex embedding.}
        \label{fig:dominance_example}
    \end{minipage}\hfill
    \begin{minipage}[t]{0.42\textwidth}
        \centering
        \subfigure[{\footnotesize initial VLE vectors $o_l(v_i)$}]{\label{subfig:vle_before}
        {\includegraphics[height=2.8cm]{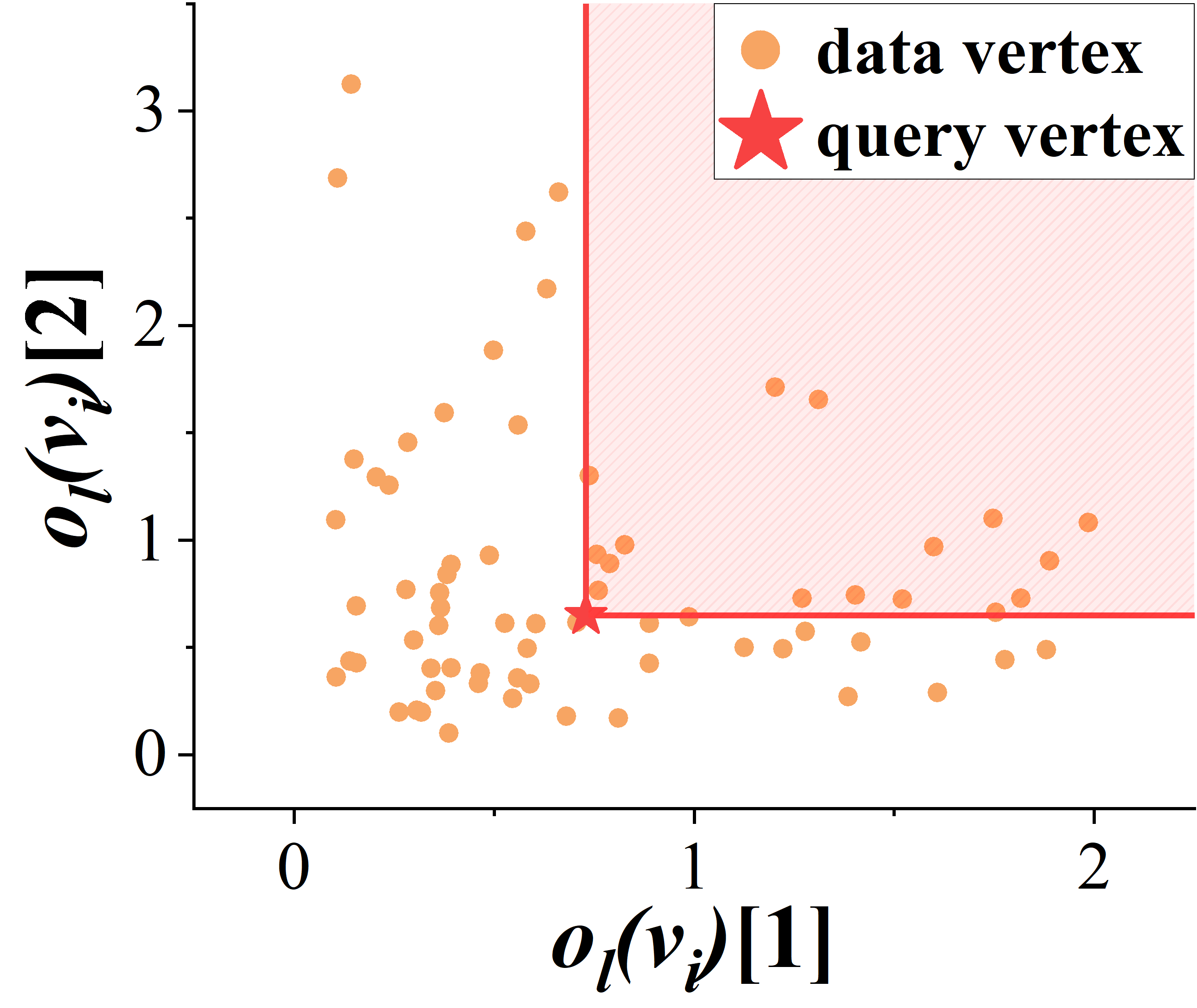}}} 
        \subfigure[{\footnotesize initial VSE vectors $o_s(v_i)$}]{\label{subfig:vse_before}
        {\includegraphics[height=2.8cm]{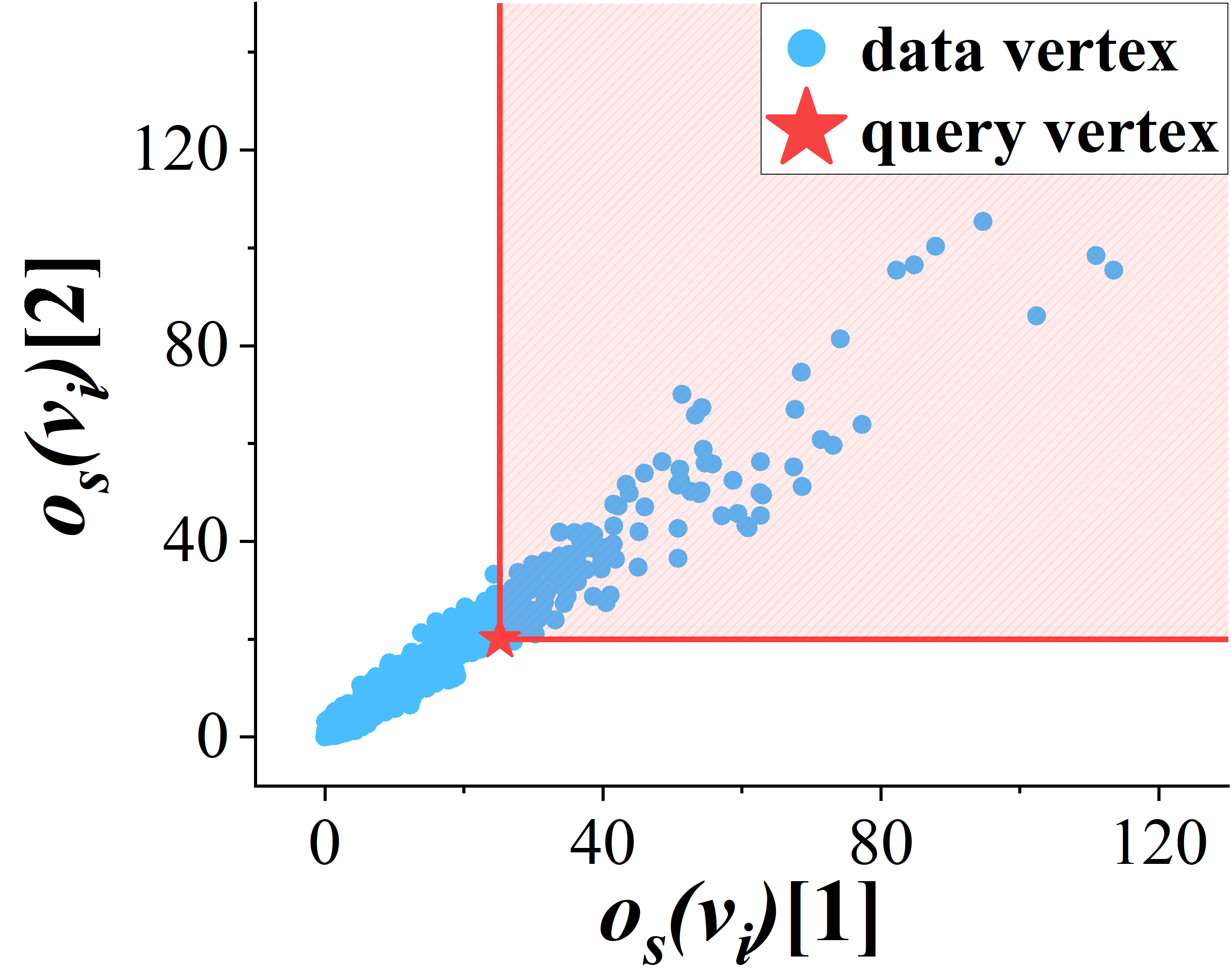}}}
        \caption{An example of initial VLE/VSE distributions.}
        \label{fig:visualization_before}
    \end{minipage}
\end{figure*}

\subsection{Monotonic Vertex Embedding Design}
\label{subsec:embedding_design}

In this subsection, we present a vertex embedding design that 
maps each vertex and its local neighborhood into an embedding space, 
enabling efficient and exact candidate filtering via vertex embeddings. 
Under this design, subgraph inclusion relationships are preserved by construction, rather than 
learned from enumerated subgraph--substructure pairs as in~\cite{ye2024efficient}.

\noindent{\bf Monotonic Vertex Embedding Function}.
We define a vertex embedding function
$f: g_1(v_i) \rightarrow o(v_i) \in \mathbb{R}^d$,
where $g_1(v_i)$ denotes the 1-hop subgraph centered at $v_i$, and $o(v_i)$ is the $d$-dimensional embedding of vertex $v_i$.

The embedding function $f(\cdot)$ is said to be \emph{monotonic} with respect to subgraph inclusion if
$s_1(v_i) \subseteq g_1(v_i)
    \;\Longrightarrow\;
    o(s_1(v_i)) \preceq o(g_1(v_i))$,
where $\preceq$ denotes \emph{coordinate-wise dominance} between two embeddings, 
\ie~$o(s_1(v_i))[j] \le o(g_1(v_i))[j]$ holds for the embedding values in all $d$ dimensions~\cite{ye2024efficient,ye2025continuous}. 
This monotonicity ensures that subgraph inclusion for each vertex in graph $G$ is faithfully reflected as dominance in the learned embedding space.

Below, we show how the embedding function $f(\cdot)$ is constructed.

\noindent{\bf Embedding Design.}
The embedding $o(v_i)$ of a vertex $v_i$ consists of two components: 
a vertex label embedding (VLE) that captures the semantic information of $v_i$, and 
a vertex structure embedding (VSE) that summarizes its local neighborhood structure.

\underline{\it Vertex Label Embedding (VLE).}
Let $o_l(v_i)$ denote the vertex label embedding of vertex $v_i$, defined as
\begin{equation}
    o_l(v_i) = \mathrm{softplus}(\mathrm{Emb}(L(v_i))),
    \label{eq:label_embedding}
\end{equation}
where $\mathrm{Emb}(\cdot)$ is a trainable embedding model (see Section~\ref{subsec:embedding_optimizaition} for its training) that maps the vertex label $L(v_i)$ to a $d$-dimensional vector, and 
$\mathrm{softplus}(x)=\log(1+e^x)$. 
The softplus activation ensures that all embedding dimensions are strictly non-negative, 
which is required to preserve dominance relationships among vertices.

As illustrated in Figure~\ref{fig:vertex_embedding}, 
vertex label embeddings $o_l(v_1)$--$o_l(v_4)$ are generated for vertices $v_1$--$v_4$ according to their respective labels $L(v_1)$--$L(v_4)$ using Eq.~\ref{eq:label_embedding}.

\underline{\it Vertex Structure Embedding (VSE).}
Let $o_s(v_i)$ denote the vertex structure embedding, which captures the local structural context of $v_i$ by aggregating the label embeddings of its 1-hop neighbors: 
\begin{equation}
    o_s(v_i) = \sum_{v_j \in \mathcal{N}_1(v_i)} o_l(v_j).
\end{equation}
This aggregation summarizes the neighborhood structure of $v_i$ in a manner consistent with its 1-hop subgraph $g_1(v_i)$. 
Since the aggregation is a summation over non-negative vectors, it is monotonic \wrt neighbor inclusion.
For example, in Figure~\ref{fig:vertex_embedding}, 
the structure embedding $o_s(v_i)$ of vertex $v_1$ is obtained by 
summing the label embeddings of its neighbors, \ie~$o_s(v_1)=o_l(v_2)+o_l(v_3)+o_l(v_4)$.

\underline{\it Monotonic Vertex Embedding.}
The final monotonic embedding of 
vertex $v_i$ is 
a weighted combination of its VLE and VSE components:\looseness=-1
\begin{equation}
    o(v_i) = \alpha\, o_l(v_i) + \beta\, o_s(v_i),
\end{equation}
where $\alpha$ and $\beta$ are non-negative weighting parameters 
(see Section~\ref{subsec:embedding_optimizaition} for their settings). 
This formulation jointly encodes vertex semantics and local structural information while preserving monotonicity.
As illustrated in Figure~\ref{fig:vertex_embedding}, 
the embedding of vertex $v_1$ is computed as $o(v_1)=\alpha\,o_l(v_1)+\beta\,o_s(v_1)$.

\underline{\it Monotonicity and Dominance Preservation.}
By construction, all components of $o_l(v_i)$ and $o_s(v_i)$ are non-negative, and both summation and weighted addition are monotonic operations.
Consequently, for any 1-hop subgraph $g_1(v_i)$ and its 1-hop substructure $s_1(v_i)$ with $s_1(v_i) \subseteq g_1(v_i)$, their embeddings satisfy $o(s_1(v_i))[j] \le o(g_1(v_i))[j]$ for all $j \in [1,d]$. 
That is, the embedding function $f(\cdot)$ preserves inclusion relationships among 1-hop subgraphs as dominance relationships in the embedding space.

\begin{lemma}
\textbf{(Monotonicity of Vertex Embeddings).}
Given a 1-hop subgraph $g_1(v_i)$ and one of its 1-hop substructures $s_1(v_i)$ with $s_1(v_i) \subseteq g_1(v_i)$, their embeddings satisfy $o(s_1(v_i)) \preceq o(g_1(v_i))$.
\label{lemma:vertex_monotonicity}
\end{lemma}

\begin{proof}
Both $g_1(v_i)$ and $s_1(v_i)$ share the same center vertex $v_i$ and thus have identical vertex label embeddings. 
Since $\mathcal{N}_1(s_1(v_i)) \subseteq \mathcal{N}_1(g_1(v_i))$, their structure embeddings satisfy
\[
o_s(s_1(v_i)) = \sum_{v_j \in \mathcal{N}_1(s_1(v_i))} o_l(v_j)
\preceq
\sum_{v_j \in \mathcal{N}_1(g_1(v_i))} o_l(v_j)
= o_s(g_1(v_i)).
\]
Combining this with the non-negative weighted sum defining $o(v_i)$ yields $o(s_1(v_i)) \preceq o(g_1(v_i))$, completing the proof.
\end{proof}

\noindent{\bf No-False-Dismissal Guarantee.}
The monotonicity property directly provides a no-false-dismissal guarantee for vertex-level candidate filtering. 
If a query vertex $q_i \in q$ can be matched to a data vertex $v_i \in G$,
then the 1-hop subgraph of $q_i$ must be a substructure of $g_1(v_i)$. 
By the monotonic embedding design, 
their embeddings satisfy $o(q_i) \preceq o(v_i)$.
Consequently, filtering candidates based on the dominance relationship among vertex embeddings preserves all valid matches and does not eliminate any true matches.

\begin{example}
Figure~\ref{fig:dominance_example} illustrates how 
monotonic vertex embedding enables safe and effective candidate filtering. 
In the data graph $G$, the 1-hop subgraph of vertex $v_2$ contains that of the query vertex $q_1$, \ie~$g_1(q_1) \subseteq g_1(v_2)$. 
By the monotonic embedding design, this inclusion relationship is preserved in the embedding space as a dominance relation, \ie~$o(q_1) \preceq o(v_2)$. 
Consequently, $v_2$ lies in the dominating region $DR(o(q_1))$, which consists of vertices whose embedding values are no smaller than $o(q_1)$ in all dimensions, and is correctly retained as a candidate for $q_1$ during index-based retrieval.

Vertices whose embeddings fall outside $DR(o(q_1))$ violate the dominance condition and can be safely pruned without affecting correctness. 
Note that some vertices (\eg~$v_3$) may satisfy the dominance condition yet fail to match $q_1$ due to higher-order structural constraints,
which will be eliminated later (Section~\ref{sec:algorithm}).
\end{example}

\begin{figure*}[t]
    \begin{minipage}[t]{0.6\textwidth}
    \centering
        \subfigure[{\footnotesize learned VLE vectors $o_l(v_i)$}]{\label{subfig:vle_after}
            {\includegraphics[height=2.8cm]{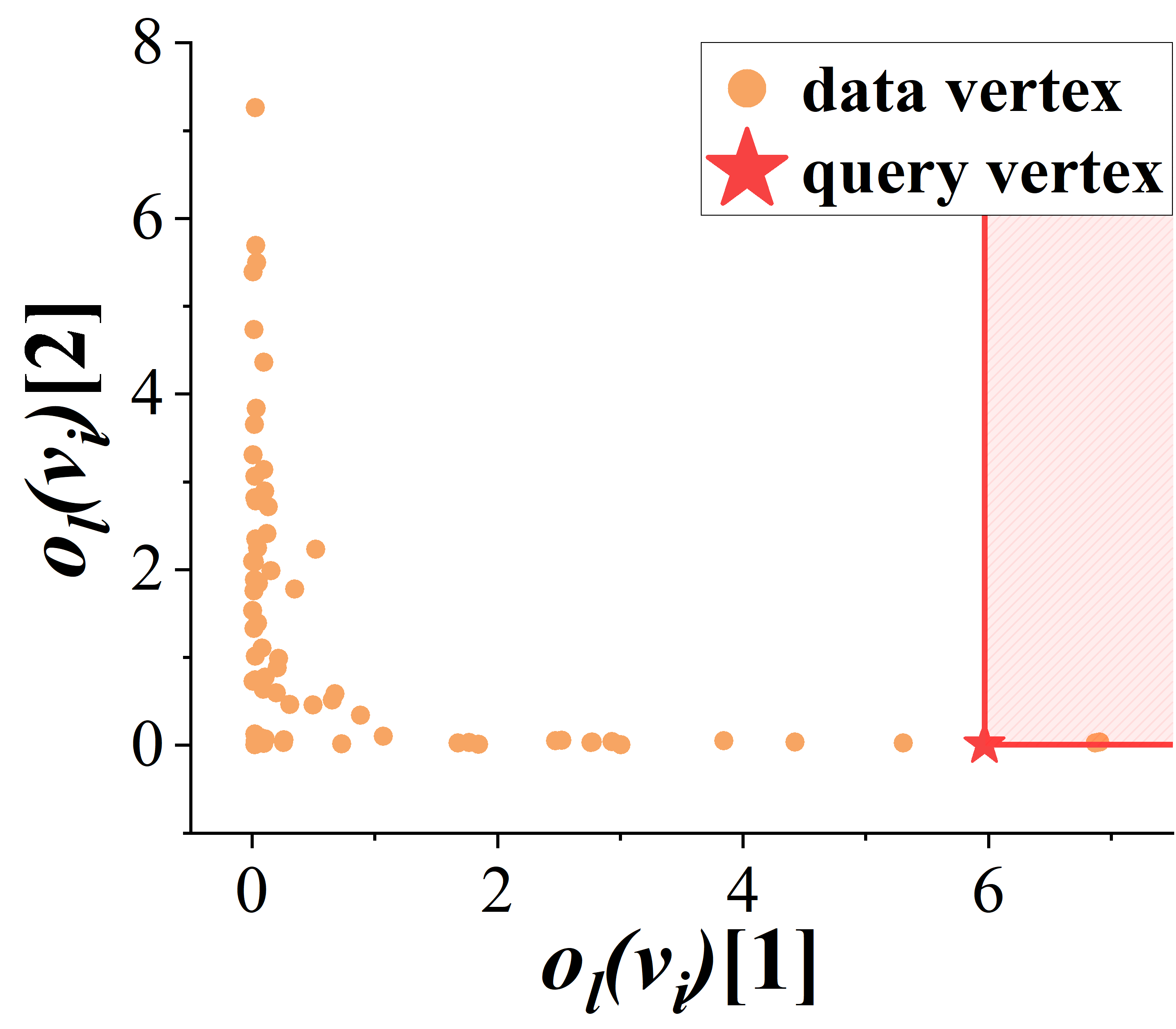}}} 
            \subfigure[{\footnotesize learned VSE vectors $o_s(v_i)$}]{\label{subfig:vse_after}
            {\includegraphics[height=2.8cm]{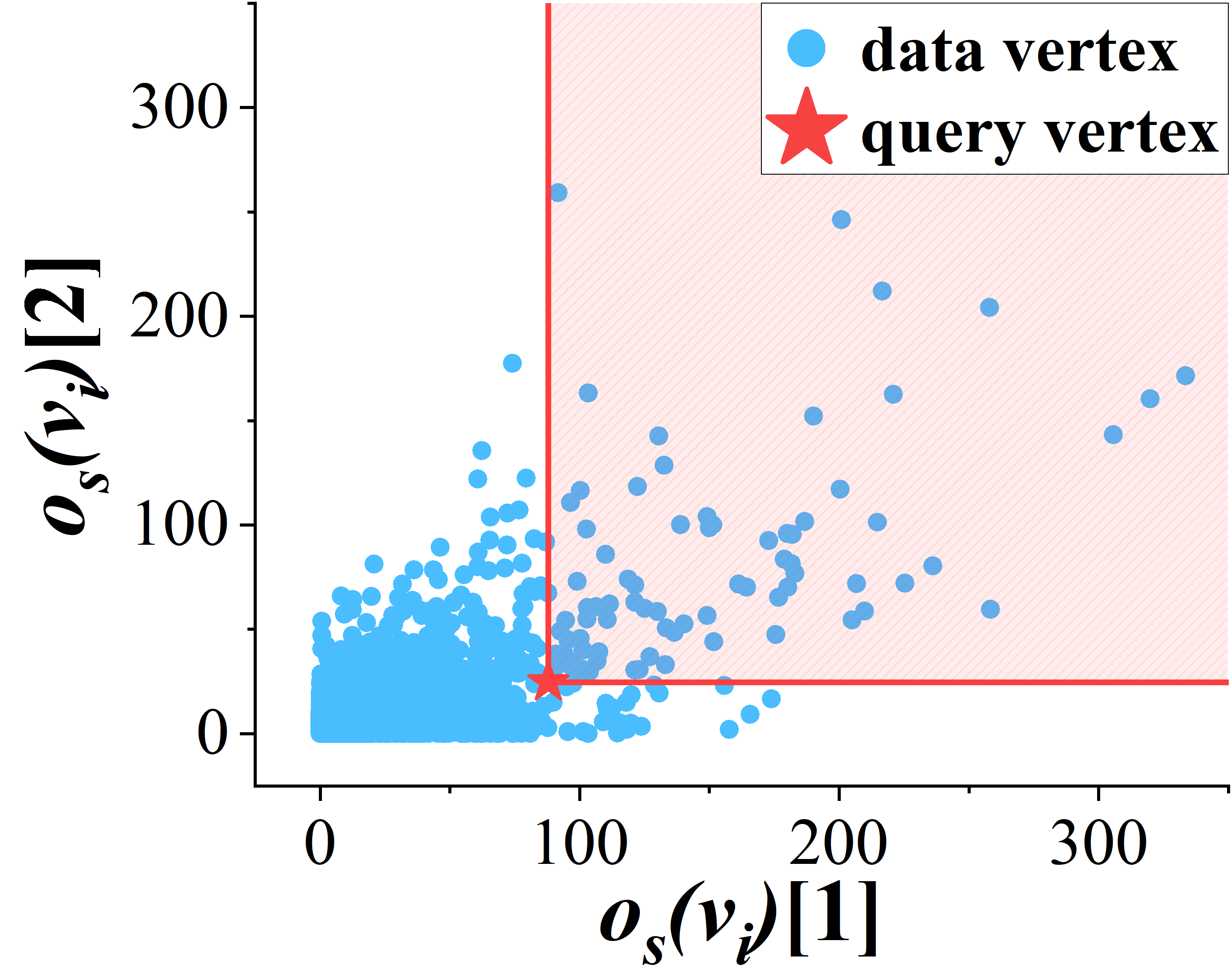}}} 
        \subfigure[{\footnotesize optimized VLE vectors $o_l(v_i)$}]{\label{subfig:vle_norm}
            {\includegraphics[height=2.8cm]{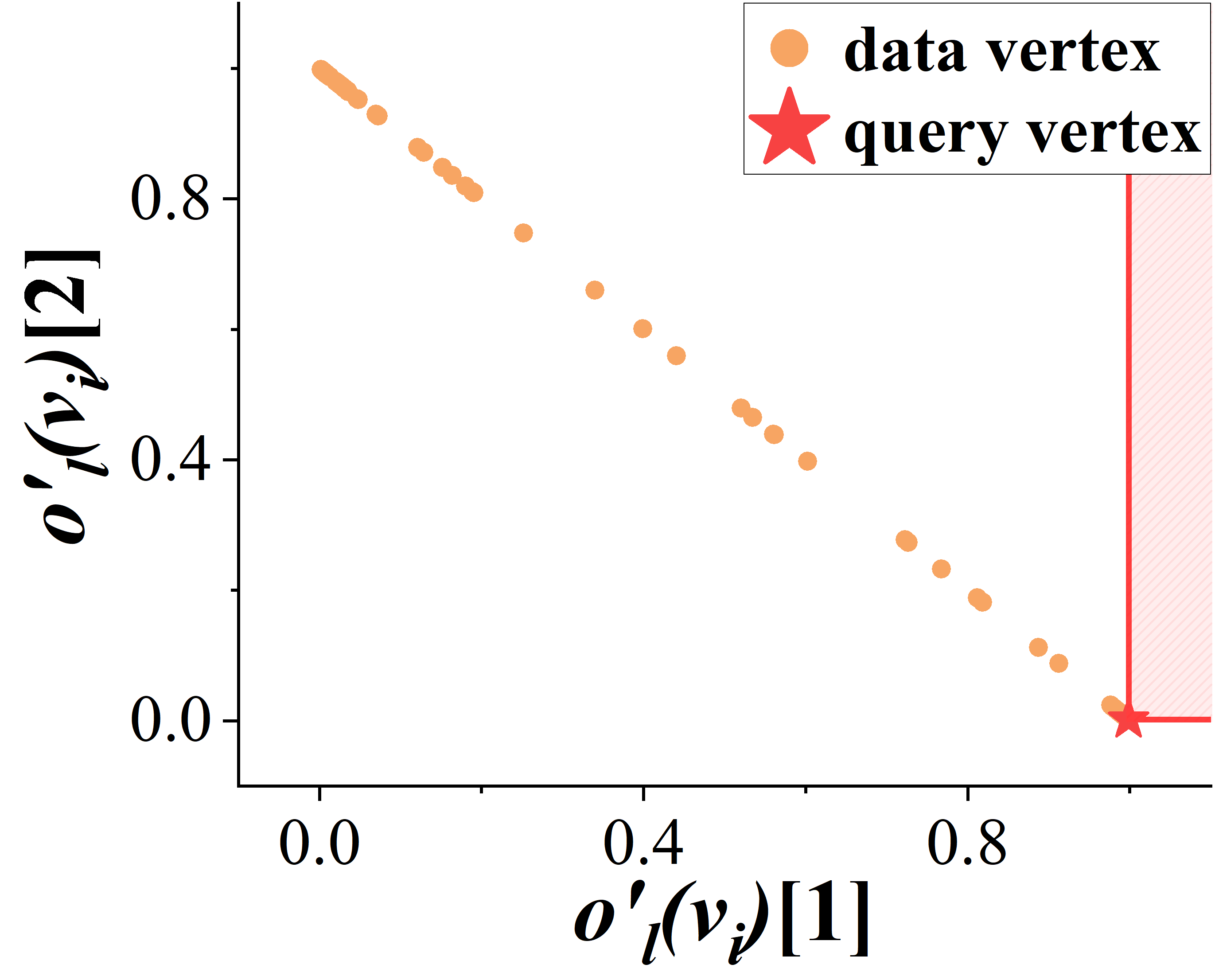}}} 
        \caption{An example of optimized VLE/VSE vector distributions.}
        \label{fig:visualization_after}
    \end{minipage}\hfill
    \begin{minipage}[t]{0.4\textwidth}
    \centering
        \subfigure[{\footnotesize $\alpha=10,\beta=1$}]{\label{subfig:mve_small}
            {\includegraphics[height=2.8cm]{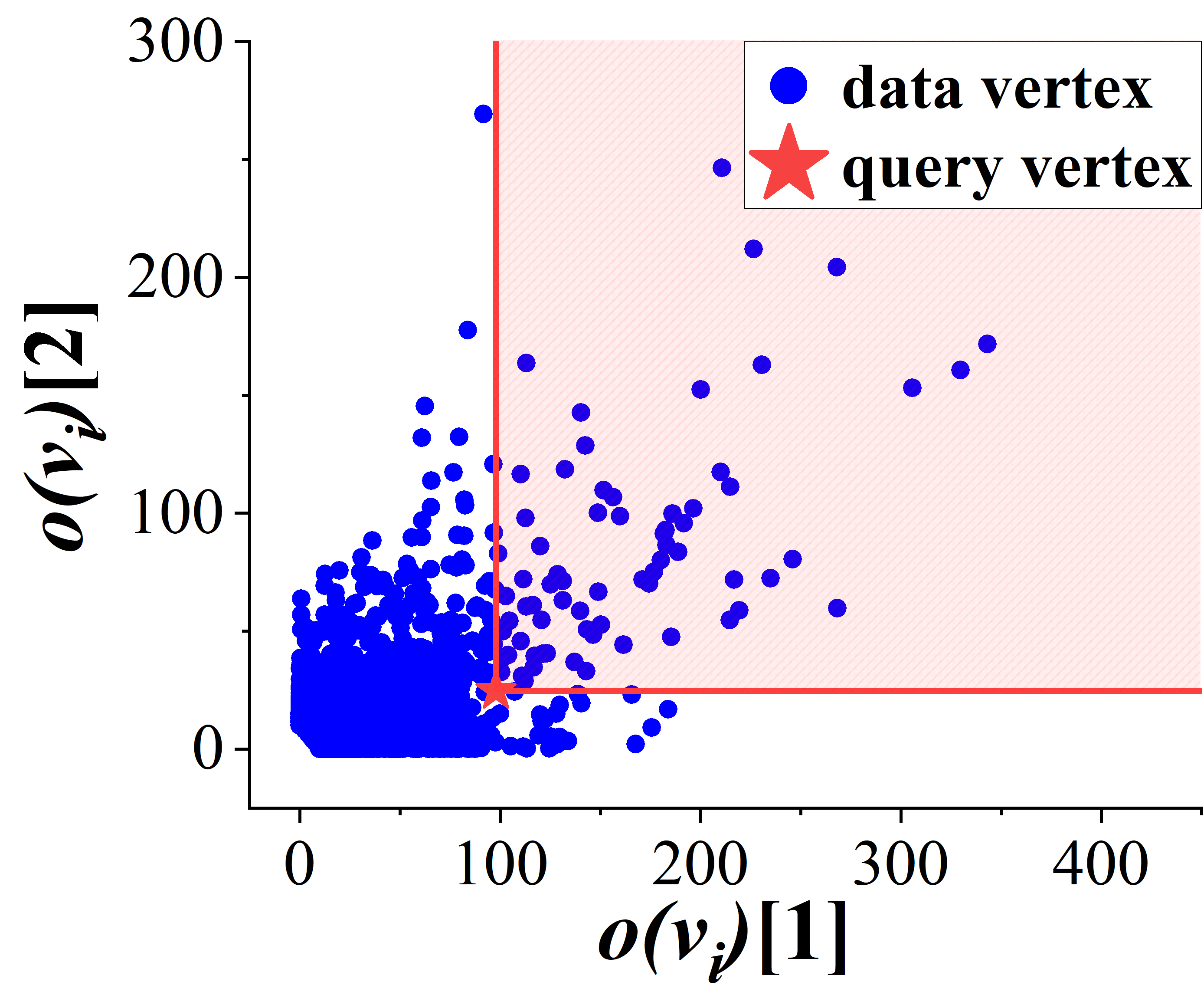}}} 
            \subfigure[{\footnotesize $\alpha=1,000,\beta=0.01$}]{\label{subfig:mve_large}
            {\includegraphics[height=2.8cm]{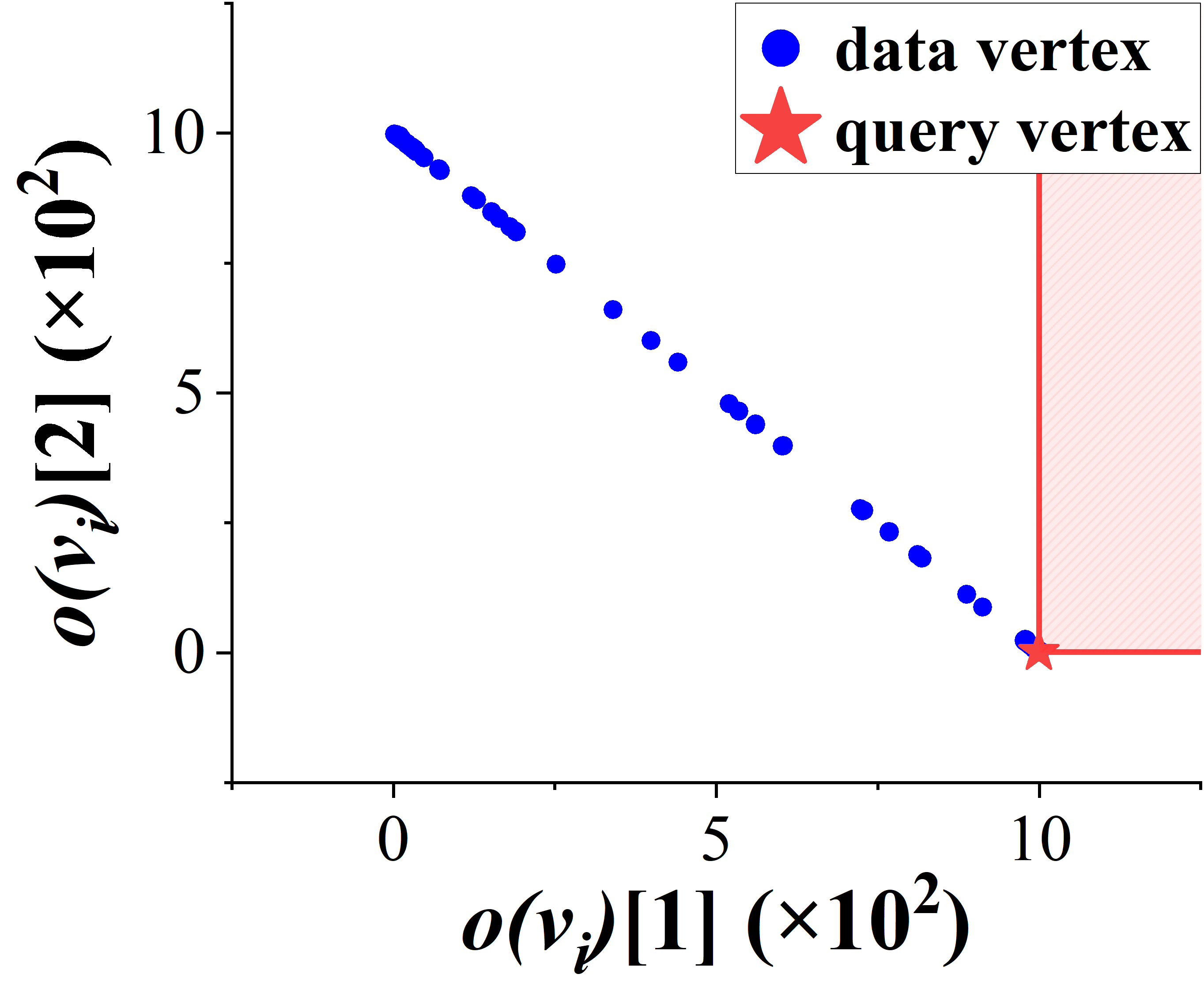}}} 
        \caption{An example of $o(v_i)$ with different $\alpha/\beta$.}
        \label{fig:visualization_final}
    \end{minipage}
\end{figure*}

\noindent{\bf Decoupling Correctness Guarantees from Embedding Optimization.}
Unlike existing dominance learning-based exact subgraph matching methods (\eg~\cite{ye2024efficient}) that tightly couple the no-false-dismissal guarantee with the embedding learning objective, LIVE preserves subgraph containment as dominance relations in the embedding space by design, rather than through learning.
As a result, the no-false-dismissal guarantee becomes an inherent structural property of the embedding, independent of the learning objective.

This decoupling enables a new optimization perspective.
Once correctness is ensured by construction, embedding learning is freed from dominance constraints and can directly target pruning power and query efficiency.
This insight motivates the cost-model–driven embedding optimization introduced in Section~\ref{subsec:embedding_optimizaition}, which explicitly minimizes query cost rather than learning dominance relations.

\subsection{Cost Model-based Embedding Optimization}
\label{subsec:embedding_optimizaition}
Given that the monotonic vertex embedding design in Section~\ref{subsec:embedding_design} guarantees dominance correctness by construction, in this subsection, we focus on improving embedding quality by maximizing pruning power to reduce query cost during candidate retrieval.

During online query processing, the cost of matching a query vertex $q_i$ is determined by the number of data vertices whose embeddings are dominated by $o(q_i)$, \ie~the size of the dominance-induced candidate set $\{v_j \mid o(q_i) \preceq o(v_j)\}$. 
Accordingly, embedding optimization aims to reduce the number of dominated vertices, thereby minimizing the size of candidates examined during query processing.
However, query embeddings are unknown during offline training, and enumerating all possible 1-hop substructures is infeasible due to their exponential growth with vertex degree ($\sum_{v \in V(G)}2^{deg(v)}$). 
To address this, we estimate query cost on data vertices and use it as a surrogate for expected query behavior.

\noindent{\bf Cost Model for Query Efficiency.}
Based on the above observation, we define a cost model that captures query efficiency by measuring the average number of vertices dominated by a data vertex embedding.
The expected query cost is approximated as
\begin{equation}
\mathrm{Cost} =
\frac{1}{|V(G)|}
\sum_{v_i \in V(G)}
\left|
\left\{
v_j \in V(G) \mid o(v_i) \preceq o(v_j)
\right\}
\right|,
\label{eq:cost_model}
\end{equation}
which represents the expected size of the candidate set induced by dominance-based filtering and serves as a proxy for query efficiency.

Figure~\ref{fig:visualization_before} visualizes the initial distributions of 2D vertex label embeddings (VLE) and vertex structure embeddings (VSE) on the Yeast dataset. 
The VLE vectors are randomly initialized and thus uniformly scattered in the embedding space (Figure~\ref{subfig:vle_before}), resulting in relatively small candidate sets. 
In contrast, VSE vectors, computed by summing non-negative VLE vectors, exhibit a strong concentration along the diagonal (Figure~\ref{subfig:vse_before}). 
This concentration causes many vertices to satisfy the dominance condition, leading to large candidate sets and high query cost.

\noindent{\bf Anti-Dominance Objective and Continuous Relaxation.}
As shown in Eq.~\ref{eq:cost_model}, the query cost induced by a query vertex embedding is determined by the number of data vertices whose embeddings are dominated by it. Accordingly, minimizing query cost is equivalent to minimizing the total number of dominance relations among vertex embeddings in the embedding space.

Intuitively, the notion of \emph{anti-dominance}, which indicates an object dominates or is dominated by only a few others~\cite{borzsony2001skyline}, naturally aligns with this optimization objective.
Following this intuition, we formalize Eq.~\ref{eq:cost_model} using the following \emph{anti-dominance cost}:

\begin{equation}
L_{\mathrm{cost}} =
\frac{1}{|V(G)|^2}
\sum_{v_i \in V(G)}
\sum_{v_j \in V(G)}
\mathbf{1}\{o(v_i) \preceq o(v_j)\},
\label{eq:discrete_cost}
\end{equation}
where $\mathbf{1}\{\cdot\}$ is an indicator function that equals $1$ if the embedding $o(v_i)$ dominates $o(v_j)$, and $0$ otherwise. 
This objective measures the expected number of dominance relations among all vertex pairs; minimizing it directly reduces the number of dominated vertices and thus improves vertex-level pruning power during query processing.

However, the anti-dominance cost in Eq.~\ref{eq:discrete_cost} is inherently discrete due to the indicator function induced by coordinate-wise dominance, and therefore cannot be optimized directly using gradient-based methods. To enable efficient optimization, we derive a continuous and differentiable relaxation of this objective. Specifically, we approximate the indicator function with a sigmoid-based surrogate and define the \emph{anti-dominance loss} as
\begin{equation}
\hspace{-2ex}
L =
\frac{1}{|V(G)|^2}
\sum_{v_i \in V(G)}
\sum_{v_j \in V(G)}
\sigma\!\left(
\frac{\min_k \left(o(v_i)[k] - o(v_j)[k]\right)}{\tau}
\right),
\label{eq:antidominance_loss}
\end{equation}
where $\sigma(x)=\frac{1}{1+e^{-x}}$ is the sigmoid function and $\tau$ is a temperature parameter that controls the smoothness of the approximation.
Note that the $\text{min}_k(\cdot)$ operator is piecewise linear and subdifferentiable.
Following~PointNet \cite{qi2017pointnet} and Deep Sets \cite{zaheer2017deep},
we use standard automatic differentiation \cite{paszke2019pytorch} to obtain a valid subgradient;
nondifferentiable ties occur with probability 0 for continuous parameters.

\begin{lemma}
\textbf{(Relationship Between Anti-Dominance Loss and Discrete Objective).}
Let $L_{\mathrm{cost}}$ and $L$ be defined in Eqs.~\ref{eq:discrete_cost} and~\ref{eq:antidominance_loss}, respectively.
As the temperature parameter $\tau \rightarrow 0^{+}$, the anti-dominance loss $L$ converges to the discrete anti-dominance cost $L_{\mathrm{cost}}$.
Moreover, for any finite $\tau > 0$, $L$ serves as a smooth upper bound of $L_{\mathrm{cost}}$.
\label{lemma:loss_relationship}
\end{lemma}

\begin{proof}
By definition, the dominance relation $o(v_i) \preceq o(v_j)$ holds if and only if $\min_k \big(o(v_i)[k] - o(v_j)[k]\big) \le 0$.
Accordingly, the discrete indicator function
$\mathbf{1}\{x \le 0\}$ can be approximated by the sigmoid function
$\sigma(x/\tau)$, which converges pointwise to the indicator as
$\tau \rightarrow 0^{+}$.

Therefore, for any vertex pair $(v_i, v_j)$, we have:
\[
\hspace{-3ex}
\lim_{\tau \rightarrow 0^{+}}
\sigma\!\left(
\frac{\min_k \big(o(v_i)[k] - o(v_j)[k]\big)}{\tau}
\right)
=
\mathbf{1}\!\left\{
\min_k \big(o(v_i)[k] - o(v_j)[k]\big) \le 0
\right\}.
\]

Substituting this limit into the definition of $L$ yields:
\[
\lim_{\tau \rightarrow 0^{+}} L
=
\frac{1}{|V(G)|^2}
\sum_{v_i, v_j}
\mathbf{1}\{o(v_i) \preceq o(v_j)\}
=
L_{\mathrm{cost}}.
\]

Moreover, for any finite $\tau > 0$, the sigmoid function satisfies
$\sigma(x/\tau) \ge \mathbf{1}\{x \le 0\}$ for all $x$.
Thus, each term in $L$ upper-bounds the corresponding indicator term in
$L_{\mathrm{cost}}$, which implies $L \ge L_{\mathrm{cost}}$.
Equality holds asymptotically as $\tau \rightarrow 0^{+}$, completing the proof.
\end{proof}

Although the anti-dominance loss in Eq.~\ref{eq:antidominance_loss} provides a differentiable surrogate for the discrete cost objective, directly optimizing it requires evaluating all vertex pairs, incurring $O(|V(G)|^2)$ time/space cost and rendering it infeasible for large graphs.
To address this, we adopt a sampling-based approximation during training: instead of summing over all vertex pairs, we uniformly sample vertex pairs $(v_i, v_j)$ from $V(G) \times V(G)$ and optimize the expected anti-dominance loss:
\begin{equation}
\hspace{-2ex}
L =
\mathbb{E}_{(v_i, v_j) \sim \mathcal{U}(V(G)\times V(G))}
\left[
\sigma\!\left(
\frac{\min_k \left(o(v_i)[k] - o(v_j)[k]\right)}{\tau}
\right)
\right]
\label{eq:sampled_loss}
\end{equation}
This sampling strategy preserves the expectation of the full loss in Eq.~\ref{eq:antidominance_loss} while reducing the computational cost to be linear in the number of sampled pairs. Consequently, the anti-dominance loss can be efficiently optimized using stochastic gradient-based methods.

\noindent{\bf Model Training.}
We train the vertex embedding model $\mathrm{Emb}(\cdot)$ by minimizing the sampled anti-dominance loss in Eq.~\ref{eq:sampled_loss} using mini-batch stochastic gradient descent. 
In each iteration, a mini-batch of vertex pairs is sampled to estimate gradients, while the temperature parameter $\tau$ is gradually annealed to tighten the approximation to the discrete cost objective. 
After convergence, the trained model $\mathrm{Emb}(\cdot)$ is used to generate embeddings for all data vertices in $G$.

Figures~\ref{subfig:vle_after} and~\ref{subfig:vse_after} show the distributions of VLE and VSE embeddings on the Yeast dataset after training. Compared to the initial distributions in Figure~\ref{fig:visualization_before}, both embeddings become substantially more dispersed, resulting in fewer dominated vertices. For example, the average query cost decreases from 1,182.6 (resp.~1,403.97) to 522.28 (resp.~1,027.27) for VLE (resp.~VSE).

\noindent{\bf Optimizations.}
We further apply the following optimizations during vertex embedding training.

\underline{\it Optimization via $L_1$ Normalization.}
Although optimizing the structure embeddings VSE effectively reduces query cost, the average cost of label embeddings VLE remains relatively high (522.28), still exceeding the theoretical lower bound $\frac{|V(G)|}{|L(G)|}=3{,}112/71=43.83$. 
This indicates room for further refinement. 
To enhance discriminability, we apply $L_1$ normalization to VLE embeddings:
\begin{equation}
o_l'(v_i) = \frac{o_l(v_i)}{\|o_l(v_i)\|_1},
\end{equation}
which constrains embeddings to a constant-$L_1$ manifold and spreads them along the anti-diagonal direction.

Empirically, as shown in Figure~\ref{subfig:vle_norm}, $L_1$ normalization yields a clear anti-diagonal alignment of VLE embeddings and significantly reduces the average query cost (\eg~from 522.28 to 284.6) on the Yeast dataset. 
The remaining gap to the theoretical lower bound $\frac{|V(G)|}{|L(G)|}$ is mainly due to label distribution skew.

\underline{\it Optimization of $\alpha$ and $\beta$.}
We refine the vertex embedding $o(v_i)=\alpha\,o_l(v_i)+\beta\,o_s(v_i)$ by tuning the weighting parameters $\alpha$ and $\beta$. 
Since the optimized VLE embeddings 
exhibit strong global separability, we treat them as the base distribution and use VSE as a fine-grained structural perturbation. 
Setting $\alpha \gg \beta$ (\eg~$\alpha = 1000; \beta = 0.01$) preserves the favorable VLE distribution while incorporating local structural cues. 
As illustrated in Figures~\ref{subfig:mve_small} and~\ref{subfig:mve_large}, increasing the $\alpha/\beta$ ratio substantially reduces dominating-region overlap and lowers the average query cost from 732.6 to 103.34.

\begin{figure*}[t]
    \begin{minipage}[t]{0.5\textwidth}
        \centering
        \subfigure{
        {\includegraphics[height=2.5cm]{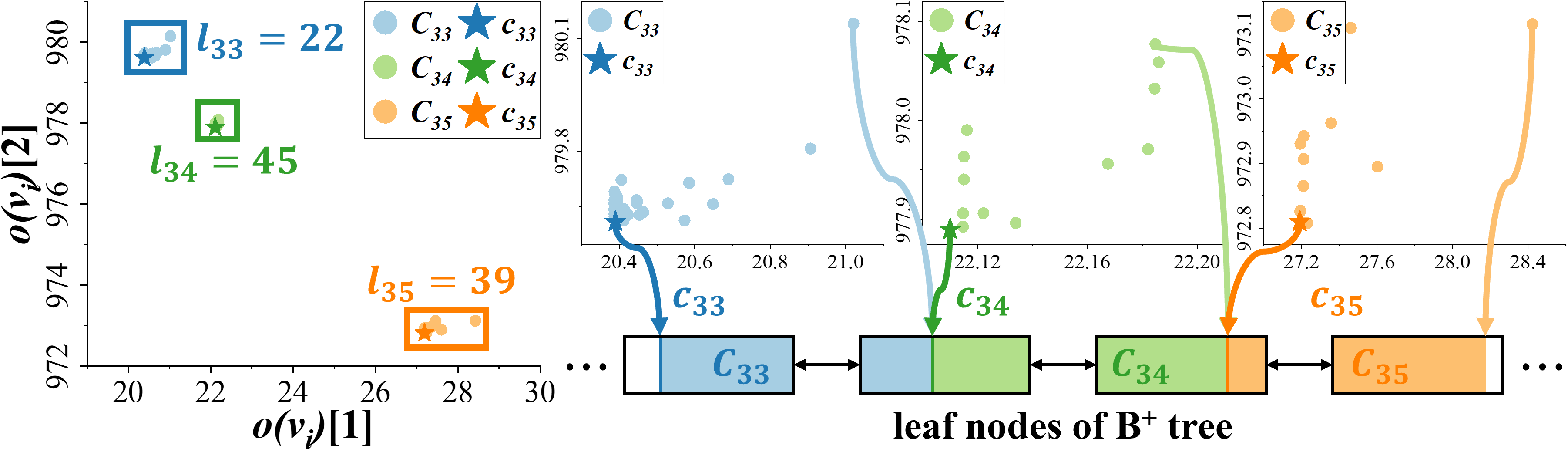}}} 
        \caption{An illustration of our \textit{iLabel} index design.}
        \label{fig:index_design}
    \end{minipage}\hfill
    \begin{minipage}[t]{0.5\textwidth}
        \centering
        \subfigure{
        {\includegraphics[height=2.8cm]{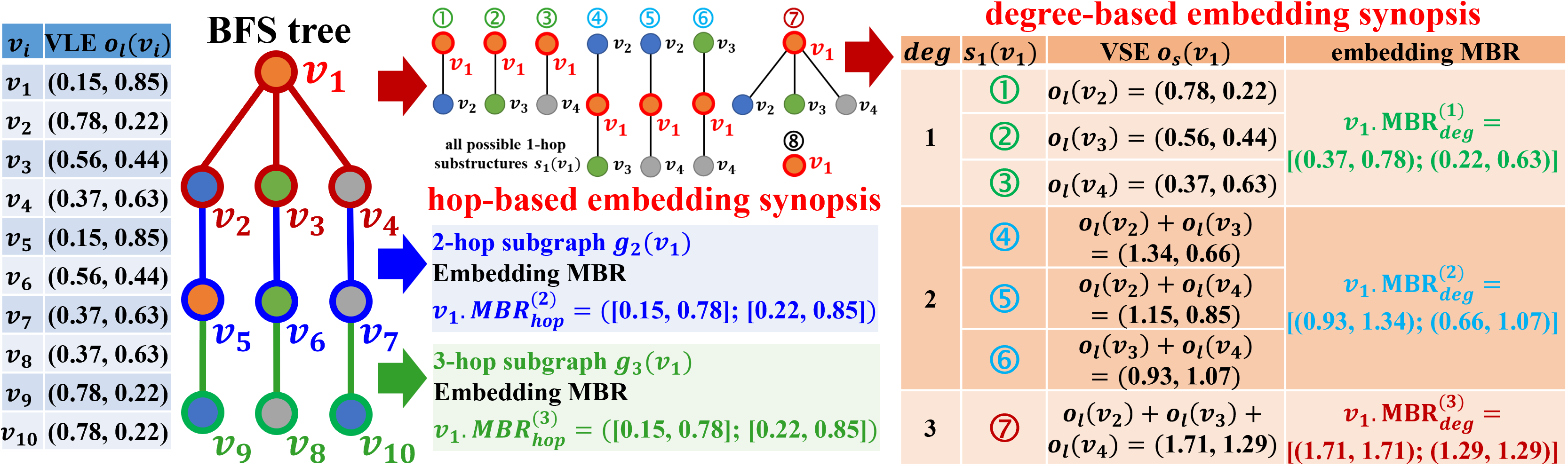}}}
        \caption{An illustration of our auxiliary synopsis design.}
        \label{fig:synopsis_example}
    \end{minipage}
\end{figure*}

\section{\textit{iLabel}: A Dominance-Preserving Index for Monotonic Vertex Embeddings}
\label{sec:index}
The monotonic vertex embeddings in Section~\ref{sec:embedding} enable exact vertex-level filtering via dominance relationships, but efficiently enumerating dominating vertices remains challenging for high-dimensional embeddings on large graphs. 
To address this, we propose \textit{iLabel}, a lightweight index that preserves dominance relationships while mapping multi-dimensional embeddings into a one-dimensional key space, enabling efficient range-based candidate retrieval.

\subsection{Design of the \textit{iLabel} Index}
\label{subsec:index_design}
The design of \textit{iLabel} exploits two properties of the learned monotonic vertex embeddings: 
(i) vertex label embeddings (VLE) naturally form semantic clusters with identical vertex labels, and 
(ii) the final embedding ($o(\cdot)$) is a weighted combination of label and structural components with $\alpha \gg \beta$, enabling non-overlapping clusters.

\noindent{\bf Label-based Clustering}
Recall that in the monotonic embedding design (Section~\ref{subsec:embedding_design}), the vertex label embedding (VLE) $o_l(v)$ is solely determined by the vertex label. Consequently, all vertices sharing the same label are mapped to the same VLE vector, naturally partitioning data vertices into disjoint label-based clusters:
\begin{equation}
    \mathcal{C} = \{C_1, C_2, \ldots, C_m\}, \qquad
    C_i = \{v_j \mid L(v_j) = l_i \},
\end{equation}
where each cluster $C_i$ consists of vertices in $G$ with label $l_i$.

For each cluster $C_i$, we define its center as the corresponding label embedding: $c_i = o_l(v_j), \forall v_j \in C_i$. Within a cluster, vertices differ only in their structural embeddings $o_s(v_j)$, which capture local neighborhood variations around the same semantic label.

During query processing, label-based clustering restricts candidate retrieval to clusters with labels compatible with the query vertex, thereby avoiding unnecessary access to irrelevant vertices.

\noindent{\bf One-dimensional Key Mapping.}
Recall that each vertex embedding is constructed as
$o(v)=\alpha\,o_l(v)+\beta\,o_s(v)$,
where $\alpha$ and $\beta$ are non-negative weights with $\alpha \gg \beta$.
This design ensures that inter-cluster separation is dominated by label embeddings,
while intra-cluster ordering is governed by structural embeddings.

\textit{iLabel} maps each vertex embedding to a one-dimensional key:
\begin{equation}
    \textit{key}(v) = \alpha \|o_l(v)\|_2 + \beta \|o_s(v)\|_2,
    \label{eq:index_key}
\end{equation}
where the first term defines a label-dependent base offset in the key space,
and the second term introduces a fine-grained ordering among vertices sharing the same label.
This mapping projects high-dimensional monotonic embeddings
into a one-dimensional space while preserving their relative dominance order.

We next show that this key mapping is sound for
dominance-based pruning, \ie~dominance in the embedding space implies
a consistent ordering in the key space.

\begin{lemma}
\textbf{(Dominance Preservation under One-dimensional Key Mapping).}
Let $v_i$ and $v_j$ be two data vertices with monotonic vertex embeddings
$o(v_i)$ and $o(v_j)$.
If $o(v_i) \preceq o(v_j)$ holds, 
then their corresponding keys satisfy
$\textit{key}(v_i) \le \textit{key}(v_j)$.
\label{lemma:key_dominance_preservation}
\end{lemma}

\begin{proof}
By definition of dominance,
$o(v_i) \preceq o(v_j)$ implies $o(v_i)[k] \le o(v_j)[k]$ for all dimensions $k$.
Since both $o_l(\cdot)$ and $o_s(\cdot)$ are non-negative vectors by construction,
their $L_2$ norms are monotonic with respect to coordinate-wise dominance.
Thus,
$\|o_l(v_i)\|_2 \le \|o_l(v_j)\|_2$
and
$\|o_s(v_i)\|_2 \le \|o_s(v_j)\|_2$.
Given $\alpha,\beta \ge 0$, it follows that
$\textit{key}(v_i)
= \alpha \|o_l(v_i)\|_2 + \beta \|o_s(v_i)\|_2
\le
\alpha \|o_l(v_j)\|_2 + \beta \|o_s(v_j)\|_2
= \textit{key}(v_j)$,
which completes the proof.
\end{proof}

\noindent{\bf Strict Separation of Label Clusters in the Key Space.}
While Lemma~\ref{lemma:key_dominance_preservation} guarantees that dominance
relations are preserved under the one-dimensional key mapping,
efficient index traversal further benefits from a stronger property:
key ranges of different label clusters should be strictly non-overlapping.
Intuitively, choosing $\alpha$ sufficiently larger than $\beta$
(\ie~$\alpha \gg \beta$) ensures that
label semantics dominate inter-cluster separation, while preserving structural discrimination within each cluster.
We formalize this intuition by giving a sufficient condition under which label-cluster key ranges in the \textit{iLabel} index are guaranteed to be strictly non-overlapping.

\begin{lemma}
\textbf{(Non-overlapping Condition for Label Clusters).}
Let $M=\max_{v_i\in V(G)} \|o_s(v_i)\|_2$
be the maximum $L_2$ norm of vertex structure embeddings (VSEs) over all data vertices, and let
$\Delta_{\min}=\min_{i}(\|o_l(l_{i+1})\|_2-\|o_l(l_i)\|_2)$
be the minimum difference between the $L_2$ norms of vertex label embeddings (VLEs) of two adjacent labels, where labels are ordered by increasing $\|o_l(l)\|_2$.
If the following holds:
\begin{equation}
\frac{\alpha}{\beta} > \frac{M}{\Delta_{\min}},
\label{eq:strict_ratio_condition}
\end{equation}
then the key ranges of any two distinct label clusters in the iLabel index are strictly non-overlapping.
\label{lemma:strict_cluster_separation}
\end{lemma}

\begin{proof}
For any data vertex $v_i\in V(G)$ with label $l=L(v_i)$, since $\|o_s(v_i)\|_2\in[0,M]$, its key value satisfies:
\[
key(v_i)
=
\alpha\|o_l(l)\|_2+\beta\|o_s(v_i)\|_2
\in
\big[
\alpha\|o_l(l)\|_2,\;
\alpha\|o_l(l)\|_2+\beta M
\big].
\]
Thus, the key range of the label cluster corresponding to $l$ is bounded by the above interval.

Consider two adjacent labels $l_i$ and $l_{i+1}$ in ascending order of $\|o_l(l)\|_2$.
The maximum key value in cluster $C_{l_i}$ is at most
$\alpha\|o_l(l_i)\|_2+\beta M$,
while the minimum key value in cluster $C_{l_{i+1}}$ is at least
$\alpha\|o_l(l_{i+1})\|_2$.
If $\alpha(\|o_l(l_{i+1})\|_2-\|o_l(l_i)\|_2) > \beta M$,
or equivalently Eq.~\ref{eq:strict_ratio_condition} holds, then the following condition can be satisfied:
\[
\max_{v\in C_{l_i}} key(v)
<
\min_{v\in C_{l_{i+1}}} key(v),
\]
which implies that the two cluster key ranges are strictly non-overlapping.
Since this condition holds for every adjacent pair of labels, all label clusters in the index are strictly separated.
\end{proof}

In practice, the value of $M$ can be obtained by a single linear scan during offline index construction, as $\|o_s(v_i)\|_2$ is already computed for key generation.
Similarly, $\Delta_{\min}$ can be computed after sorting label embeddings by their $L_2$ norms.
Thus, the condition in Eq.~\ref{eq:strict_ratio_condition} can be efficiently verified and enforced in the offline phase.

\begin{example}
Figure~\ref{fig:index_design} illustrates how \textit{iLabel} organizes monotonic vertex embeddings via label-based clustering and one-dimensional key mapping. In this example, vertices have three labels, forming clusters $C_{33}$, $C_{34}$, and $C_{35}$ corresponding to labels $l_{33}$, $l_{34}$, and $l_{35}$, respectively. Each cluster is anchored by its label embedding, serving as the cluster center, and vertices within a cluster share the same label embedding while differing in structural embeddings.

Following Eq.~\ref{eq:index_key}, each vertex is mapped to a one-dimensional key $key(v_i)=\alpha\|o_l(v_i)\|_2+\beta\|o_s(v_i)\|_2$ with $\alpha \gg \beta$. The first term determines the cluster offset, while the second induces a fine-grained ordering within the cluster. Consequently, key ranges of different clusters are strictly separated (\eg~$\max C_{33} < \min C_{34}$), enabling range queries to be performed independently within each cluster. Meanwhile, variations in the structural component $\beta\|o_s(v_i)\|_2$ order vertices within the same cluster, enabling efficient enumeration of candidates that satisfy dominance-based filtering.
\end{example}

\subsection{Index Construction}
\label{subsec:index_construction}
Based on the \textit{iLabel} key mapping in Section~\ref{subsec:index_design}, this subsection describes how the index is constructed and organized, and how auxiliary information is incorporated to improve pruning efficiency during traversal. 
The index structure is solely responsible for organizing vertex keys and supporting efficient range search, while all auxiliary synopses are precomputed offline and used only as lightweight pruning aids during online query processing.

\noindent{\bf Index Construction.}
The construction of the \textit{iLabel} index is performed entirely offline. After training the embedding model $\text{Emb}(\cdot)$ using the cost-driven objective in Section~\ref{subsec:embedding_design}, we generate the monotonic vertex embedding $o(v_i)$ for each data vertex $v_i \in V(G)$ and map it to a one-dimensional key using Eq.~\ref{eq:index_key}. A B$^+$-tree index $\mathcal{I}$ is then built over these keys to support efficient range-based access.

The \textit{iLabel} index follows a standard B$^+$-tree organization. Each leaf node stores data vertices together with their key values and vertex identifiers, while non-leaf nodes maintain routing entries that record the key ranges of their child subtrees. 
Non-leaf nodes store only key ranges, with no aggregated embedding or structural information, enabling traversal solely via key comparisons and keeping the index lightweight and scalable.

\noindent{\bf Auxiliary Synopses for Pruning.}
While key-based range search effectively narrows candidate retrieval, 
dominance-based filtering alone may still yield false positives. 
Thus, we precompute two types of auxiliary embedding synopses for each data vertex;
these synopses are lightweight geometric summaries that 
enable early pruning during index traversal without introducing false dismissals.

\underline{\it Hop-based Embedding Synopsis.}
We construct a \emph{Minimum Bounding Rectangle} (MBR)–based embedding synopsis $v_i.\textit{MBR}_{hop}^{(t)}$ for each vertex $v_i$ and hop level $t$ ($t \ge 1$), to capture higher-order neighborhood semantics. 
For the $t$-hop subgraph $g_t(v_i)$ centered at $v_i$, we compute per-dimension lower and upper bounds over the vertex label embeddings of all neighbor vertices of $v_i$ within $g_t(v_i)$:
\begin{equation}
\begin{aligned}
    v_i.\textit{MBR}_{hop}^{(t)}[2k]   &= \min_{v_j \in V(g_t(v_i))} o_l(v_j)[k], \\
    v_i.\textit{MBR}_{hop}^{(t)}[2k+1] &= \max_{v_j \in V(g_t(v_i))} o_l(v_j)[k],
\end{aligned}
\end{equation}
for $k \in [0, d)$. 
The resulting $2d$-dimensional vector bounds the $t$-hop neighborhood embeddings, enabling pruning based on multi-hop semantic inconsistency without explicit neighborhood traversal.

\underline{\it Degree-based Embedding Synopsis.}
We also construct a degree-based embedding synopsis $v_i.\textit{MBR}_{deg}^{(\delta)}$ to summarize degree-$\delta$ neighbor combinations within the 1-hop neighborhood of a vertex $v_i$. 
Specifically, given the 1-hop neighbors $\mathcal{N}_1(v_i)$ of the vertex $v_i$, we maintain sorted lists of their label embedding values, denoted by $v_i.\textit{vle\_list}_k$, for each embedding dimension $k \in [0,d)$. 
For a specified substructure degree $\delta$, the synopsis records lower and upper bounds on the aggregated structural embeddings of all possible $\delta$-neighbor combinations:
\begin{equation}
\begin{aligned}
    v_i.\textit{MBR}_{deg}^{(\delta)}[2k]   &= \sum_{r=1}^{\delta} v_i.\textit{vle\_list}_k[r], \\
    v_i.\textit{MBR}_{deg}^{(\delta)}[2k+1] &= \sum_{r=\deg(v_i)-\delta+1}^{\deg(v_i)} v_i.\textit{vle\_list}_k[r],
\end{aligned}
\end{equation}
for $k \in [0,d)$. 
This bounds the structural embeddings of all degree-$\delta$ 1-hop substructures, avoiding explicit enumeration of substructures whose number grows exponentially with vertex degree.

\begin{example}
Figure~\ref{fig:synopsis_example} illustrates the construction of the degree-based embedding synopsis for a vertex $v_1$ with three 1-hop neighbors $\{v_2, v_3, v_4\}$, where different neighbor combinations induce different 1-hop substructures.
For degree $\delta=1$, the three possible substructures yield VSEs $(0.78,0.22)$, $(0.56,0.44)$, and $(0.37,0.63)$, resulting in $v_1.\textit{MBR}_{deg}^{(1)} = ([0.37,0.78]; [0.22,0.63])$.
For $\delta=2$, the three neighbor pairs produce VSEs $(1.34,0.66)$, $(1.15,0.85)$, and $(0.93,1.07)$, giving $v_1.\textit{MBR}_{deg}^{(2)} = ([0.93,1.34]; [0.66,1.07])$.
When $\delta=3$, all neighbors are combined, yielding a single VSE $(1.71,1.29)$ and $v_1.\textit{MBR}_{deg}^{(3)} = ([1.71,1.71]; [1.29,1.29])$.

These synopses compactly bound the structural embeddings of all degree-$\delta$ ($\in [1, 3]$) 1-hop substructures of $v_1$. 
By precomputing them offline, \textsc{LIVE} avoids explicit substructure enumeration while enabling effective degree-aware pruning in online 
matching.
\end{example}

\noindent{\bf Complexity Analysis.}
We provide the time and space complexity of constructing the \textit{iLabel} index together with its auxiliary synopses as follows.

\underline{\it Time.}
\nop{
It needs $O(d \cdot |E(G)| + t \cdot (|V(G)| + |E(G)|)
+ \sum_{v_i \in V(G)} d \cdot deg(v_i)\log deg(v_i)
+ |V(G)|\log |V(G)|)$ time to build the \textit{iLabel} index and its auxiliary information, which is near-linear in the graph size for sparse graphs and small hop number~$t$.
}
Given a data graph $G$, for each vertex $v_i \in V(G)$, computing its monotonic embedding
$o(v_i)=\alpha o_l(v_i)+\beta o_s(v_i)$
requires aggregating the label embeddings of its 1-hop neighbors.
This incurs $O(d \cdot deg(v_i))$ time per vertex, where $d$ is the embedding dimension.
Since $\sum_{v_i \in V(G)} deg(v_i) = 2|E(G)|$, the total cost of embedding generation is
$O(d \cdot |E(G)|)$.

For hop-based synopses, we compute $t$-hop neighborhoods using BFS expansion.
For each hop level, all vertices and edges are visited at most once.
Thus, constructing hop-based MBRs for all vertices takes
$O(t \cdot (|V(G)| + |E(G)|))$ time, where $t$ is a small constant.

To computer degree-based synopses, for each vertex $v_i$, we sort the label embeddings of its 1-hop neighbors once per embedding dimension.
This yields a per-vertex cost of $O(d \cdot deg(v_i)\log deg(v_i))$.
Summing over all vertices, the total time complexity is
$O(\sum_{v_i \in V(G)} d \cdot $\\$ deg(v_i)\log deg(v_i))$.

After computing all key values, data vertices are sorted by $key(v_i)$, which takes
$O(|V(G)|\log |V(G)|)$ time.
Using bulk loading, B$^+$-tree can be constructed in linear time with respect to the number of sorted records, i.e., $O(|V(G)|)$.

Combining all steps, the overall time complexity of \textit{iLabel} index construction is: $O(
d \cdot |E(G)|
+ t \cdot (|V(G)| + |E(G)|)
+ \sum_{v_i \in V(G)} d \cdot $\\$deg(v_i)\log deg(v_i)
+ |V(G)|\log |V(G)|)$,
which is near-linear in the graph size and small $t$.

\underline{\it Space.}
\nop{
It needs $O(|V(G)| \cdot d \cdot t + d \cdot |E(G)|)$ space to
build the \textit{iLabel} index and the synopses,
which scales linearly with graph size and embedding dimension, making it suitable for large graphs.
}
Given a data graph $G$ and embedding dimension $d$, each data vertex $v_i$ stores its $d$-dimensional monotonic embedding $o(v_i)$ and two types of auxiliary synopses. Storing vertex embeddings for all vertices requires $O(|V(G)| \cdot d)$ space.

For hop-based embedding synopses, each vertex $v_i$ maintains $t$ hop-based MBRs $v_i.MBR_{hop}^{(t)}$, each consisting of two $d$-dimensional boundary vectors. The total space for hop-based synopses is therefore $O(|V(G)| \cdot t \cdot d)$.

For degree-based embedding synopses, each vertex $v_i$ stores one degree-based MBR $v_i.MBR_{deg}^{(\delta)}$ for each possible substructure degree
$\delta \le deg(v_i)$.
Each MBR consists of two $d$-dimensional vectors, resulting in
$O(d \cdot deg(v_i))$ space per vertex.
Summed over all vertices, the total space requirement is
$O(d \cdot |E(G)|)$.

The B$^+$-tree stores only one-dimensional keys and pointers.
Both leaf and non-leaf nodes require $O(|V(G)|)$ space in total, which is negligible compared to embedding and synopsis storage.

Overall, the space complexity of the \textit{iLabel} index is $O(|V(G)| \cdot d \cdot t$\\$ + d \cdot |E(G)|)$, which scales linearly with the graph size and embedding dimension, and is suitable for large graphs.

\section{Subgraph Matching with Monotonic Vertex Embedding}
\label{sec:algorithm}
This section presents the online subgraph matching algorithm built on the learned monotonic vertex embeddings and the \textit{iLabel} index.
Given a query graph, the online phase efficiently retrieves candidate vertices for each query vertex without false dismissals and assembles valid matches through refinement.

We introduce low-cost pruning strategies for index traversal and candidate retrieval in Section~\ref{subsec:pruning}, and then present the integrated online subgraph matching algorithm in Section~\ref{subsec:online-algorithm}.

\subsection{Pruning Strategies}
\label{subsec:pruning}
During online query processing, candidates for each query vertex are retrieved by traversing the \textit{iLabel} index and scanning data vertices within a query-specific key range. To efficiently eliminate invalid candidates without sacrificing correctness, we apply the following pruning strategies to progressively refine the candidate set. These strategies are derived from necessary conditions under monotonic embeddings and therefore introduce no false dismissals.

\noindent{\bf Key Lower-Bound Pruning (Strategy~1).}
For a query vertex $q_i$, we first derive a query-specific lower bound $key(q_i)$ from the B$^+$-tree in the \textit{iLabel} index. 
Due to the consistency between embedding dominance and key ordering, any data vertex with a key value smaller than $key(q_i)$ cannot be dominated by $o(q_i)$ and therefore cannot be a valid candidate. 
Such vertices, along with their corresponding index entries, can be safely pruned.

\begin{lemma}
\textbf{(Key Lower-Bound Pruning).}
Given a query vertex $q_i \in q$ with $key(q_i)$, for any data vertex $v_i \in G$ with $key(v_i)$ or index entry $N_i \in \mathcal{I}$ with key range $[N_i.min, N_i.max)$, if $key(v_i) < key(q_i)$ or $N_i.max < key(q_i)$ holds, then $v_i$ or the entire subtree rooted at $N_i$ can be safely pruned.
\label{lemma:key_lower}
\end{lemma}

\begin{proof}
By the construction of the key mapping (Eq.~\ref{eq:index_key}), key ordering is consistent with the dominance ordering of vertex embeddings.
If $key(v_i) < key(q_i)$, then $o(q_i) \npreceq o(v_i)$, implying that $v_i$ cannot be dominated by $o(q_i)$ and thus cannot be a valid candidate.

Similarly, if $N_i.\max < key(q_i)$, then all data vertices in the subtree rooted at $N_i$ have keys smaller than $key(q_i)$ and therefore cannot be dominated by $o(q_i)$. 
Pruning under these conditions preserves all valid candidates and introduces no false dismissals.
\end{proof}

\noindent{\bf Key Upper-Bound Pruning (Strategy~2).}
After determining the lower bound of the query range, we further restrict candidate retrieval to a label-consistent key interval.
Since key ranges of different label clusters are strictly disjoint, any data vertex with a key no smaller than the key of the next label (\ie~$key(L(q_i)+1)$) belongs to a different label cluster and cannot match the query vertex. 
Such vertices and index entries can therefore be safely pruned.

\begin{lemma}
\textbf{(Key Upper-Bound Pruning).}
Given a query vertex $q_i \in q$ with label $L(q_i)$, for any data vertex $v_i \in G$ with key $key(v_i)$ or index entry $N_i \in \mathcal{I}$ with key range $[N_i.min, N_i.max)$, if $key(v_i) \ge key(L(q_i)+1)$ or $N_i.min \ge key(L(q_i)+1)$ holds, then $v_i$ or the entire subtree rooted at $N_i$ can be safely pruned.
\label{lemma:key_upper}
\end{lemma}

\begin{proof}
By the \textit{iLabel} key construction, all vertices with label $L(q_i)$ are indexed within the interval $[key(L(q_i)),\, key(L(q_i)+1))$, and key intervals of different labels are disjoint. Hence, any data vertex with $key(v_i) \ge key(L(q_i)+1)$ belongs to a different label cluster and cannot match $q_i$.

Similarly, if $N_i.\min \ge key(L(q_i)+1)$, then all vertices in the sub-tree rooted at $N_i$ fall outside the label-consistent key range of $q_i$. Pruning under these conditions preserves all valid candidates and introduces no false dismissals.
\end{proof}

\noindent{\bf Embedding Dominance Pruning (Strategy~3).}
Within the query-specific key range, we scan data vertices and apply a dominance check against the query embedding. 
Any vertex whose embedding are not dominated by the query embedding cannot satisfy 1-hop containment and is safely pruned.

\begin{lemma}
\textbf{(Embedding Dominance Pruning).}
Given a query vertex $q_i$ and a data vertex $v_i$, if their monotonic vertex embeddings satisfy $o(q_i) \npreceq o(v_i)$ (\ie~$o(v_i)$ is not dominated by $o(q_i)$), then $v_i$ cannot be a valid candidate for $q_i$ and can be safely pruned.
\label{lemma:vertex_dominance}
\end{lemma}

\begin{proof}
A data vertex $v_i$ matches $q_i$ only if the 1-hop subgraph $g_1(q_i)$ is a substructure of $g_1(v_i)$.
Under the monotonic vertex embedding design (Section~\ref{subsec:embedding_design}), this subgraph containment relationship implies $o(q_i) \preceq o(v_i)$ in the embedding space.
If $o(q_i) \npreceq o(v_i)$ holds, then the necessary dominance condition is violated, and $v_i$ cannot satisfy the required subgraph relationship.
Therefore, pruning $v_i$ does not introduce false dismissals.
\end{proof}

\noindent{\bf Hop-based Synopsis Pruning (Strategy~4).}
Vertex-level dominance alone does not capture higher-order structural feasibility. To further prune false positives at low cost, we exploit hop-based embedding synopses that bound multi-hop neighborhoods in the embedding space. During candidate evaluation, we compare query and data synopses in a coarse-to-fine manner, starting from larger hop distances and terminating early upon detecting a violation.

\begin{lemma}
\textbf{(Hop-based Synopsis Pruning).}
Given a query vertex $q_i$ and a data vertex $v_i$, let $q_i.MBR_{hop}^{(t)}$ and $v_i.MBR_{hop}^{(t)}$ denote their hop-based embedding synopses at hop distance $t$.
If there exists $t \in \{2, \ldots, k\}$ such that $q_i.MBR_{hop}^{(t)} \nsubseteq v_i.MBR_{hop}^{(t)}$, then $v_i$ cannot be a valid candidate for $q_i$ and can be safely pruned.
\label{lemma:hop_synopsis}
\end{lemma}

\begin{proof}
The synopsis $v_i.\textit{MBR}_{hop}^{(t)}$ bounds the embeddings of all vertices in the $t$-hop neighborhood of $v_i$.
If $v_i$ were a valid candidate for $q_i$, the $t$-hop neighborhood of $q_i$ would have to be realizable within that of $v_i$, which under monotonic embeddings requires $q_i.\textit{MBR}_{hop}^{(t)} \subseteq v_i.\textit{MBR}_{hop}^{(t)}$.
Violation of this condition for any $t$ implies infeasibility.
Thus, $v_i$ can be safely pruned.
\end{proof}

\noindent{\bf Degree-based Synopsis Pruning (Strategy~5).}
For high-degree data vertices, explicitly enumerating 1-hop substructures is prohibitive. Instead, we leverage degree-based embedding synopses (Section~\ref{subsec:index_construction}) that bound the structural embeddings of all 1-hop substructures of a given degree.
A necessary condition for matching a query vertex to a data vertex is that the query’s structural embedding lies within the corresponding degree-based bounding region; otherwise, the data vertex can be safely pruned.

\begin{lemma}
\textbf{(Degree-based Synopsis Pruning).}
Given a query vertex $q_i$ and a data vertex $v_i$, let $q_i.vse$ denote the vertex structure embedding of $q_i$.
If $q_i.vse \notin v_i.MBR_{deg}^{(deg(q_i))}$, then $v_i$ cannot be a valid candidate for matching $q_i$ and can be safely pruned.
\label{lemma:degree_synopsis}
\end{lemma}

\begin{proof}
If $v_i$ is a valid candidate for $q_i$, then the 1-hop subgraph $g_1(q_i)$ must correspond to one of the 1-hop substructures of $v_i$ with degree $deg(q_i)$.
By construction, $v_i.MBR_{deg}^{(\delta)}$ bounds the structural embeddings of all such substructures with degree $\delta$.

Hence, $q_i.vse$ must lie within $v_i.MBR_{deg}^{(deg(q_i))}$.
If this condition is violated, no feasible 1-hop substructure of $v_i$ can match $g_1(q_i)$.
Pruning $v_i$ under this condition introduces no false dismissals.
\end{proof}

\begin{algorithm}[t]
\caption{{\bf Exact Subgraph Matching with Learnable Monotonic Vertex Embedding}}
\label{alg:online_query}
{\footnotesize
\KwIn{
    a query graph $q$, a data graph $G$, 
    a trained embedding model $\text{Emb}(\cdot)$, and an \textit{iLabel} 
    index $\mathcal{I}$
}
\KwOut{
    a set $\mathcal{S}$ of matching subgraphs in $G$
}

\tcp{Candidate retrieval for each query vertex}
\For{each query vertex $q_i \in V(q)$}{
    generate query vertex embedding $o(q_i)$ using $\text{Emb}(\cdot)$\;
    compute hop-based synopsis $q_i.MBR_{hop}$ of $q_i$\;

    \tcp{Locate the first relevant leaf node}
    $N \leftarrow \mathcal{I}.root$\;
    \While{$N$ is not a leaf node}{
        \For{each entry $N_j \in N$}{
            \If{$[N_j.min, N_j.max) \cap [key(q_i), key(L(q_i)+1)) \neq \emptyset$}{\tcp{Lemmas \ref{lemma:key_lower} and \ref{lemma:key_upper}}
                $N \leftarrow N_j$\;
                \textbf{break}\;
            }
        }
    }

    \tcp{Sequential leaf scan within the query range}
    \While{$N$ is not null}{
        \For{each data vertex $v_i \in N$}{
            \If{$key(v_i) \ge key(L(q_i)+1)$}{\tcp{Lemma \ref{lemma:key_upper}}
                \textbf{break}\;
            }
            \If{$o(q_i) \preceq o(v_i)$}{\tcp{Lemma \ref{lemma:vertex_dominance}}
                \If{$q_i.MBR_{hop}^{(t)} \subseteq v_i.MBR_{hop}^{(t)}$ for all $t \in [2,k]$}{\tcp{Lemma \ref{lemma:hop_synopsis}}
                    \If{$q_i.vse \in v_i.MBR_{deg}^{(deg(q_i))}$}{\tcp{Lemma \ref{lemma:degree_synopsis}}
                        $q_i.cand\_set.add(v_i)$\;
                    }
                }
            }
        }
        $N \leftarrow N.next$\;
    }
}

\tcp{Matching order generation}
generate an ordered list $Q$ of query vertices $q_i \in q$ based on candidate set sizes $|q_i.cand\_set|$\;

\tcp{Candidate refinement}
invoke a backtracking search procedure to enumerate valid matches in $G$ and add them to $\mathcal{S}$\;

\Return{$\mathcal{S}$}\;
}
\end{algorithm}

\subsection{Online Subgraph Matching Algorithm}
\label{subsec:online-algorithm}
Algorithm~\ref{alg:online_query} presents the complete online subgraph matching procedure based on monotonic vertex embeddings and the \textit{iLabel} index $\mathcal{I}$.
It consists of three stages: candidate retrieval with index traversal and pruning (lines~1--18), matching order determination (line~19), and candidate refinement (line~20).

\noindent{\bf Candidate Retrieval (Stage~1).}
For each query vertex $q_i \in V(q)$, the algorithm first computes its monotonic embedding $o(q_i)$ using the trained embedding model $\text{Emb}(\cdot)$ and computes the hop-based synopsis $q_i.MBR_{hop}$ with the embeddings of its hop-based neighbors (lines~1--3).
To retrieve candidate data vertices, a key-range search is performed on the \textit{iLabel} index $\mathcal{I}$.
In particular, starting from the root, the algorithm descends to the first leaf node whose key range intersects the query interval $[key(q_i),\, key(L(q_i)+1))$, leveraging key lower- and upper-bound pruning to skip irrelevant subtrees (lines~4--9; Lemmas~\ref{lemma:key_lower} and~\ref{lemma:key_upper}).

When reaching the leaf level, it performs a sequential scan over leaf nodes within the query range using the linked-leaf structure and terminates once the upper bound is exceeded (lines~10--13).
For each encountered data vertex, a cascade of pruning strategies is applied, including embedding dominance pruning (line~14; Lemma~\ref{lemma:vertex_dominance}), hop-based synopsis pruning (line~15; Lemma~\ref{lemma:hop_synopsis}), and degree-based synopsis pruning (lines~16--17; Lemma~\ref{lemma:degree_synopsis}).
Vertices that pass all pruning checks are added to the candidate set $q_i.\textit{cand\_set}$ (line~17).

\noindent{\bf Matching Order Generation (Stage~2).}
After constructing candidate sets for all query vertices, the algorithm determines a matching order $Q$ to guide the refinement phase (line~19).
To minimize intermediate join results, it selects the query vertex from $q$ with the smallest candidate set as the starting point and then iteratively appends adjacent query vertices with the smallest candidate sets until all vertices are ordered.

\noindent{\bf Refinement (Stage~3).}
Then, a backtracking-based refinement procedure~\cite{kankanamge2017graphflow,shang2008taming,ye2025continuous} is applied to enumerate valid matches using the candidate sets and the matching order (line~20).
During refinement, partial mappings are incrementally extended while enforcing injective mapping constraints and edge-consistency conditions between the query graph $q$ and the data graph $G$; and all complete mappings that satisfy the query are added into the result set $\mathcal{S}$.
Finally, the algorithm returns $\mathcal{S}$ as the final output (line~21).

\noindent{\bf Complexity Analysis.}
\nop{
The overall time complexity of Algorithm~\ref{alg:online_query} is $O(\sum_{q_i \in V(q)}(deg(q_i)\cdot d + |E(q)| + k \cdot |V(q)| \cdot d + h + R_i \cdot k \cdot d) + |V(q)|^2 + \prod_{q_i \in V(q)} |q_i.cand\_set|)$, where $h$ is the height of the \textit{iLabel} index. 
Please refer to Appendix \ref{subsec:app_query_time} in~\cite{appendix} for more details.
}
For each query vertex $q_i \in V(q)$, generating its monotonic vertex embedding requires aggregating embeddings of its 1-hop neighbors, which costs $O(deg(q_i)\cdot d)$ time (line~2), where $d$ is the embedding dimension.
Computing the hop-based synopsis $q_i.MBR_{hop}$ involves performing BFS expansions on the query graph up to $k$ hops while maintaining per-dimension minimum and maximum values.
This costs $O(|E(q)| + k \cdot |V(q)| \cdot d)$ per query vertex (line~3).
Since query graphs are typically small, this cost is negligible compared to index traversal and refinement.

For each query vertex, descending the B$^+$-tree from the root to the first relevant leaf node incurs $O(h)$ time, where $h$ is the height of the index (lines~4--9).
Let $R_i$ denote the number of data vertices whose keys fall within the query-specific range $[key(q_i), key(L(q_i)+1))$.
During the sequential leaf scan, each such vertex is examined once.
For each data vertex, embedding dominance pruning costs $O(d)$ (line~14), hop-based synopsis pruning costs $O((k-1)\cdot d)$ (line~15), and degree-based synopsis pruning costs $O(d)$ (lines~16--17).
Thus, the total pruning cost per query vertex is $O(R_i \cdot k \cdot d)$ (lines~10--18).

After candidate sets are constructed, the greedy matching order generation procedure selects query vertices based on candidate set sizes.
In the worst case, this procedure takes $O(|V(q)|^2)$ time (line~19).

Finally, the refinement phase performs a backtracking-based exact subgraph matching guided by the matching order and candidate sets (line~20).
In the worst case, it enumerates all combinations of candidates, yielding 
$O(\prod_{q_i \in V(q)} |q_i.cand\_set|)$ time.

Therefore, the overall time complexity of Algorithm~\ref{alg:online_query} is $O($\\$\sum_{q_i \in V(q)}(deg(q_i)\cdot d + |E(q)| + k \cdot |V(q)| \cdot d + h + R_i \cdot k \cdot d) + |V(q)|^2 + \prod_{q_i \in V(q)} |q_i.cand\_set|)$.

\begin{table}[t]\scriptsize
\setlength{\tabcolsep}{1pt}
\begin{center}
\caption{Parameter settings.}
\label{tab:parameters}
\begin{tabular}{|l||l|}
\hline
\textbf{Parameters}&\textbf{Values} \\
\hline\hline
    the dimension, $d$, of the vertex embedding vector  & {\bf 2}, 3, 4, 5\\\hline
    the weight ratio, $\alpha/\beta$ & 1K, 10K, {\bf 100K}, 1M\\\hline
    the number, $t$, of synopsis hops & 1, {\bf 2}, 3, 4\\\hline    
    the size, $|V(q)|$, of the query graph $q$ & 5, 6, {\bf 8}, 10, 12\\\hline
    the average degree, $avg\_deg(q)$, of the query graph $q$ & 2, {\bf 3}, 4\\\hline  
    the number, $|\Sigma|$, of distinct labels & 5, 10, {\bf 15}, 20, 25\\\hline
    the average degree, $avg\_deg(G)$, of the data graph $G$ & 3, 4, {\bf 5}, 6, 7\\\hline     
    the size, $|V(G)|$, of the data graph $G$ & 10K, 30K, {\bf 50K}, 80K, 100K, 500K, 1M, 5M, 10M\\\hline 
\end{tabular}
\end{center}
\end{table}

\begin{table}[t]\scriptsize
\begin{center}
\caption{Statistics of real-world graph data sets.}
\label{tab:datasets}
\begin{tabular}{|l||c|c|c|c|}
\hline
\textbf{\text{ }\text{ }Data Sets}&\textbf{$|V(G)|$}&\textbf{$|E(G)|$}&\textbf{$|\Sigma|$}&\textbf{$avg\_deg(G)$} \\
\hline\hline
    Yeast (ye) & 3,112 & 12,519 & 71 & 8.0\\\hline
    HPRD (hp) & 9,460 & 34,998 & 307 & 7.4\\\hline
    DBLP (db) & 317,080 & 1,049,866 & 15 & 6.6\\\hline
    Youtube (yt) & 1,134,890 & 2,987,624 & 25 & 5.3\\\hline
    US Patents (up) & 3,774,768 & 16,518,947 & 20 & 8.8\\\hline
\end{tabular}
\end{center}
\end{table}

\section{Experimental Evaluation}
\label{sec:experiment}
To evaluate the effectiveness and efficiency of \textsc{LIVE}, we conducted extensive experiments on synthetic/real-world graph datasets.

\subsection{Experiment Settings}
\label{subsec:exp_setting}
All experiments were performed on an Ubuntu server equipped with an Intel Core i9-12900K CPU, 128GB memory, and an NVIDIA GeForce RTX~4090 GPU.
Offline embedding optimization and model training were implemented in PyTorch, while the online exact subgraph matching components,
including index traversal, pruning, and refinement, were implemented in C++.
We used the Adam optimizer with a learning rate of $\eta=0.01$ for embedding training.
To balance GPU memory usage and training cost, 
we fixed the number of sampled vertex pairs to $K=4{,}096$ and trained for $1{,}000$ epochs across all datasets.
The source code and all tested datasets are available at \url{https://github.com/JamesWhiteSnow/LIVE}.

\noindent{\bf Baseline Methods.}
We compared \textsc{LIVE} with eleven representative exact subgraph matching methods, 
including both classical rule-based approaches and recent learning-based solutions:
GraphQL (GQL)~\cite{he2008graphs},
QuickSI (QSI)~\cite{shang2008taming},
RI~\cite{bonnici2013subgraph},
CFLMatch (CFL)~\cite{bi2016efficient},
VF2++ (VF)~\cite{juttner2018vf2++},
DP-iso (DP)~\cite{han2019efficient},
CECI~\cite{bhattarai2019ceci},
Hybrid~\cite{sun2020memory} (a hybrid of GQL, RI, and QSI),
RapidMatch (RM)~\cite{sun2020rapidmatch},
GNN-PE~\cite{ye2024efficient},
and BSX~\cite{lu2025b}.
These baselines collectively represent the state of the art in exact subgraph matching with different indexing, pruning, and learning strategies, providing a comprehensive comparison against \textsc{LIVE}.

\noindent{\bf Graph Datasets.}
We evaluated \textsc{LIVE} on both synthetic and real-world graph datasets, following widely adopted experimental settings as in prior works~\cite{sun2020memory,sun2020rapidmatch,ye2024efficient,ye2025continuous,jiang2025comprehensive}.

\begin{figure*}[t]
    \begin{minipage}[t]{0.6\textwidth}
    \centering
        \subfigure[{\small varying $d$}]{\label{subfig:parameter_dim}
            {\includegraphics[height=2.7cm]{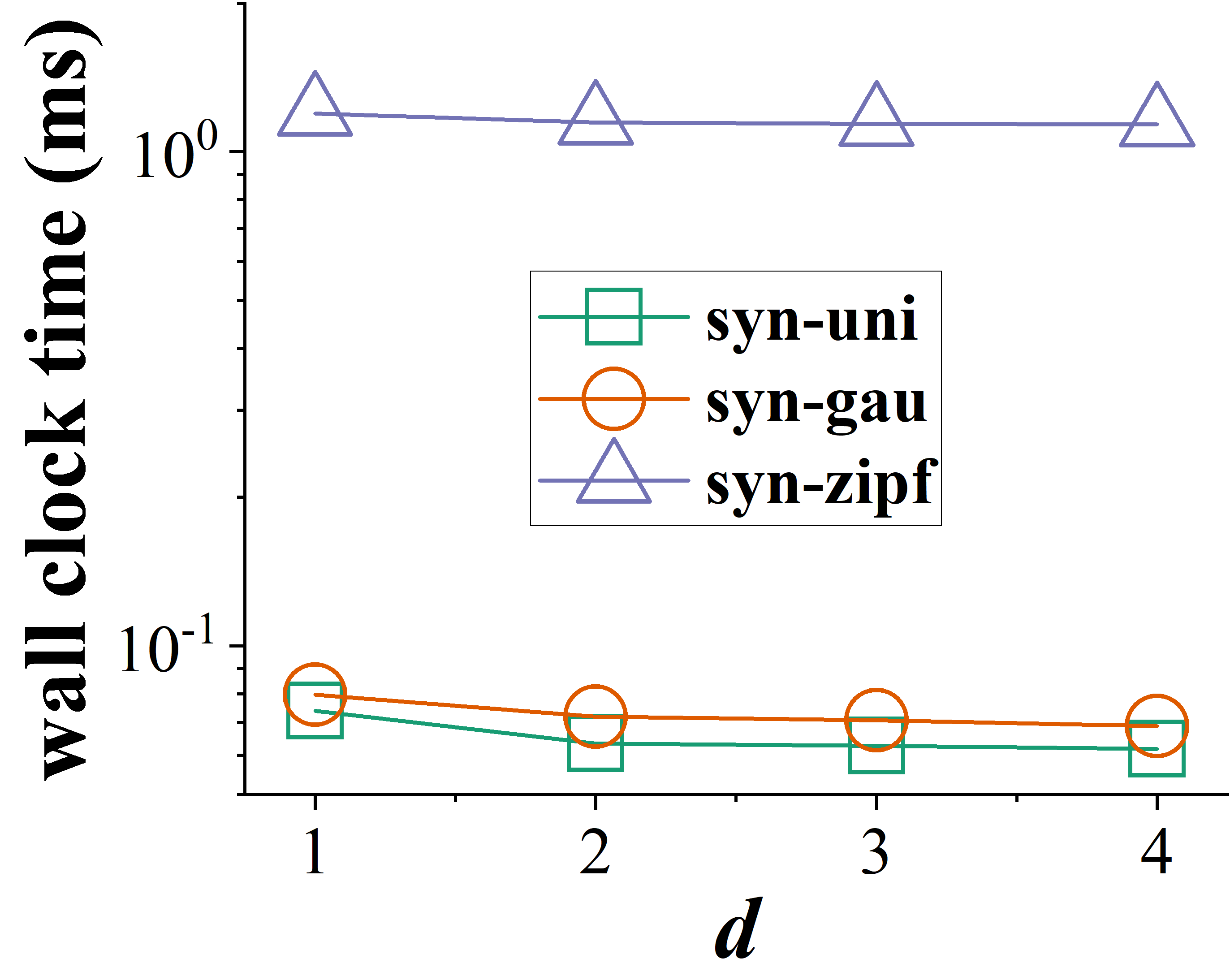}}} 
            \subfigure[{\small varying $\alpha/\beta$}]{\label{subfig:parameter_ratio}
            {\includegraphics[height=2.7cm]{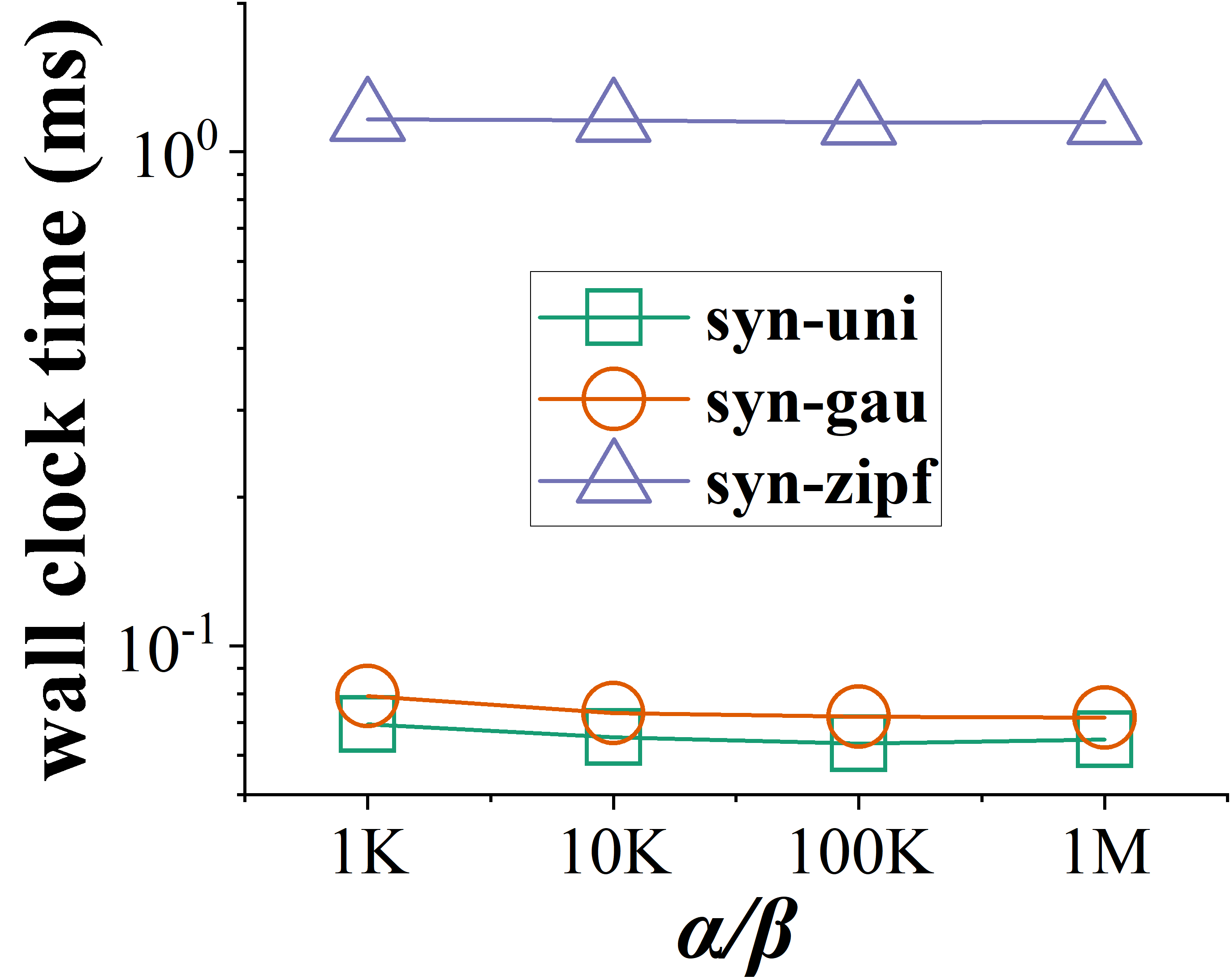}}} 
        \subfigure[{\small varying $t$}]{\label{subfig:parameter_hop}
            {\includegraphics[height=2.7cm]{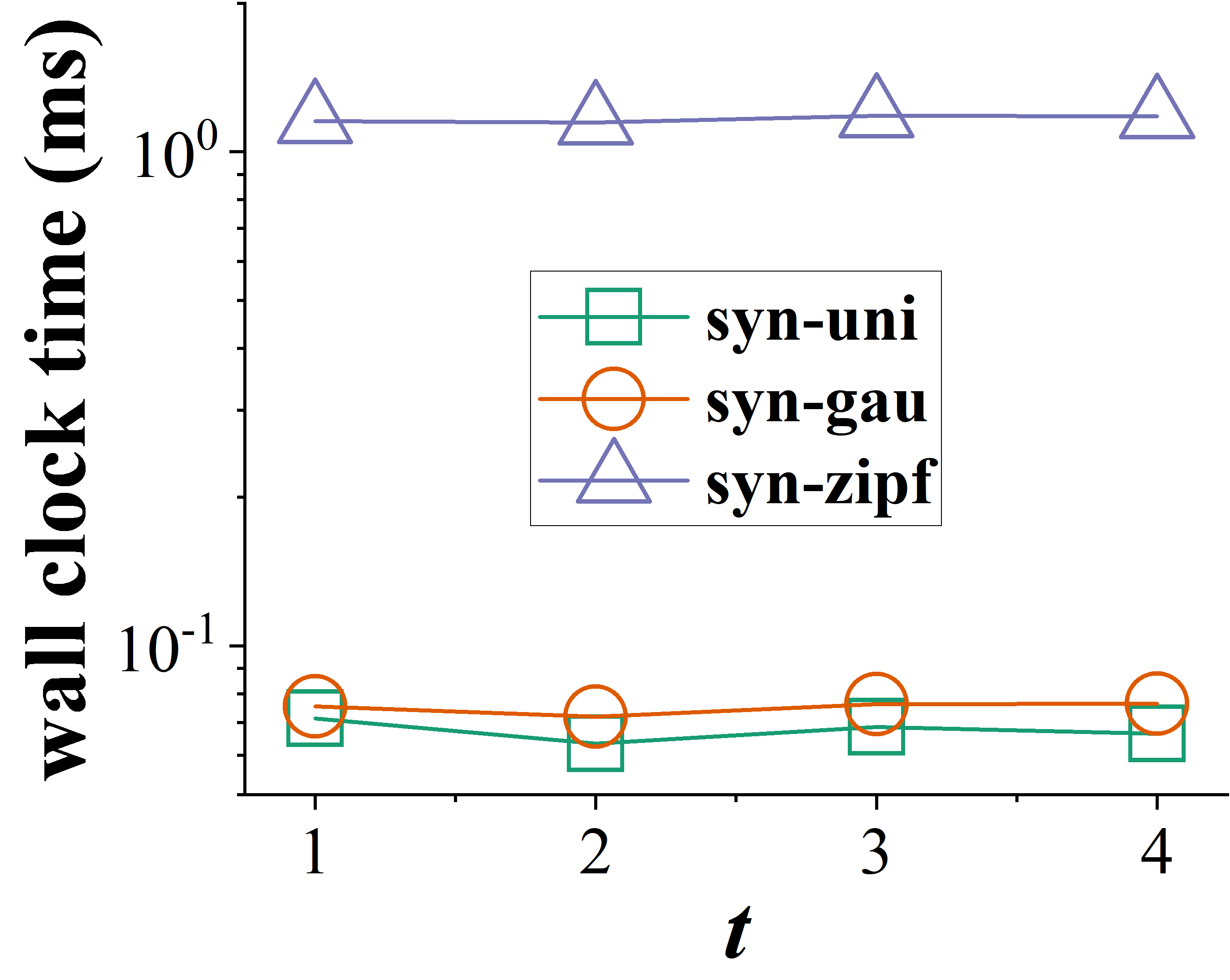}}} 
        \caption{\textsc{LIVE} efficiency w.r.t different parameters $\bm{d}$, $\bm{\alpha/\beta}$, and $\bm{t}$.}
        \label{fig:parameters}
    \end{minipage}\hfill
    \begin{minipage}[t]{0.4\textwidth}
    \centering
        \subfigure[{\small real-world graphs}]{\label{subfig:pruning_real}
            {\includegraphics[height=2.7cm]{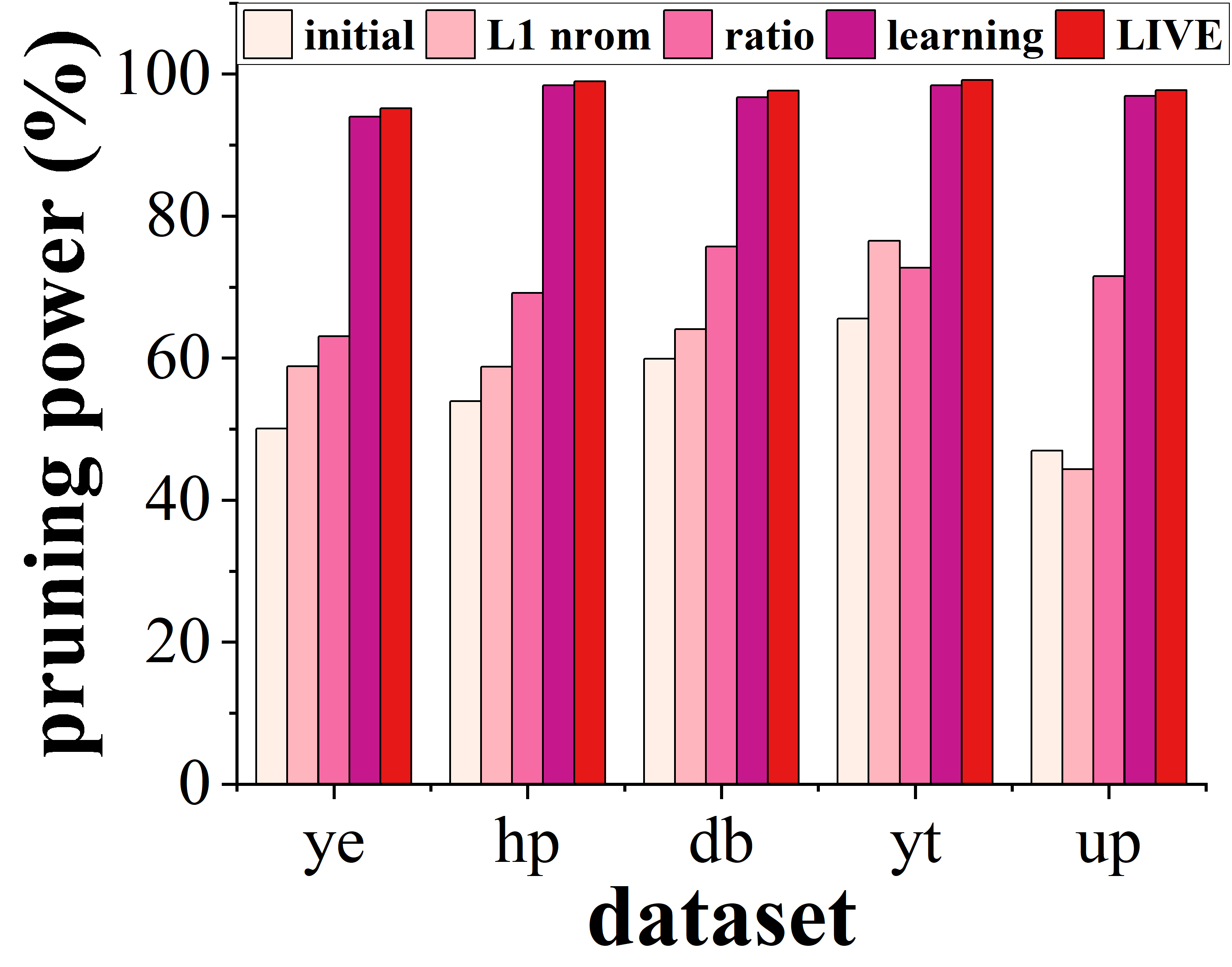}}} 
            \subfigure[{\small synthetic graphs}]{\label{subfig:pruning_syn}
            {\includegraphics[height=2.7cm]{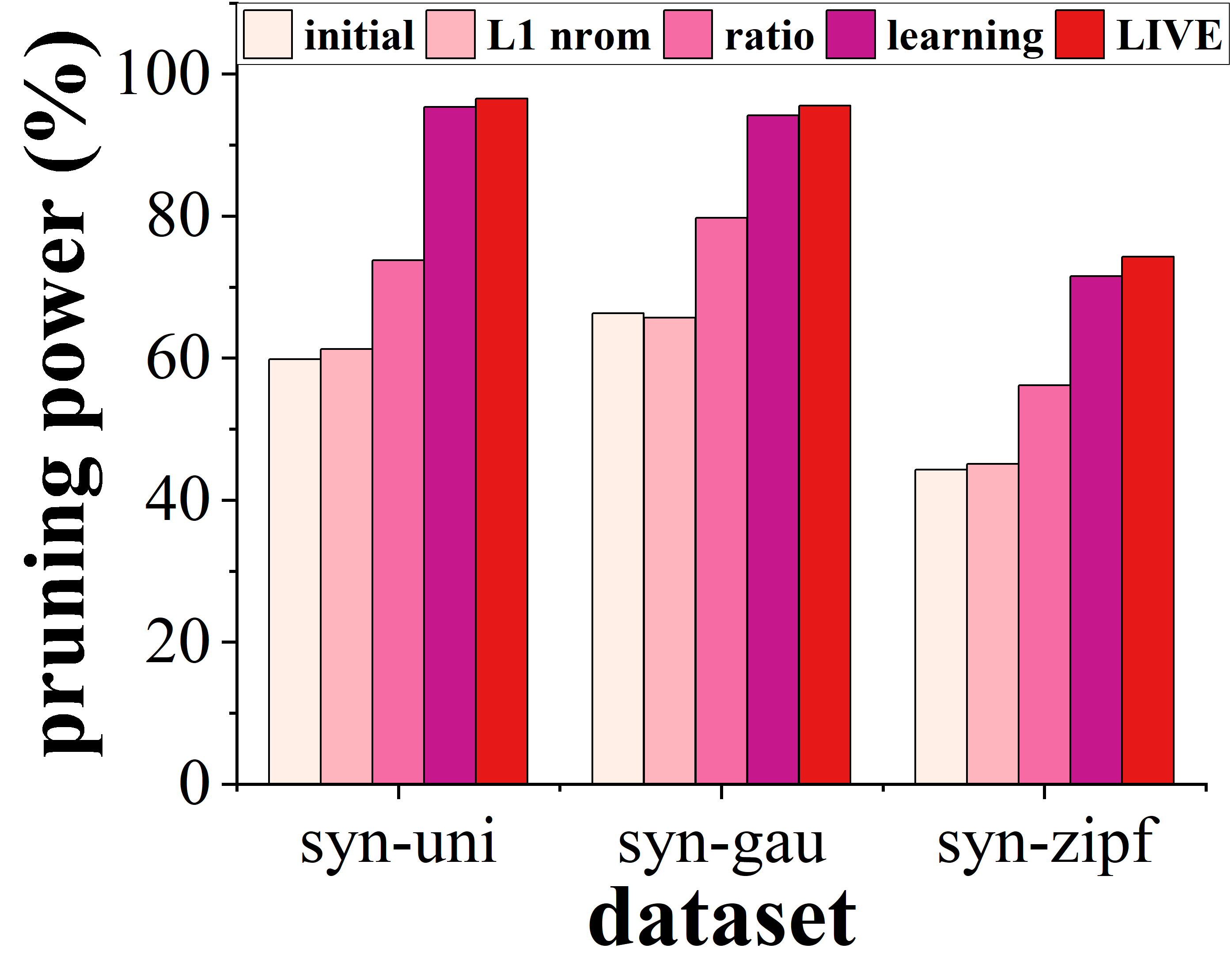}}} 
        \caption{\textsc{LIVE} pruning power.}
        \label{fig:pruning_power}
    \end{minipage}
\end{figure*}

\underline{\it Synthetic Graphs.}
Following \cite{sun2020memory,ye2024efficient,lu2025b}, we generated synthetic graphs using the Newman--Watts--Strogatz model~\cite{watts1998collective} via the NetworkX package~\cite{hagberg2020networkx}.
For each vertex $v_i$, its label $L(v_i)$ was randomly assigned from the range $[1, |\Sigma|]$ following one of three distributions: \kw{Uniform}, \kw{Gaussian}, or \kw{Zipf}.
Accordingly, we constructed three types of synthetic graphs, denoted as \textit{syn-uni}, \textit{syn-gau}, and \textit{syn-zipf}.
Detailed parameter settings are summarized in Table~\ref{tab:parameters}.

\underline{\it Real-world Graphs.}
We further evaluated \textsc{LIVE} on five real-world graph datasets that have been widely used in prior subgraph matching studies~\cite{he2008graphs, shang2008taming, zhao2010graph, sun2012efficient, han2013turboiso, ren2015exploiting, bi2016efficient, bhattarai2019ceci, han2019efficient, sun2020memory}.
The statistics of these datasets are reported in Table~\ref{tab:datasets}.

\noindent{\bf Query Graphs.}
For each data graph $G$, we generated 
100 connected query graphs following standard practices~\cite{sun2012efficient, han2013turboiso, ren2015exploiting, bi2016efficient, katsarou2017subgraph, archibald2019sequential, bhattarai2019ceci, han2019efficient}.
Each query graph $q$ was obtained by performing a random walk on $G$ until $|V(q)|$ vertices were collected.
If the induced subgraph exceeded the target average degree $avg\_deg(q)$, edges were randomly removed to match the desired density; otherwise, the sampling process was restarted.
Detailed query parameters are listed in Table~\ref{tab:parameters}.

\noindent{\bf Evaluation Metrics.}
Following \cite{sun2020memory,sun2020rapidmatch,ye2024efficient,lu2025b}, we evaluated both efficiency and effectiveness of \textsc{LIVE} and all baselines.
Efficiency was measured by the wall-clock query processing time, including both filtering and refinement phases.
Effectiveness was evaluated in terms of pruning power, defined as the percentage of data vertices eliminated during candidate retrieval by the proposed pruning strategies (Section~\ref{subsec:pruning}).

Unless otherwise stated, all reported results were averaged over 100 query graphs.
In addition, we reported the offline pre-computation costs of \textsc{LIVE}, including embedding model training time, \textit{iLabel} index construction time, and \textit{iLabel} index storage overhead, 
to demonstrate the practicality of the proposed framework.

Table~\ref{tab:parameters} summarizes the default parameter settings used in our experiments, where default values are highlighted in bold.
In subsequent experiments, we vary one parameter at a time while keeping all others fixed at their default values.

\subsection{Parameter Tuning}
In this subsection, we studied the impact of key parameters in \textsc{LIVE} on query efficiency using synthetic datasets, varying one parameter at a time while holding the others at their default values.

\noindent{\bf Impact of Vertex Embedding Dimension $\bm{d}$.}
Figure~\ref{subfig:parameter_dim} reports the query processing time of \textsc{LIVE} as the embedding dimension $d$ varied from 1 to 4.
As shown there, increasing $d$ from 1 to 2 significantly reduces query time due to stronger pruning from more expressive embeddings 
, while further increases yield marginal benefits as pruning saturated and per-comparison cost grew linearly.
This indicates that low-dimensional monotonic vertex embeddings (\eg~$d=2$) 
can support efficient subgraph matching on widely tested datasets.
Nevertheless, \textsc{LIVE} maintains consistently low query latency across different values of $d$ (\eg~below 0.07$ms$ on \textit{syn-uni} and \textit{syn-gau}, and below 1.19$ms$ on \textit{syn-zipf}), demonstrating robustness to the embedding dimensionality.

\noindent{\bf Impact of Weight Ratio $\bm{\alpha/\beta}$.}
Figure~\ref{subfig:parameter_ratio} evaluated the impact of the weight ratio $\alpha/\beta$ on query efficiency, 
where $\alpha$ and $\beta$ control the relative contributions of vertex label embeddings (VLE) and vertex structure embeddings (VSE), respectively.
As $\alpha/\beta$ increases from $10^3$ to $10^5$, query time decreases 
and then stabilizes when $\alpha/\beta$ exceeds $10^5$.
This behavior aligns with the design of \textit{iLabel}: 
(i) a larger $\alpha/\beta$ enhances separation among label-based clusters in the key space, resulting in tighter query ranges and fewer vertices examined during index traversal; and 
(ii) an effective $\alpha/\beta$ can be achieved with a relatively small value (e.g., $10^5$).
Overall, \textsc{LIVE} maintains low query latency across a wide range of $\alpha/\beta$ values (ranging from 0.06$ms$ to 1.16$ms$).

\noindent{\bf Impact of Synopsis Hop Parameter $\bm{t}$.}
Figure~\ref{subfig:parameter_hop} examined the effect of the hop parameter $t$ in the hop-based embedding synopsis, which controls the amount of multi-hop structural information used for pruning.
As shown in the figure, query time first decreases as $t$ increases from 1 to 2, and then 
increases for larger values of $t$.
This behavior is expected: while a larger $t$ enables stronger pruning by incorporating richer multi-hop information, the overhead of comparing higher-dimensional synopses eventually outweighs the additional pruning benefit.
The resulting U-shaped trend reflects the trade-off between pruning effectiveness and comparison cost.
Nevertheless, \textsc{LIVE} maintains low query latency across all settings, ranging from 0.06$ms$ to 1.18$ms$.

\noindent{\bf Default Parameter Selection.}
Based on the above observations, we set
$d=2$, $\alpha/\beta=100$K, and $t=2$ as the default parameters in the remaining experiments.
These values strike a good balance between pruning effectiveness and computational overhead, and provide stable performance across different datasets and query workloads.

\begin{figure}[t]
\centering
\subfigure[][{\small real-world graphs}]{                    
\scalebox{0.17}[0.15]{\includegraphics{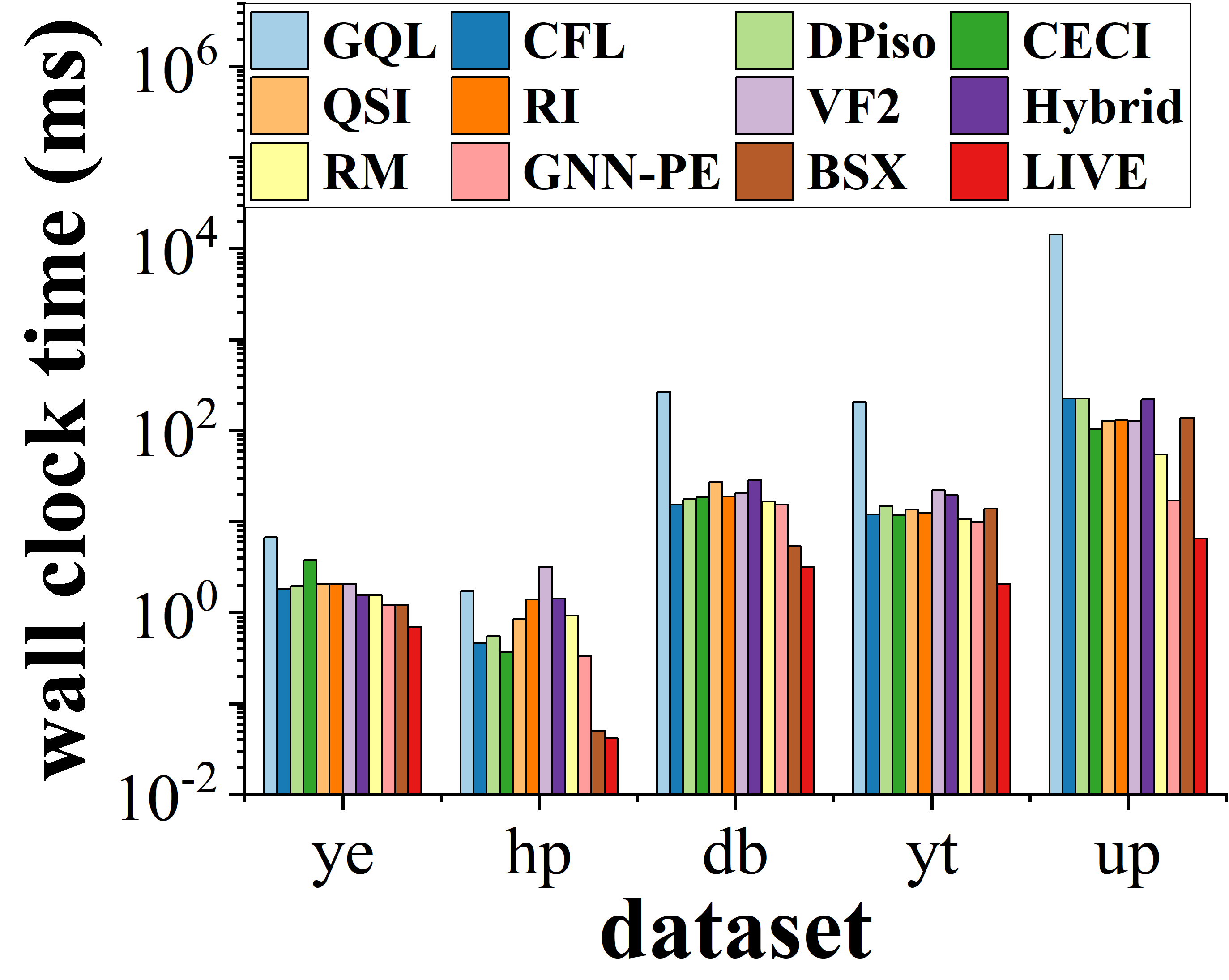}}\label{subfig:efficiency_real}}
\qquad
\subfigure[][{\small synthetic graphs}]{
\scalebox{0.15}[0.15]{\includegraphics{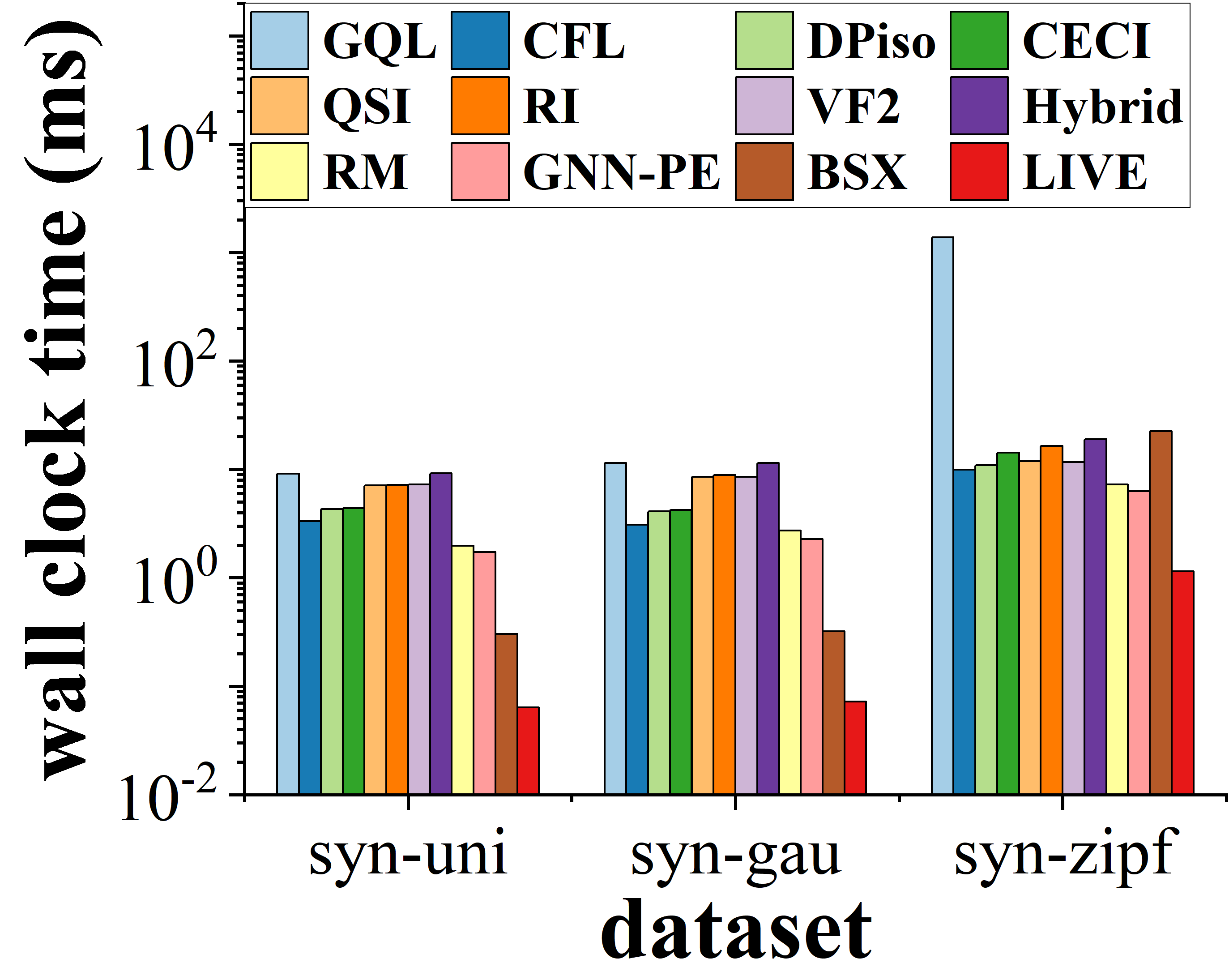}}\label{subfig:efficiency_syn}}
\caption{\textsc{LIVE} efficiency on synthetic/real-world graphs.} 
\label{fig:efficiency}
\end{figure}

\subsection{Evaluation of \textsc{LIVE} Effectiveness}
In this subsection, we evaluated the effectiveness of \textsc{LIVE} in terms of pruning power, measured as the percentage of data vertices 
safely eliminated during candidate retrieval.

\begin{figure*}[t]
\centering
\subfigure[][{\small varying $|\Sigma|$}]{                    
\label{subfig:data_label}
    {\includegraphics[height=2.8cm]{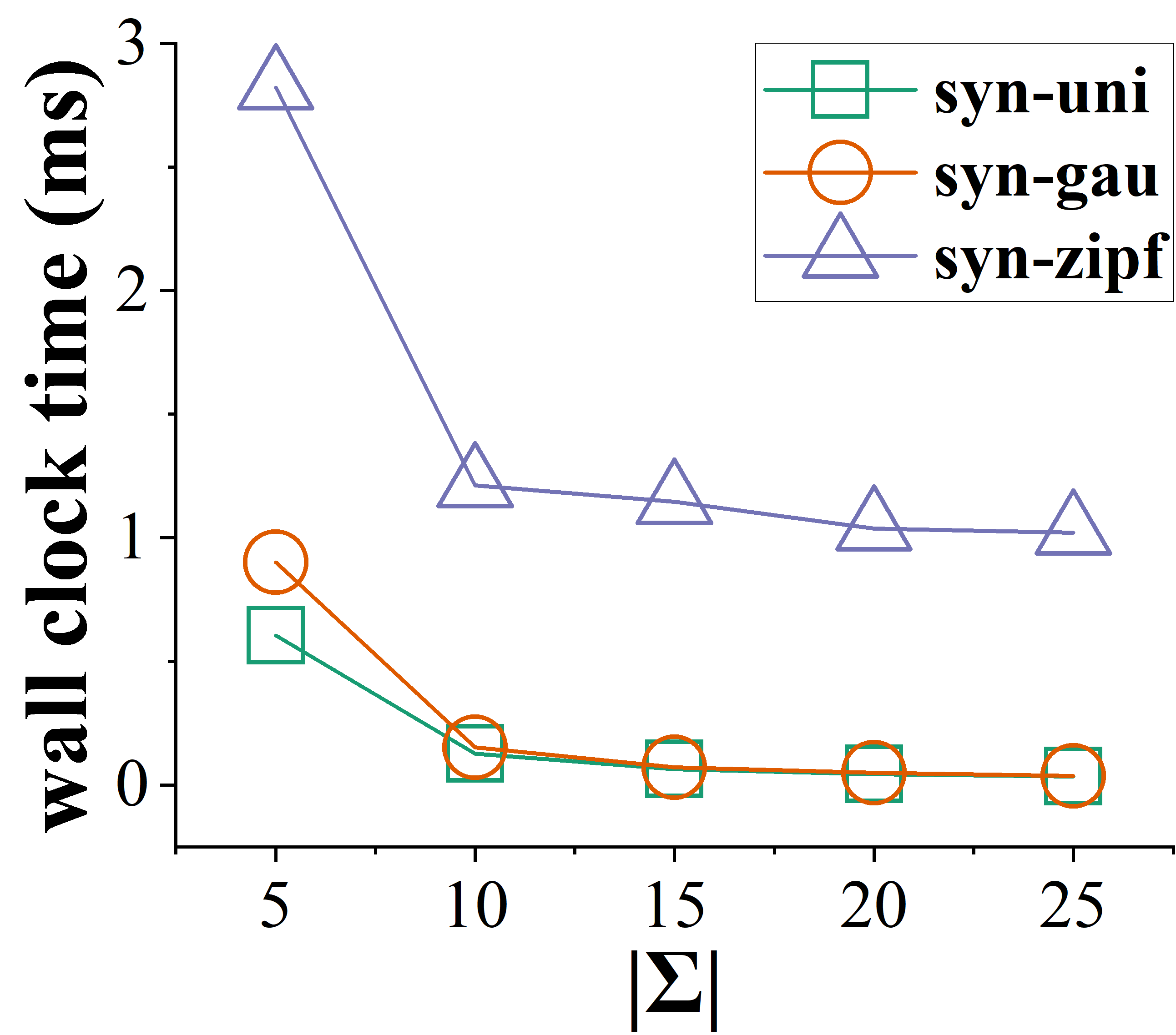}}}
\subfigure[][{\small varying $avg\_deg(q)$}]{                    
\label{subfig:query_degree}
    {\includegraphics[height=2.8cm]{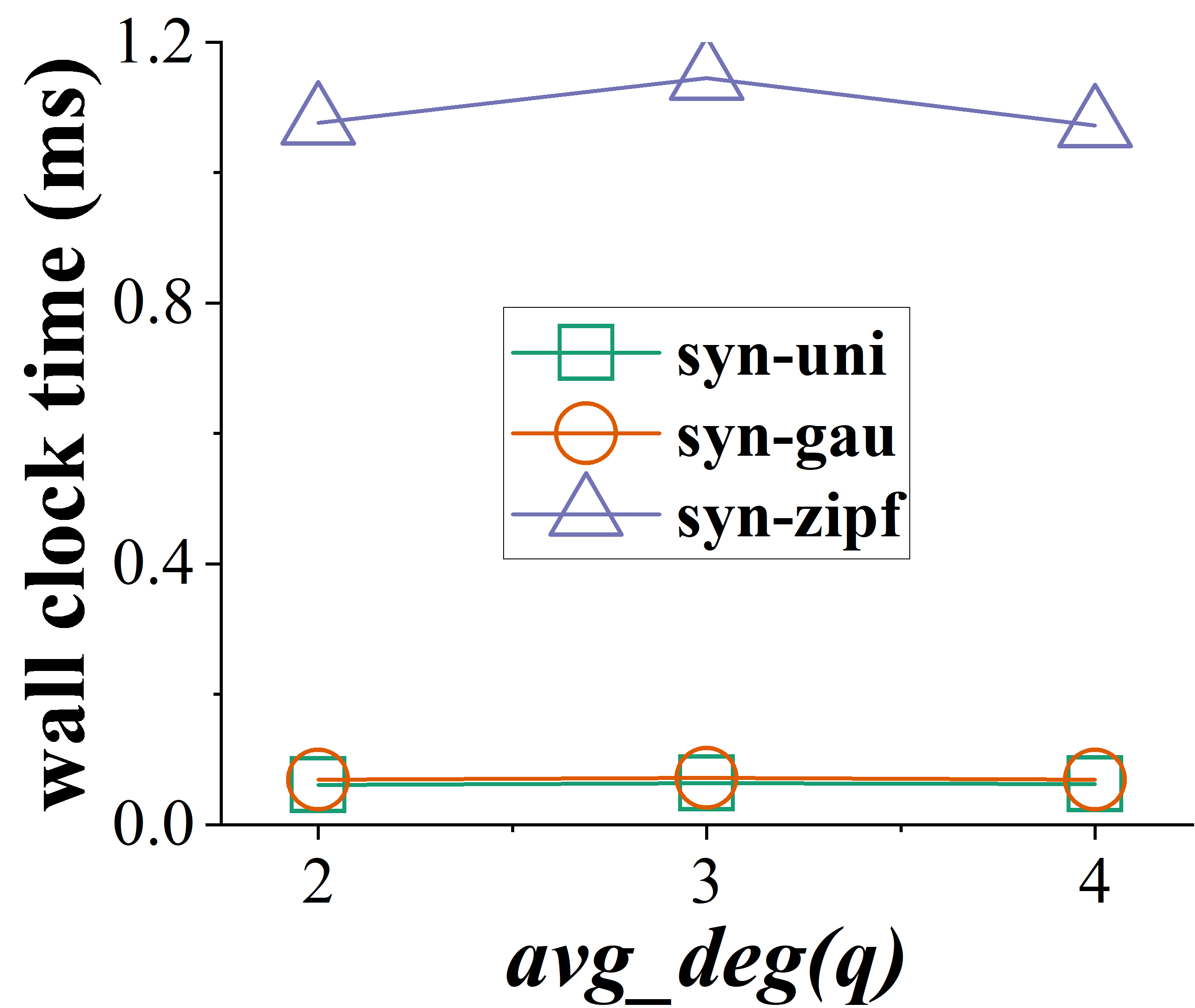}}}
\subfigure[][{\small varying $|V(q)|$}]{                    
\label{subfig:query_size}
    {\includegraphics[height=2.8cm]{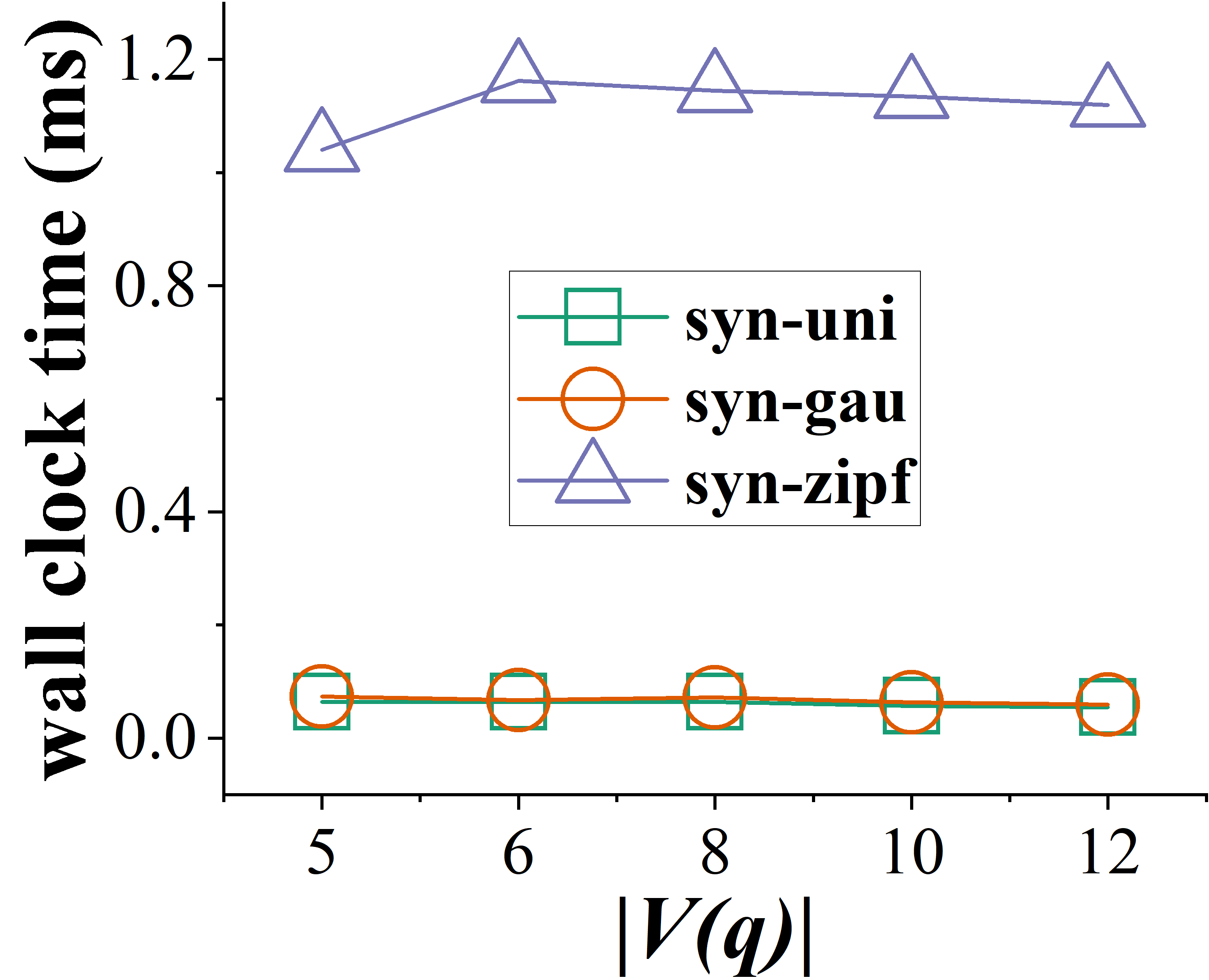}}}
\subfigure[][{\small varying $avg\_deg(G)$}]{                    
\label{subfig:data_degree}
    {\includegraphics[height=2.8cm]{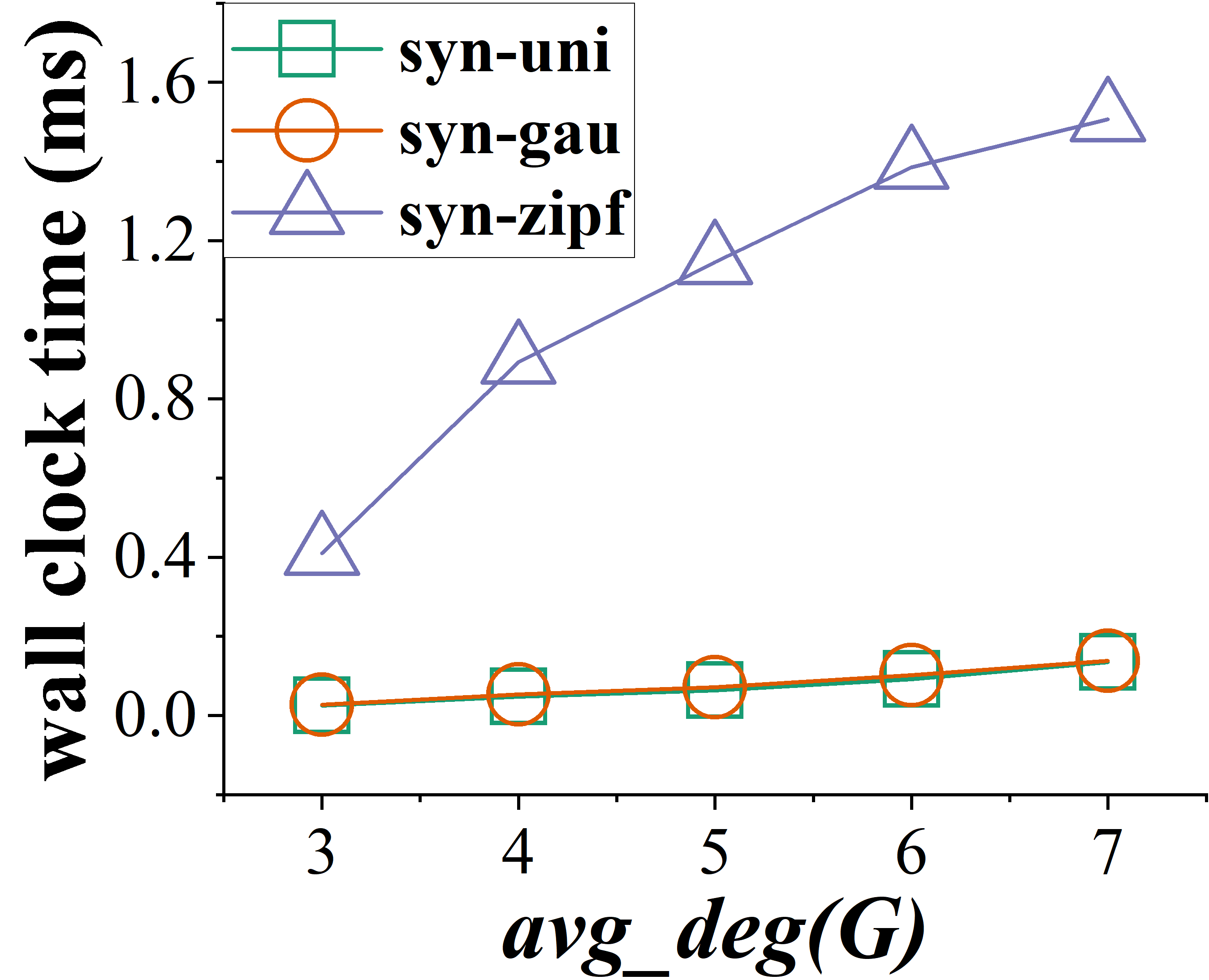}}}
\subfigure[][{\small varying $|V(G)|$}]{                    
\label{subfig:data_size}
    {\includegraphics[height=2.8cm]{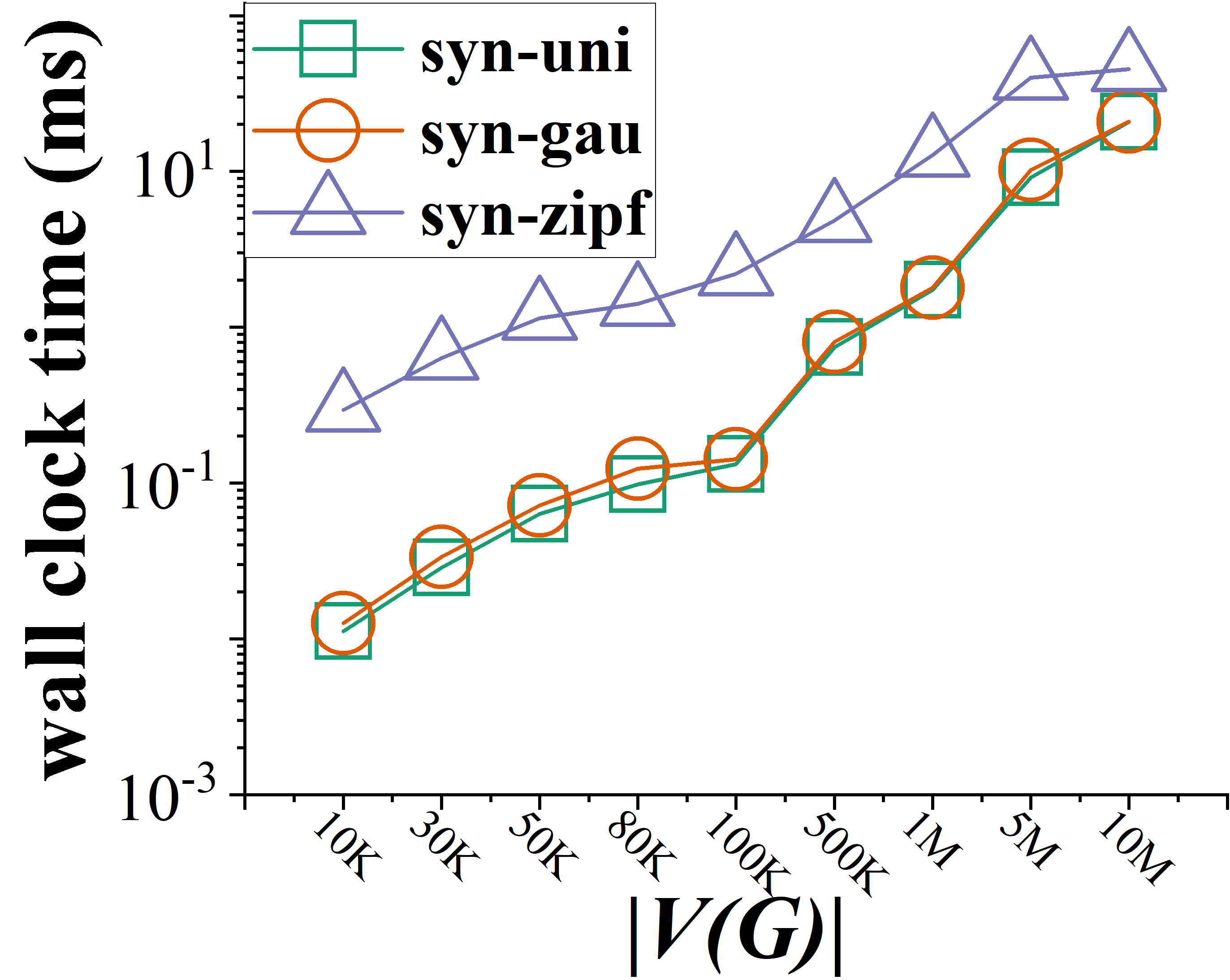}}}
\caption{\textsc{LIVE} efficiency evaluation on synthetic graphs.}
\label{fig:efficiency_parameter}
\end{figure*}

\noindent{\bf Pruning Power of \textsc{LIVE} on Real-world and Synthetic Graphs.}
Figure~\ref{fig:pruning_power} evaluated the pruning power of \textsc{LIVE} under different embedding configurations and optimizations on both real-world and synthetic graphs, with all parameters set to their default values.

(i) Initial embeddings (\textit{initial}). Vertex embeddings are randomly initialized. While monotonicity guarantees correctness, pruning power is limited (\eg~only up to 66.32\% of data vertices are eliminated).
(ii) Lightweight optimizations (\textit{L1 norm} and \textit{ratio}). These apply non-learning-based optimizations, including $L_1$ normalization of vertex label embeddings and tuning the weight ratio $\alpha/\beta$. Both reshape the embedding distribution: $L_1$ normalization reduces dominance region size, while adjusting $\alpha/\beta$ strengthens label-based separation in the \textit{iLabel} key space, yielding consistent but moderate gains (\eg~ up to 76.49\% and 79.73\% pruning).
(iii) Embedding learning (\textit{learning}). Optimizing embeddings with the proposed anti-dominance loss significantly increases dispersion in the embedding space, resulting in much stronger pruning (\eg~up to 98.43\%).
(iv) Full \textsc{LIVE} pruning (\textit{\textsc{LIVE}}). Combining embedding learning with all pruning strategies, \ie~key-based, dominance-based, hop-based, and degree-based pruning, further eliminates false positives and achieves high pruning power, up to 99.18\%, removing over 99\% of data vertices before refinement.


These results confirm the effectiveness of \textsc{LIVE}'s embedding design and pruning strategies.
More importantly, they highlight a clear separation between correctness guarantees and pruning effectiveness.
While monotonicity alone ensures no false dismissals
(even with randomly initialized embeddings),
strong pruning power is achieved only when vertex embeddings are explicitly optimized to reduce dominance.
This empirical evidence shows that dominance correctness by itself does not imply efficient filtering, and validates the core design of \textsc{LIVE}: decoupling correctness guarantees from pruning-oriented embedding optimization.

\subsection{Evaluation of \textsc{LIVE} Efficiency}
In this subsection, we evaluated the efficiency of \textsc{LIVE} in terms of wall-clock time, and compared it with state-of-the-art baseline methods on both real-world and synthetic graphs.
Unless otherwise stated, all parameters are set to their default values.

\noindent{\bf Efficiency of \textsc{LIVE} on Synthetic and Real-world Graphs.}
Figure~\ref{fig:efficiency} compared the time cost of \textsc{LIVE} with eleven baselines on both synthetic and real-world datasets.
As shown, \textsc{LIVE} consistently outperforms all baselines across all datasets, achieving up to an order-of-magnitude speedup over the strongest competitors on large-scale real-world graphs such as \textit{yt} and \textit{up}.
Notably, \textsc{LIVE} answers exact subgraph matching queries on the largest \textit{up}, which contains 3.77$M$ vertices, in only 6.53$ms$.
These results demonstrate the high efficiency of \textsc{LIVE} and its strong scalability to large graphs.

To further evaluate the efficiency of \textsc{LIVE}, 
we varied key parameters (e.g., $|\Sigma|$, $avg\_deg(q)$, $|V(q)|$, $avg\_deg(G)$, and $|V(G)|$) on synthetic graphs in subsequent tests.
To better illustrate performance trends, baseline results are omitted below.

\noindent{\bf Efficiency of \textsc{LIVE} w.r.t. $\bm{\#}$ of Distinct Vertex Labels, $\bm{|\Sigma|}$.}
Figure~\ref{subfig:data_label} shows the query time of \textsc{LIVE} as the number of distinct vertex labels $|\Sigma|$ increases from 5 to 25.
As $|\Sigma|$ increases, label-based clustering in the \textit{iLabel} index becomes more effective, leading to smaller query ranges and fewer candidate vertices, and thus lower query time.
Although different label distributions introduce some variation in pruning behavior, \textsc{LIVE} consistently maintains low query latency across all settings, ranging from 0.03$ms$ to 2.82$ms$.

\noindent{\bf Efficiency of \textsc{LIVE} w.r.t. Average Degree, $\bm{avg\_deg(q)}$, of the Query Graph $\bm{q}$.}
Figure~\ref{subfig:query_degree} evaluated the impact of the query graph's average degree on efficiency.
As $avg\_deg(q)$ increases, query time first rises and then declines. 
This is because although higher query density enables stronger pruning due to tighter matching requirements, it introduces additional overhead in computing hop-based synopses and enforcing structural constraints.
Nevertheless, query time remains low across all settings 
($0.06ms\sim 1.15ms$).

\noindent{\bf Efficiency of \textsc{LIVE} w.r.t. Query Graph Size $\bm{|V(q)|}$.}
Figure~\ref{subfig:query_size} reports query time as the query graph size increased from 5 to 12 vertices.
Intuitively, larger queries involve more vertices and can increase index traversal and refinement costs, but they also yield richer hop-based synopses that strengthen pruning and reduce candidate sets.
As a result, \textsc{LIVE} exhibits decreasing query time as $|V(q)|$ grows, indicating that stronger structural constraints effectively offset the additional processing overhead.
Across all query sizes, query time remains low, ranging from 0.05$ms$ to 1.16$ms$.

\noindent{\bf Efficiency of \textsc{LIVE}  w.r.t. Average Degree, $\bm{avg\_deg(G)}$, of the Data Graph $\bm{G}$.}
Figure~\ref{subfig:data_degree} evaluated \textsc{LIVE} under varying data graph densities, with $avg\_deg(G)$ from 3 to 7.
As shown in the figure, query time increases only moderately as graph density grows, even on graphs with skewed label distributions,
\eg~from 0.41$ms$ to 1.51$ms$ on \textit{syn-zipf}.
This trend arises because higher density weakens pruning effectiveness and enlarges candidate sets.
\textsc{LIVE} mitigates this effect through degree-based synopsis pruning, which restricts structural checks to substructures whose degrees match those of query vertices.
Overall, \textsc{LIVE} is efficient across all settings.

\noindent{\bf Scalability of \textsc{LIVE} w.r.t. Data Graph Size $\bm{|V(G)|}$.}
Figure~\ref{subfig:data_size} evaluated the scalability of \textsc{LIVE} as the data graph size increased from 10$K$ to 10$M$ vertices.
\textsc{LIVE} consistently achieves the lowest query time across all settings,
\eg~it answers exact subgraph matching queries within 45.26$ms$ on graphs with 10$M$ vertices.
This efficiency stems from the one-dimensional \textit{iLabel} index design combined with effective pruning strategies, which tightly bound candidate exploration even as graph size grows.
This indicates that \textsc{LIVE} scales well for exact subgraph matching on large-scale graphs.

\begin{figure}[t]
\subfigure[][{\small model training time}]{                    
\scalebox{0.11}[0.11]{\includegraphics{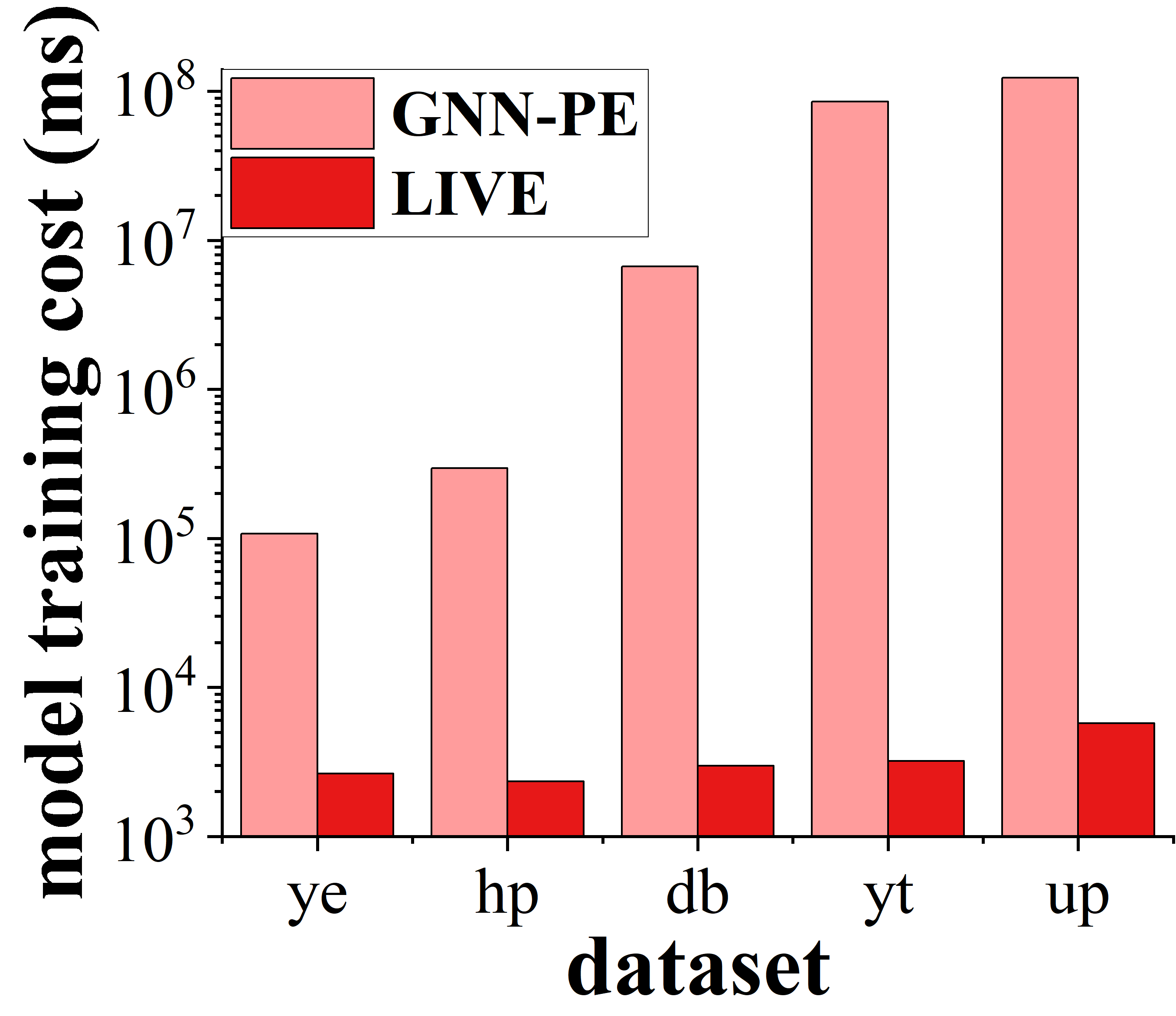}}\label{subfig:model_training_time}}
\subfigure[][{\small index building time}]{
\scalebox{0.11}[0.11]{\includegraphics{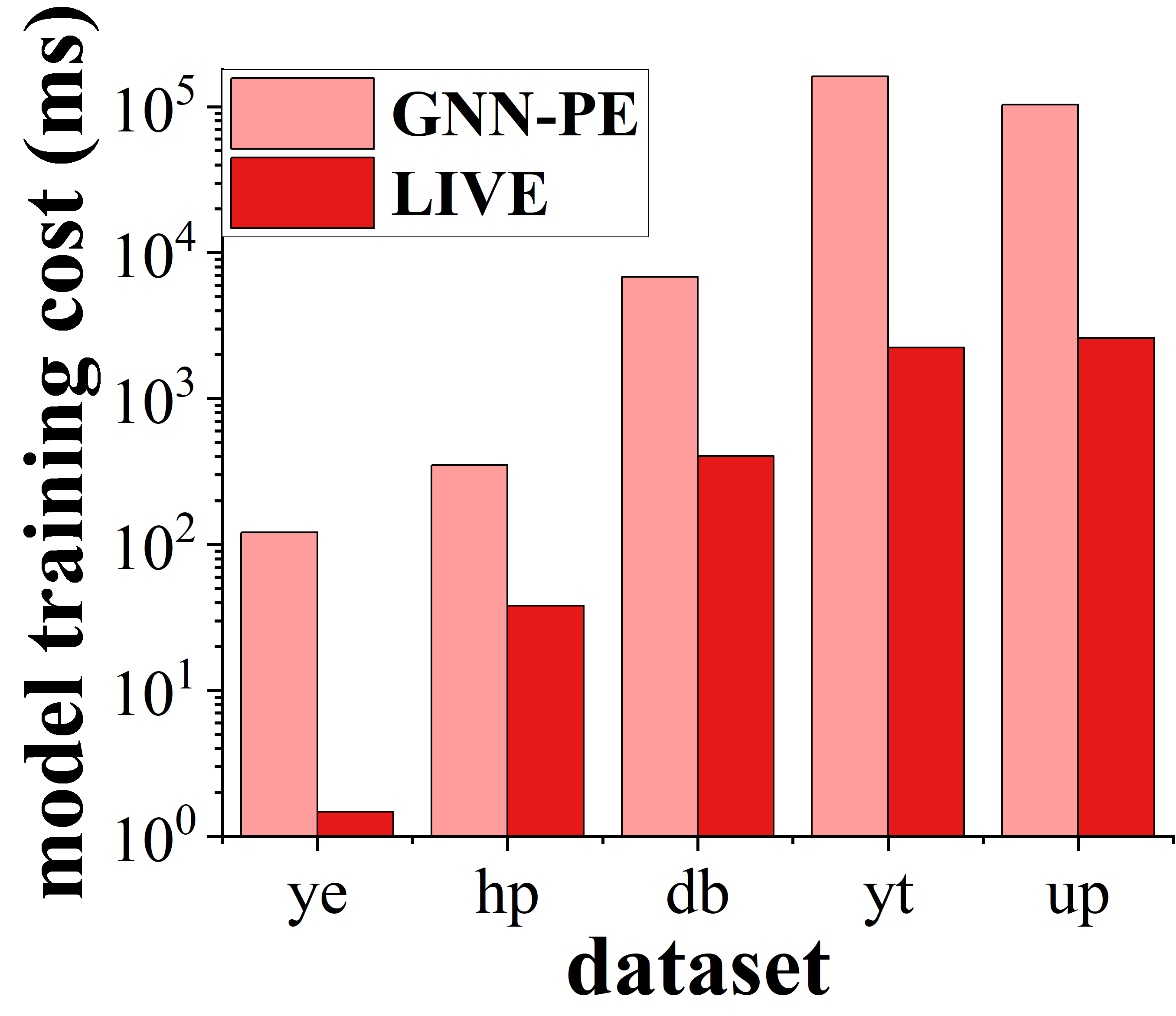}}\label{subfig:index_construction_time}}
\subfigure[][{\small index storage cost}]{
\scalebox{0.11}[0.11]{\includegraphics{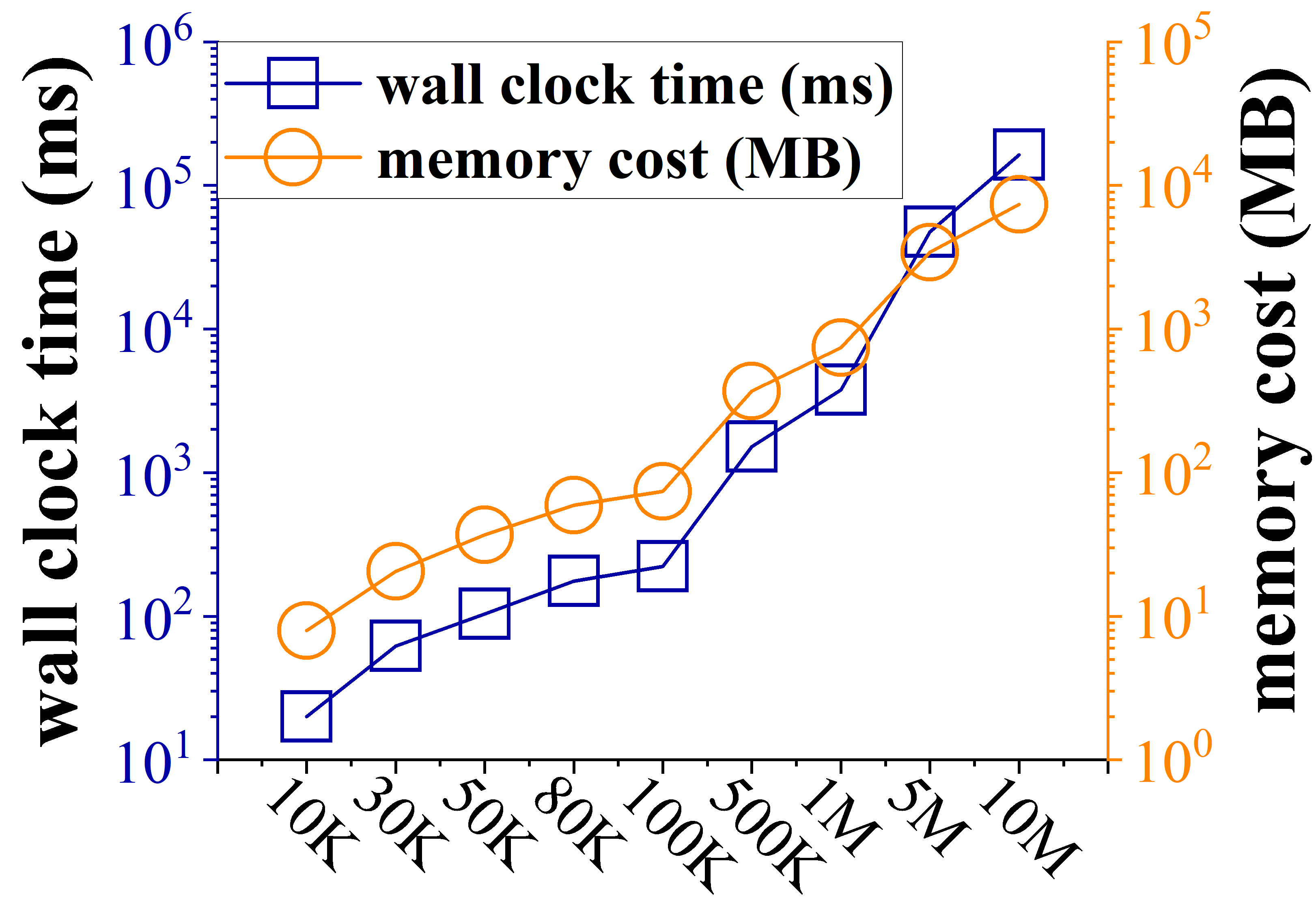}}\label{subfig:index_cost}}
\caption{The \textsc{LIVE} offline pre-computation cost.}
\label{fig:precomputation_cost}
\end{figure}

\subsection{Offline Pre-Computation Performance}
\label{subsec:offline_performance}
In this subsection, we evaluated the offline pre-computation cost of \textsc{LIVE}, including embedding model training and index construction, and compared it with the representative learning-based baseline GNN-PE.
Experiments were conducted on both real-world and synthetic graphs using default parameter settings.

\noindent{\bf Offline Pre-Computation Cost on Real-world Graphs.}
Figures~\ref{subfig:model_training_time} and~\ref{subfig:index_construction_time} compared the model training and index construction costs of \textsc{LIVE} and GNN-PE on real-world datasets.

As shown in Figure~\ref{subfig:model_training_time}, GNN-PE incurs 10$^4$$\times$ higher training cost, 
which increases sharply with graph size due to exhaustive enumeration of 1-hop subgraphs and their substructures for learning dominance relations.
In contrast, \textsc{LIVE} decouples dominance correctness from the learning objective and adopts a sampling-based anti-dominance loss with a fixed number of sampled vertex pairs, resulting in low and stable training cost across datasets, with only minor overhead from embedding computation.
Figure~\ref{subfig:index_construction_time} further shows that although index construction time increases with graph size for both methods, \textsc{LIVE} consistently achieves 82$\times$ lower cost by building a lightweight \textit{iLabel} index over vertex embeddings, whereas GNN-PE relies on path-based structures that substantially inflate index size and construction overhead.


These results demonstrate the superior cost performance of \textsc{LIVE} in both model training and index construction on large-scale graphs.
More importantly, they highlight a fundamental limitation of dominance learning-based approach (\eg~GNN-PE): by tightly coupling correctness guarantees with the learning objective, such methods inevitably incur prohibitively high training costs due to exhaustive subgraph–substructure enumeration.
In contrast, \textsc{LIVE} attains the same no-false-dismissal guarantee through monotonic embedding construction and decouples correctness from embedding optimization.
This structural decoupling enables lightweight, sampling-based training and avoids the offline cost bottleneck inherent in dominance learning-based frameworks.

\noindent{\bf Index Construction Time and Space Costs of \textsc{LIVE} w.r.t. Data Graph Size $|V(G)|$.}
Figure~\ref{subfig:index_cost} further evaluated the time and space overhead of constructing the \textit{iLabel} index on synthetic graphs as the data graph size $|V(G)|$ increases from 10$K$ to 10$M$.
Across this range, the total index construction time,
including vertex embedding generation, auxiliary synopsis computation, and B$^+$-tree construction,
grows from 19.98$ms$ to 163.59$sec$.
Correspondingly, index storage increases from 7.93$MB$ to 7{,}400.51$MB$.
These results indicate that the offline cost of \textsc{LIVE} scales gracefully with graph size and remains practical for large-scale graphs.

\section{Related Work}
\label{sec:related_work}

We classify existing works as follows.

\noindent{\bf Exact Subgraph Matching.}
Existing exact subgraph matching methods fall 
into two categories:
(i) join-based algorithms and (ii) backtracking-search-based algorithms.
Join-based methods~\cite{lai2015scalable,lai2016scalable,aberger2017emptyheaded,ammar2018distributed,mhedhbi2019optimizing,lai2019distributed,sun2020rapidmatch}
decompose a query graph into smaller substructures (\eg~edges, paths, or triangles), enumerate matches for each substructure independently, and join these partial results to obtain final matches.
Backtracking-based algorithms~\cite{shang2008taming,carletti2015vf2,bonnici2013subgraph,sun2012efficient,he2008graphs,bi2016efficient,bhattarai2019ceci,han2019efficient,carletti2017challenging}
first compute candidate sets for query vertices using filtering techniques, and then enumerate exact matches via depth-first search while enforcing injective mapping and edge-consistency constraints;
recent work further applies reinforcement learning to optimize matching orders in backtracking-based frameworks~\cite{wang2022reinforcement,yang2025neuso}.

\textsc{LIVE} differs from prior work in three key aspects:
(i) it scales to dense queries and large graphs via effective pruning with the \textit{iLabel} index, avoiding the large intermediates of join-based methods~\cite{lai2015scalable,lai2016scalable,aberger2017emptyheaded,ammar2018distributed,mhedhbi2019optimizing,lai2019distributed,sun2020rapidmatch};
(ii) it relies on learned vertex embeddings rather than explicit structural comparisons, improving scalability over backtracking-based algorithms~\cite{shang2008taming,carletti2015vf2,bonnici2013subgraph,sun2012efficient,he2008graphs,bi2016efficient,bhattarai2019ceci,han2019efficient,carletti2017challenging}; and
(iii) unlike feature-vector graph search targeting small graphs~\cite{shasha2002algorithmics,james2004daylight,yan2004graph}, \textsc{LIVE} enumerates all exact matches within a single large graph.

\nop{\textsc{LIVE} differs from prior approaches in the following.
(i) \textsc{LIVE} scales to dense queries and large graphs by leveraging effective pruning strategies supported by the proposed \textit{iLabel} index, rather than incurring large intermediate results and high memory over-\\head as in join-based methods~\cite{lai2015scalable,lai2016scalable,aberger2017emptyheaded,ammar2018distributed,mhedhbi2019optimizing,lai2019distributed,sun2020rapidmatch}.
(ii) \textsc{LIVE} performs subgraph matching using learned vertex embeddings, avoiding explicit structural comparisons that limit the scalability of traditional backtracking-based algorithms~\cite{shang2008taming,carletti2015vf2,bonnici2013subgraph,sun2012efficient,he2008graphs,bi2016efficient,bhattarai2019ceci,han2019efficient,carletti2017challenging}.
(iii) Although early studies explored feature-vector representations for graph search~\cite{shasha2002algorithmics,james2004daylight,yan2004graph}, they primarily targeted graph database retrieval over collections of small graphs.
In contrast, \textsc{LIVE} focuses on enumerating all exact matches within a single large graph, which introduces fundamentally different scalability challenges.}

\noindent{\bf Approximate Subgraph Matching.}
Another line of work studies approximate subgraph matching, which aims to efficiently retrieve subgraphs that are similar to a given query graph rather than exactly isomorphic.
To improve scalability, these methods relax exact structural constraints and employ heuristic or distance-based similarity measures, such as graph edit distance~\cite{li2018efficient}, neighborhood similarity~\cite{wen2025s3and}, or feature-based approximations~\cite{du2017first,dutta2017neighbor,guo2025efficient}.
As a result, approximate matching algorithms trade accuracy for efficiency and may return false positives or miss valid exact matches.

In contrast, the problem addressed in this work requires enumerating all subgraphs that are exactly isomorphic to the query graph, with no false positives or false dismissals.
Therefore, approximate subgraph matching methods do not provide the correctness guarantees required here and are orthogonal to the focus of \textsc{LIVE}.

\noindent{\bf Learning-based Subgraph Matching.}
Recent work explores learning techniques to accelerate subgraph matching by avoiding explicit graph-structure comparisons.
Early learning-based approaches focus on approximate matching, using GNNs or DNNs to embed graphs or subgraphs into embedding spaces and compare them via distance functions~\cite{mcfee2009partial,vendrov2015order,li2019graph,bai2019simgnn,xu2019cross,lou2020neural}.
More recently, learning has been brought to \emph{exact} subgraph matching. For instance, GNN-PE~\cite{ye2024efficient} is an early representative framework that learns dominance-preserving embeddings for graph paths, enabling candidate filtering without false dismissals; meanwhile, DIVINE~\cite{ye2025continuous} studies dynamic continuous matching using rule-based vertex-dominance embeddings generated by handcrafted rules rather than a learned model.

This work differs from prior studies in two key respects.
(a) \textsc{LIVE} provides correctness guarantees for exact subgraph isomorphism on large-scale \emph{static} graphs, whereas prior learning-based methods largely target approximate matching~\cite{mcfee2009partial,vendrov2015order,li2019graph,bai2019simgnn,xu2019cross,lou2020neural}.
(b) \textsc{LIVE} enforces dominance preservation by construction via monotonic vertex embeddings, decoupling correctness from learning; this yields lower training cost and stronger pruning than prior embedding-based approaches (\eg~\cite{ye2024efficient,ye2025continuous}) that rely on costly enumeration or heuristic designs.

\nop{\color{blue}
More recently, learning has been applied to \emph{exact} subgraph matching;
for example, GNN-PE~\cite{ye2024efficient} learns dominance-preserving embeddings for graph paths, enabling candidate filtering without false dismissals and representing the first learning-based framework with correctness guarantees.
As another embedding-based method, DIVINE \cite{ye2025continuous} studies dynamic continuous subgraph matching using vertex dominance embeddings generated by rules rather than a learning model.

This work differs from prior studies in two aspects.
(i) \textsc{LIVE} provides correctness guarantees for exact subgraph isomorphism on static large-scale graphs, unlike heuristic learning-based approaches~\cite{mcfee2009partial,vendrov2015order,li2019graph,bai2019simgnn,xu2019cross,lou2020neural}.
(ii) LIVE enforces dominance preservation by construction through monotonic vertex embedding design, decoupling correctness from learning objectives and achieving lower training cost and stronger pruning effectiveness than prior learning-based or embedding-based methods (\eg~\cite{ye2024efficient} and \cite{ye2025continuous}), which rely on costly enumeration of subgraph--substructure pairs or heuristic embedding rules.
}

\nop{
for example, GNN-PE~\cite{ye2024efficient} learns dominance-preserving embeddings for graph paths, enabling candidate filtering without false dismissals and representing the first learning-based framework with correctness guarantees.

This work differs from prior studies in two aspects.
(i) \textsc{LIVE} provides correctness guarantees for exact subgraph isomorphism, unlike heuristic learning-based approaches~\cite{mcfee2009partial,vendrov2015order,li2019graph,bai2019simgnn,xu2019cross,lou2020neural}.
(ii) \textsc{LIVE} enforces dominance preservation by construction through monotonic vertex embedding design, decoupling correctness from learning objectives and achieving lower training cost and stronger pruning effectiveness than prior learning-based methods (\eg~\cite{ye2024efficient}), which rely on costly enumeration of subgraph--substructure pairs.
}

\balance
\section{Conclusion}
\label{sec:conclusion}
In this paper, we propose \textsc{LIVE}, a learning-based framework for exact subgraph matching on large graphs that achieves both scalability and exactness. 
By enforcing monotonicity in vertex embedding construction, \textsc{LIVE} guarantees dominance correctness by design and fundamentally decouples no-false-dismissal guarantees from embedding optimization.
This decoupling enables candidate filtering without false dismissals while allowing embedding learning to directly focus on improving vertex-level pruning power.
We further introduce a query cost model with a differentiable surrogate objective to guide offline embedding training, and design a lightweight one-dimensional \textit{iLabel} index that preserves dominance relationships and supports efficient candidate retrieval without complex indexing structures. Extensive experiments on synthetic and real-world datasets validate the performance of \textsc{LIVE}.

\clearpage
\balance
\bibliographystyle{ACM-Reference-Format}
\bibliography{sample}

@String{Computer = "{IEEE} Computer" }

@article{lewis1983michael,
  title={Computers and intractability: A guide to the theory of NP-completeness},
  author={Michael R. Garey and David S. Johnson},
  journal={The Journal of Symbolic Logic},
  volume={48},
  number={2},
  pages={498--500},
  year={1983}
}

@inproceedings{borzsony2001skyline,
  title={The skyline operator},
  author={Borzsony, Stephan and Kossmann, Donald and Stocker, Konrad},
  booktitle={Proceedings of the International Conference on Data Engineering (ICDE)},
  pages={421--430},
  year={2001}
}

@article{cordella2004sub,
  title={A (sub) graph isomorphism algorithm for matching large graphs},
  author={Cordella, Luigi P and Foggia, Pasquale and Sansone, Carlo and Vento, Mario},
  journal={IEEE Transactions on Pattern Analysis and Machine Intelligence},
  volume={26},
  number={10},
  pages={1367--1372},
  year={2004},
}

@article{grohe2020graph,
  title={The graph isomorphism problem},
  author={Grohe, Martin and Schweitzer, Pascal},
  journal={Communications of the ACM},
  volume={63},
  number={11},
  pages={128--134},
  year={2020}
}

@inproceedings{he2008graphs,
  title={Graphs-at-a-time: query language and access methods for graph databases},
  author={He, Huahai and Singh, Ambuj K},
  booktitle={Proceedings of the International Conference on Management of Data (SIGMOD)},
  pages={405--418},
  year={2008}
}

@inproceedings{shang2008taming,
  title={Taming verification hardness: an efficient algorithm for testing subgraph isomorphism},
  author={Shang, Haichuan and Zhang, Ying and Lin, Xuemin and Yu, Jeffrey Xu},
  booktitle={Proceedings of the International Conference on Very Large Data Bases (PVLDB)},
  pages={364--375},
  year={2008},
}

@article{bonnici2013subgraph,
  title={A subgraph isomorphism algorithm and its application to biochemical data},
  author={Bonnici, Vincenzo and Giugno, Rosalba and Pulvirenti, Alfredo and Shasha, Dennis and Ferro, Alfredo},
  journal={BMC Bioinformatics},
  volume={14},
  number={7},
  pages={1--13},
  year={2013}
}

@inproceedings{bi2016efficient,
  title={Efficient subgraph matching by postponing cartesian products},
  author={Bi, Fei and Chang, Lijun and Lin, Xuemin and Qin, Lu and Zhang, Wenjie},
  booktitle={Proceedings of the International Conference on Management of Data (SIGMOD)},
  pages={1199--1214},
  year={2016}
}

@article{juttner2018vf2++,
  title={VF2++—An improved subgraph isomorphism algorithm},
  author={J{\"u}ttner, Alp{\'a}r and Madarasi, P{\'e}ter},
  journal={Discrete Applied Mathematics},
  volume={242},
  pages={69--81},
  year={2018}
}

@inproceedings{han2019efficient,
  title={Efficient subgraph matching: Harmonizing dynamic programming, adaptive matching order, and failing set together},
  author={Han, Myoungji and Kim, Hyunjoon and Gu, Geonmo and Park, Kunsoo and Han, Wook-Shin},
  booktitle={Proceedings of the International Conference on Management of Data (SIGMOD)},
  pages={1429--1446},
  year={2019}
}

@inproceedings{bhattarai2019ceci,
  title={Ceci: Compact embedding cluster index for scalable subgraph matching},
  author={Bhattarai, Bibek and Liu, Hang and Huang, H Howie},
  booktitle={Proceedings of the International Conference on Management of Data (SIGMOD)},
  pages={1447--1462},
  year={2019}
}

@inproceedings{sun2020memory,
  title={In-memory subgraph matching: An in-depth study},
  author={Sun, Shixuan and Luo, Qiong},
  booktitle={Proceedings of the International Conference on Management of Data (SIGMOD)},
  pages={1083--1098},
  year={2020}
}

@inproceedings{kankanamge2017graphflow,
  title={Graphflow: An active graph database},
  author={Kankanamge, Chathura and Sahu, Siddhartha and Mhedbhi, Amine and Chen, Jeremy and Salihoglu, Semih},
  booktitle={Proceedings of the International Conference on Management of Data (SIGMOD)},
  pages={1695--1698},
  year={2017}
}

@article{wasserman1994social,
  title={Social network analysis: Methods and applications},
  author={Wasserman, Stanley and Faust, Katherine},
  year={1994},
  publisher={Cambridge University Press}
}

@article{karlebach2008modelling,
  title={Modelling and analysis of gene regulatory networks},
  author={Karlebach, Guy and Shamir, Ron},
  journal={Nature Reviews Molecular Cell Biology},
  volume={9},
  number={10},
  pages={770--780},
  year={2008}
}

@article{szklarczyk2015string,
  title={STRING v10: protein--protein interaction networks, integrated over the tree of life},
  author={Szklarczyk, Damian and  others},
  journal={Nucleic Acids Research},
  volume={43},
  number={D1},
  pages={D447--D452},
  year={2015}
}

@inproceedings{ammar2018distributed,
  title={Distributed Evaluation of Subgraph Queries Using Worst-case Optimal Low-Memory Dataflows},
  author={Ammar, Khaled and McSherry, Frank and Salihoglu, Semih and Joglekar, Manas},
  booktitle={Proceedings of the International Conference on Very Large Data Bases (PVLDB)},
  pages={691--704},
  year={2018}
}

@inproceedings{katsarou2017subgraph,
  title={Subgraph querying with parallel use of query rewritings and alternative algorithms},
  author={Katsarou, Foteini and Ntarmos, Nikos and Triantafillou, Peter},
booktitle={Proceedings of the International Conference on Extending Database Technology (EDBT)},
pages={25--36},
  year={2017}
}

@inproceedings{ren2015exploiting,
  title={Exploiting vertex relationships in speeding up subgraph isomorphism over large graphs},
  author={Ren, Xuguang and Wang, Junhu},
  booktitle={Proceedings of the International Conference on Very Large Data Bases (PVLDB)},
  pages={617--628},
  year={2015}
}

@inproceedings{zhao2010graph,
  title={On graph query optimization in large networks},
  author={Zhao, Peixiang and Han, Jiawei},
  booktitle={Proceedings of the International Conference on Very Large Data Bases (PVLDB)},
  pages={340--351},
  year={2010}
}

@inproceedings{sun2012efficient,
  title={Efficient Subgraph Matching on Billion Node Graphs},
  author={Sun, Zhao and Wang, Hongzhi and Wang, Haixun and Shao, Bin and Li, Jianzhong},
  booktitle={Proceedings of the International Conference on Very Large Data Bases (PVLDB)},
  pages={788--799},
  year={2012}
}

@inproceedings{han2013turboiso,
  title={Turboiso: towards ultrafast and robust subgraph isomorphism search in large graph databases},
  author={Han, Wook-Shin and Lee, Jinsoo and Lee, Jeong-Hoon},
  booktitle={Proceedings of the International Conference on Management of Data (SIGMOD)},
  pages={337--348},
  year={2013}
}

@article{hagberg2020networkx,
  title={Networkx: Network analysis with python},
  author={Hagberg, Aric and Conway, Drew},
  journal={URL: https://networkx. github. io},
  year={2020}
}

@article{watts1998collective,
  title={Collective dynamics of ‘small-world’networks},
  author={Watts, Duncan J and Strogatz, Steven H},
  journal={Nature},
  volume={393},
  number={6684},
  pages={440--442},
  year={1998}
}

@inproceedings{archibald2019sequential,
  title={Sequential and parallel solution-biased search for subgraph algorithms},
  author={Archibald, Blair and Dunlop, Fraser and Hoffmann, Ruth and McCreesh, Ciaran and Prosser, Patrick and Trimble, James},
  booktitle={Proceedings of the  Integration of Constraint Programming, Artificial Intelligence, and Operations Research (CPAIOR)},
  pages={20--38},
  year={2019}
}

@inproceedings{li2019graph,
  title={Graph matching networks for learning the similarity of graph structured objects},
  author={Li, Yujia and Gu, Chenjie and Dullien, Thomas and Vinyals, Oriol and Kohli, Pushmeet},
  booktitle={Proceedings of the International Conference on Machine Learning (ICML)},
  pages={3835--3845},
  year={2019},
}

@inproceedings{bai2019simgnn,
  title={Simgnn: A neural network approach to fast graph similarity computation},
  author={Bai, Yunsheng and Ding, Hao and Bian, Song and Chen, Ting and Sun, Yizhou and Wang, Wei},
  booktitle={Proceedings of the International Conference on Web Search and Data Mining (WSDM)},
  pages={384--392},
  year={2019}
}

@inproceedings{ye2024efficient,
  title={Efficient Exact Subgraph Matching via GNN-based Path
Dominance Embedding},
  author={Ye, Yutong and Lian, Xiang and Chen, Mingsong},
  booktitle={Proceedings of the International Conference on Very Large Data Bases (PVLDB)},
  pages={1628--1641},
  year={2024}
}

@article{alon2007network,
  title={Network motifs: theory and experimental approaches},
  author={Alon, Uri},
  journal={Nature Reviews Genetics},
  volume={8},
  number={6},
  pages={450--461},
  year={2007},
}

@inproceedings{qiao2017subgraph,
  title={Subgraph matching: on compression and computation},
  author={Qiao, Miao and Zhang, Hao and Cheng, Hong},
  booktitle={Proceedings of the International Conference on Very Large Data Bases (PVLDB)},
  pages={176--188},
  year={2017},
}

@inproceedings{sahu2017ubiquity,
  title={The ubiquity of large graphs and surprising challenges of graph processing},
  author={Sahu, Siddhartha and Mhedhbi, Amine and Salihoglu, Semih and Lin, Jimmy and {\"O}zsu, M Tamer},
  booktitle={Proceedings of the International Conference on Very Large Data Bases (PVLDB)},
  pages={420--431},
  year={2017},
}

@inproceedings{lian2011efficient,
  title={Efficient query answering in probabilistic RDF graphs},
  author={Lian, Xiang and Chen, Lei},
  booktitle={Proceedings of the International Conference on Management of Data (SIGMOD)},
  pages={157--168},
  year={2011}
}

@inproceedings{babai2018group,
  title={Group, graphs, algorithms: the graph isomorphism problem},
  author={Babai, L{\'a}szl{\'o}},
  booktitle={Proceedings of the International Congress of Mathematicians: Rio de Janeiro 2018},
  pages={3319--3336},
  year={2018},
  organization={World Scientific}
}

@inproceedings{al2020topic,
  title={Topic-based community search over spatial-social networks},
  author={Al-Baghdadi, Ahmed and Lian, Xiang},
  booktitle={Proceedings of the International Conference on Very Large Data Bases (PVLDB)},
  pages={2104--2117},
  year={2020},
}

@inproceedings{zhang2024top,
  title={Top-$ L $ Most Influential Community Detection Over Social Networks},
  author={Zhang, Nan and Ye, Yutong and Lian, Xiang and Chen, Mingsong},
  booktitle={Proceedings of the International Conference on Data Engineering (ICDE)},
  pages={5767--5779},
  year={2024},
}

@inproceedings{yan2008mining,
  title={Mining significant graph patterns by leap search},
  author={Yan, Xifeng and Cheng, Hong and Han, Jiawei and Yu, Philip S},
  booktitle={Proceedings of the International Conference on Management of Data (SIGMOD)},
  pages={433--444},
  year={2008}
}

@inproceedings{deutsch2022graph,
  title={Graph pattern matching in GQL and SQL/PGQ},
  author={Deutsch, Alin and Francis, Nadime and Green, Alastair and Hare, Keith and Li, Bei and Libkin, Leonid and Lindaaker, Tobias and Marsault, Victor and Martens, Wim and Michels, Jan and others},
  booktitle={Proceedings of the International Conference on Management of Data (SIGMOD)},
  pages={2246--2258},
  year={2022}
}

@inproceedings{ye2025continuous,
  title={Continuous Subgraph Matching via Cost-Model-based Dynamic Vertex Dominance Embeddings},
  author={Ye, Yutong and Lian, Xiang and Zhang, Nan and Chen, Mingsong},
  booktitle={Proceedings of the International Conference on Management of Data (SIGMOD)},
  pages={1--27},
  year={2025},
}

@inproceedings{sun2020rapidmatch,
  title={Rapidmatch: A holistic approach to subgraph query processing},
  author={Sun, Shixuan and Sun, Xibo and Che, Yulin and Luo, Qiong and He, Bingsheng},
  booktitle={Proceedings of the International Conference on Very Large Data Bases (PVLDB)},
  pages={176--188},
  year={2020}
}

@inproceedings{lu2025b,
  title={BSX: Subgraph Matching with Batch Backtracking Search},
  author={Lu, Yujie and Zhang, Zhijie and Zheng, Weiguo},
  booktitle={Proceedings of the International Conference on Management of Data (SIGMOD)},
  pages={1--27},
  year={2025},
}

@inproceedings{lai2019distributed,
  title={Distributed subgraph matching on timely dataflow},
  author={Lai, Longbin and Qing, Zhu and Yang, Zhengyi and Jin, Xin and Lai, Zhengmin and Wang, Ran and Hao, Kongzhang and Lin, Xuemin and Qin, Lu and Zhang, Wenjie and others},
  booktitle={Proceedings of the International Conference on Very Large Data Bases (PVLDB)},
  pages={1099--1112},
  year={2019},
}

@inproceedings{mhedhbi2019optimizing,
  title={Optimizing subgraph queries by combining binary and worst-case optimal joins},
  author={Mhedhbi, Amine and Salihoglu, Semih},
  booktitle={Proceedings of the International Conference on Very Large Data Bases (PVLDB)},
pages={1692--1704},
  year={2019}
}

@inproceedings{lai2015scalable,
  title={Scalable subgraph enumeration in mapreduce},
  author={Lai, Longbin and Qin, Lu and Lin, Xuemin and Chang, Lijun},
  booktitle={Proceedings of the International Conference on Very Large Data Bases (PVLDB)},
  pages={974--985},
  year={2015},
}

@inproceedings{lai2016scalable,
  title={Scalable distributed subgraph enumeration},
  author={Lai, Longbin and Qin, Lu and Lin, Xuemin and Zhang, Ying and Chang, Lijun and Yang, Shiyu},
  booktitle={Proceedings of the International Conference on Very Large Data Bases (PVLDB)},
  pages={217--228},
  year={2016},
}

@article{aberger2017emptyheaded,
  title={Emptyheaded: A relational engine for graph processing},
  author={Aberger, Christopher R and Lamb, Andrew and Tu, Susan and N{\"o}tzli, Andres and Olukotun, Kunle and R{\'e}, Christopher},
  journal={ACM Transactions on Database Systems},
  volume={42},
  number={4},
  pages={1--44},
  year={2017},
}

@inproceedings{carletti2015vf2,
  title={VF2 Plus: An improved version of VF2 for biological graphs},
  author={Carletti, Vincenzo and Foggia, Pasquale and Vento, Mario},
  booktitle={International Workshop on Graph-Based Representations in Pattern Recognition (GbRPR)},
  pages={168--177},
  year={2015},
}

@article{carletti2017challenging,
  title={Challenging the time complexity of exact subgraph isomorphism for huge and dense graphs with VF3},
  author={Carletti, Vincenzo and Foggia, Pasquale and Saggese, Alessia and Vento, Mario},
  journal={IEEE Transactions on Pattern Analysis and Machine Intelligence},
  volume={40},
  number={4},
  pages={804--818},
  year={2017},
}

@article{lou2020neural,
  title={Neural subgraph matching},
  author={Lou, Zhaoyu and You, Jiaxuan and Wen, Chengtao and Canedo, Arquimedes and Leskovec, Jure and others},
  journal={arXiv preprint arXiv:2007.03092},
  year={2020}
}

@inproceedings{shasha2002algorithmics,
  title={Algorithmics and applications of tree and graph searching},
  author={Shasha, Dennis and Wang, Jason TL and Giugno, Rosalba},
  booktitle={Proceedings of the Principles of Database Systems (PODS)},
  pages={39--52},
  year={2002}
}

@article{james2004daylight,
  title={Daylight theory manual},
  author={James, Craig A},
  journal={http://www. daylight. com/dayhtml/doc/theory/theory. toc. html},
  year={2004}
}

@inproceedings{yan2004graph,
  title={Graph indexing: a frequent structure-based approach},
  author={Yan, Xifeng and Yu, Philip S and Han, Jiawei},
  booktitle={Proceedings of the International Conference on Management of Data (SIGMOD)},
  pages={335--346},
  year={2004}
}

@inproceedings{du2017first,
  title={First: Fast interactive attributed subgraph matching},
  author={Du, Boxin and Zhang, Si and Cao, Nan and Tong, Hanghang},
  booktitle={Proceedings of the International Conference on Knowledge Discovery and Data Mining (SIGKDD)},
  pages={1447--1456},
  year={2017}
}

@inproceedings{dutta2017neighbor,
  title={Neighbor-aware search for approximate labeled graph matching using the chi-square statistics},
  author={Dutta, Sourav and Nayek, Pratik and Bhattacharya, Arnab},
  booktitle={Proceedings of the Web Conference (WWW)},
  pages={1281--1290},
  year={2017}
}

@inproceedings{li2018efficient,
  title={An efficient probabilistic approach for graph similarity search},
  author={Li, Zijian and Jian, Xun and Lian, Xiang and Chen, Lei},
  booktitle={Proceedings of the International Conference on Data Engineering (ICDE)},
  pages={533--544},
  year={2018},
}

@article{xu2019cross,
  title={Cross-lingual knowledge graph alignment via graph matching neural network},
  author={Xu, Kun and Wang, Liwei and Yu, Mo and Feng, Yansong and Song, Yan and Wang, Zhiguo and Yu, Dong},
  journal={arXiv preprint arXiv:1905.11605},
  year={2019}
}

@inproceedings{mcfee2009partial,
  title={Partial order embedding with multiple kernels},
  author={McFee, Brian and Lanckriet, Gert},
  booktitle={Proceedings of the International Conference on Machine Learning (ICML)},
  pages={721--728},
  year={2009}
}

@inproceedings{vendrov2015order,
  title={Order-embeddings of images and language},
  author={Vendrov, Ivan and Kiros, Ryan and Fidler, Sanja and Urtasun, Raquel},
  booktitle={Proceedings of the International Conference on Learning Representations (ICLR)},
  pages={1--12},
  year={2016}
}

@inproceedings{wang2022reinforcement,
  title={Reinforcement learning based query vertex ordering model for subgraph matching},
  author={Wang, Hanchen and Zhang, Ying and Qin, Lu and Wang, Wei and Zhang, Wenjie and Lin, Xuemin},
  booktitle={Proceedings of the International Conference on Data Engineering (ICDE)},
  pages={245--258},
  year={2022},
}

@inproceedings{yang2025neuso,
  title={NeuSO: Neural Optimizer for Subgraph Queries},
  author={Yang, Linglin and Zou, Lei and Zhao, Chunshan},
  booktitle={Proceedings of the International Conference on Management of Data (SIGMOD)},
  pages={1--28},
  year={2025},
}

@article{appendix,
  title={Appendix},
  author={anonymous authors},
  journal={https://anonymous.4open.science/r/LIVE-4C37},
  year={2026}
}

@inproceedings{paszke2019pytorch,
  title={Pytorch: An imperative style, high-performance deep learning library},
  author={Paszke, Adam and Gross, Sam and Massa, Francisco and Lerer, Adam and Bradbury, James and Chanan, Gregory and Killeen, Trevor and Lin, Zeming and Gimelshein, Natalia and Antiga, Luca and others},
  booktitle={Proceedings of the Advances in Neural Information Processing Systems (NeurIPS)},
  pages={1--12},
  year={2019}
}

@inproceedings{qi2017pointnet,
  title={PointNet: Deep Learning on Point Sets for 3D Classification and Segmentation},
  author={Qi, Charles R and Su, Hao and Mo, Kaichun and Guibas, Leonidas J},
  booktitle={Proceedings of the IEEE Conference on Computer Vision and Pattern Recognition (CVPR)},
  pages={652--660},
  year={2017}
}

@inproceedings{zaheer2017deep,
  title={Deep Sets},
  author={Zaheer, Manzil and Kottur, Satwik and Ravanbakhsh, Siamak and Poczos, Barnabas and Salakhutdinov, Ruslan and Smola, Alexander},
  booktitle={Advances in Neural Information Processing Systems (NeurIPS)},
  volume={30},
  year={2017}
}

@inproceedings{wen2025s3and,
  title={S3AND: Efficient Subgraph Similarity Search Under Aggregated Neighbor Difference Semantics},
  author={Wen, Qi and Ye, Yutong and Lian, Xiang and Chen, Mingsong},
  booktitle={Proceedings of the International Conference on Very Large Data Bases (PVLDB)},
  pages={3708--3720},
  year={2025}
}

@inproceedings{li2025subgraph,
  title={Subgraph Matching: A New Decomposition Based Approach},
  author={Li, Qiyan and Yu, Jeffrey Xu and He, Zongyan},
  booktitle={Proceedings of the International Conference on Very Large Data Bases (PVLDB)},
  pages={4282--4294},
  year={2025},
}

@inproceedings{guo2025efficient,
  title={Efficient and Accurate Subgraph Counting: A Bottom-up Flow-learning Based Approach},
  author={Guo, Qiuyu and Yang, Jianye and Zhang, Wenjie and Wang, Hanchen and Zhang, Ying and Lin, Xuemin},
  booktitle={Proceedings of the International Conference on Very Large Data Bases (PVLDB)},
  pages={2695--2708},
  year={2025},
}

@inproceedings{jiang2025comprehensive,
  title={A Comprehensive Survey of Subgraph Matching:[Experiments \& Analysis]},
  author={Jiang, Haolin and Pandey, Santosh and Liu, Hang},
  booktitle={Proceedings of the International Conference on Management of Data (SIGMOD)},
  pages={1--30},
  year={2025},
}

\end{document}